\documentclass[10pt]{article}

\usepackage{amsthm}
\usepackage{graphicx} 
\usepackage{array} 

\usepackage{amsmath, amssymb, amsfonts, verbatim}
\usepackage{hyphenat,epsfig,subcaption,multirow}
\usepackage{nicefrac}
\usepackage{paralist}

\usepackage[font=small, labelfont=bf]{caption}

\usepackage[usenames,dvipsnames]{xcolor}
\usepackage[ruled]{algorithm2e}

\DeclareFontFamily{U}{mathx}{\hyphenchar\font45}
\DeclareFontShape{U}{mathx}{m}{n}{
      <5> <6> <7> <8> <9> <10>
      <10.95> <12> <14.4> <17.28> <20.74> <24.88>
      mathx10
      }{}
\DeclareSymbolFont{mathx}{U}{mathx}{m}{n}
\DeclareMathSymbol{\bigtimes}{1}{mathx}{"91}

\usepackage{tcolorbox}
\tcbuselibrary{skins,breakable}
\tcbset{enhanced jigsaw}

\usepackage[normalem]{ulem}
\usepackage[compact]{titlesec}

\definecolor{DarkRed}{rgb}{0.5,0.1,0.1}
\definecolor{DarkBlue}{rgb}{0.1,0.1,0.5}

\usepackage{nameref}
\definecolor{ForestGreen}{rgb}{0.1333,0.5451,0.1333}
\definecolor{Red}{rgb}{0.9,0,0}
\usepackage[linktocpage=true,
	pagebackref=true,colorlinks,
		urlcolor=black,
	linkcolor=DarkRed,citecolor=ForestGreen,
	bookmarks,bookmarksopen,bookmarksnumbered]
	{hyperref}
\usepackage[noabbrev,nameinlink]{cleveref}
\crefname{property}{property}{Property}
\creflabelformat{property}{(#1)#2#3}
\crefname{equation}{eq}{Eq}
\creflabelformat{equation}{(#1)#2#3}

\usepackage{bm}
\usepackage{url}
\usepackage{xspace}
\usepackage[mathscr]{euscript}

\usepackage{tikz}
\usetikzlibrary{arrows}
\usetikzlibrary{arrows.meta}
\usetikzlibrary{shapes}
\usetikzlibrary{backgrounds}
\usetikzlibrary{positioning}
\usetikzlibrary{decorations.markings}
\usetikzlibrary{patterns}
\usetikzlibrary{calc}
\usetikzlibrary{fit}
\usetikzlibrary{snakes}

\usepackage{mdframed}

\usepackage[noend]{algpseudocode}
\makeatletter
\def\BState{\State\hskip-\ALG@thistlm}
\makeatother

\usepackage{cite}
\usepackage{enumitem}

\usepackage[margin=1in]{geometry}

\newtheorem{theorem}{Theorem}
\newtheorem{lemma}{Lemma}[section]
\newtheorem{property}{Property}[section]

\newtheorem{proposition}[lemma]{Proposition}

\newtheorem{claim}[lemma]{Claim}

\newtheorem{definition}[lemma]{Definition}

\newtheorem{problem}{Problem}

\newtheorem*{claim*}{Claim}
\newtheorem*{proposition*}{Proposition}
\newtheorem*{lemma*}{Lemma}
\newtheorem*{problem*}{Problem}
\newtheorem*{property*}{Property}

\crefname{lemma}{Lemma}{Lemmas}
\crefname{claim}{Claim}{Claims}
\crefname{property}{Property}{Properties}

\newtheorem{mdresult}{Result}

\newtheoremstyle{restate}{}{}{\itshape}{}{\bfseries}{~(restated).}{.5em}{\thmnote{#3}}
\theoremstyle{restate}

\theoremstyle{definition}
\newtheorem{mdalg}{Algorithm}
\newenvironment{Algorithm}{\begin{tbox}\begin{mdalg}}{\end{mdalg}\end{tbox}}

\allowdisplaybreaks

\renewcommand{\qed}{\nobreak \ifvmode \relax \else
      \ifdim\lastskip<1.5em \hskip-\lastskip
      \hskip1.5em plus0em minus0.5em \fi \nobreak
      \vrule height0.75em width0.5em depth0.25em\fi}

\newcommand{\Qed}[1]{\ensuremath{\qed_{\textnormal{~#1}}}}

\newcommand{\Ot}{\ensuremath{\widetilde{O}}}
\newcommand{\eps}{\ensuremath{\varepsilon}}
\newcommand{\Paren}[1]{\Big(#1\Big)}

\newcommand{\bracket}[1]{\left[#1\right]}
\newcommand{\paren}[1]{\ensuremath{\left(#1\right)}\xspace}
\newcommand{\card}[1]{\left\vert{#1}\right\vert}
\newcommand{\Omgt}{\ensuremath{\widetilde{\Omega}}}

\newcommand{\expect}[1]{\Exp\bracket{#1}}

\newcommand{\set}[1]{\ensuremath{\left\{ #1 \right\}}}
\newcommand{\poly}{\mbox{\rm poly}}
\newcommand{\polylog}{\mbox{\rm  polylog}}

\newcommand{\OPT}{\ensuremath{\mbox{\sc opt}}\xspace}

\newcommand{\ALG}{\ensuremath{\mbox{\sc alg}}\xspace}

\DeclareMathOperator*{\Exp}{\ensuremath{{\mathbb{E}}}}
\DeclareMathOperator*{\Prob}{\ensuremath{\textnormal{Pr}}}
\renewcommand{\Pr}{\Prob}

\newcommand{\Ex}{\Exp}

\newenvironment{tbox}{\begin{tcolorbox}[
		enlarge top by=5pt,
		enlarge bottom by=5pt,
		 breakable,
		 boxsep=0pt,
                  left=4pt,
                  right=4pt,
                  top=10pt,
                  arc=0pt,
                  boxrule=1pt,toprule=1pt,
                  colback=white
                  ]
	}
{\end{tcolorbox}}

\newcommand{\II}{\ensuremath{\mathbb{I}}}

\newcommand{\mireal}[1][]{
  \ifx\relax#1\relax%
    \II(\mione \,; \mitwo)%
  \else%
    \II(\mione \,; \mitwo\mid #1)%
  \fi
}

\renewcommand{\SS}{\ensuremath{\mathcal{S}}}

\newcommand{\degp}[1]{\textnormal{deg}^{\!+}\!(#1)}

\title{Sublinear Time and Space Algorithms for Correlation Clustering via Sparse-Dense Decompositions}
 \author{
	Sepehr Assadi\footnote{(\href{mailto:sepehr.assadi@rutgers.edu}{sepehr.assadi@rutgers.edu) Department of Computer Science, Rutgers University.  Research supported in part by a NSF CAREER Grant CCF-2047061, and a gift from Google Research.}} \and 
Chen Wang\footnote{(\href{mailto:vihan.shah98@rutgers.edu}{\text{wc497@cs.rutgers.edu}}) Department of Computer Science, Rutgers University.  Research supported in part by a NSF CAREER Grant CCF-2047061, and a gift from Google Research. }  
}

\date{}

\begin{document}
\maketitle

\pagenumbering{roman}


\begin{abstract}

We present a new approach for solving (minimum disagreement) correlation clustering that results in sublinear algorithms with highly efficient time and space complexity for this problem. In particular, we obtain the following
algorithms for $n$-vertex $(+/-)$-labeled graphs $G$: 

\begin{itemize}
	\item A {sublinear-time} algorithm that with high probability returns a constant approximation clustering of $G$ in  $O(n\log^2{n})$ time assuming  access to the adjacency list of the $(+)$-labeled edges of $G$ (this is almost quadratically faster than even reading the input once). Previously, no sublinear-time algorithm was known for this problem with any multiplicative approximation guarantee. 
	 
	 \item A semi-streaming algorithm that with high probability returns a constant approximation clustering of $G$ in $O(n\log{n})$ space and a single pass over the edges of the graph $G$ (this memory is almost quadratically smaller 
	 than input size). Previously, no single-pass algorithm with $o(n^2)$ space was known for this problem with any approximation guarantee. 
\end{itemize}

	The main ingredient of our approach is a novel connection to \textbf{sparse-dense graph decompositions} that are used extensively in the graph coloring literature. 
	To our knowledge, this connection is the first application of these decompositions beyond graph coloring, and in particular for the correlation clustering problem, and can be of independent interest.  
\end{abstract}

\bigskip

\clearpage

\setcounter{tocdepth}{3}
\tableofcontents

\clearpage

\pagenumbering{arabic}
\setcounter{page}{1}

\section{Introduction}
\label{sec:intro}

Correlation clustering is an extensively studied problem in theoretical computer science and machine learning. In this problem, we are given a complete undirected graph $G=(V,E)$ with edges labeled by $(+)$ or $(-)$. 
 The general goal is to cluster the vertices in a way that $(+)$ edges appear more inside the clusters and $(-)$ edges appear more outside. Correlation clustering has found its applications in various areas, including image segmentation \cite{KimYNK14}, document clustering \cite{bansal2004correlation},  community detection \cite{ShiDELM21}, cross-lingual link detection \cite{GaelZ07}, matrix decomposition \cite{Aszalos21}, among others \cite{ChenSX12,BonchiGU13}. 

One of the most popular optimization objectives for correlation clustering is disagreement minimization, wherein the goal is to minimize the total number of $(+)$ edges that cross different clusters and $(-)$ edges that are inside the same clusters. We study
this disagreement minimization variant of correlation clustering in this paper. The problem is known to be both NP-hard and APX-hard, and there is a classical polynomial-time algorithm that achieves $2.06$-approximation \cite{bansal2004correlation}. Since then, disagreement minimization have been explored under various contexts, including the semi-random model \cite{Chierichetti2021}, fair clustering \cite{AhmadianE0M20}, quantum approximation \cite{weggemans2021solving}, and local clustering \cite{BonchiGK13,JafarovKMM21}, among others.

Nevertheless, for applications to modern massive datasets, even the efficiency of the polynomial-time approximation algorithms become insufficient. In particular, for a modern massive graph, even simple tasks like storing and processing all the edges \emph{once} becomes challenging. Therefore, there is a quest for obtaining \emph{sublinear} algorithms for correlation clustering. In such algorithms, the resource costs is usually asymptotically smaller than the input size, which allows correlation clustering to scale up to massive datasets. 

Two of the most canonical examples of sublinear algorithms are \emph{sublinear-time} algorithms and \emph{(sublinear-space) streaming} algorithms. The former model assumes the data is provided to the algorithm in a specific format, say, the adjacency list
of the input graph, and one can query each entry of the input in $O(1)$ time; the goal is then to solve the problem faster than even reading the entire input once. The latter model instead focuses on space of the 
algorithms by assuming the data is presented to the algorithm in a stream and the goal is to process this stream in a space much smaller than the input size. In light of the above discussion, we study the following fundamental question in this paper: 
\begin{quote}
\emph{Can we design sublinear time and/or space algorithms for correlation clustering?}
\end{quote}

This question and similar variants have already been pursued extensively in the literature. For sublinear-time algorithms,~\cite{BonchiGK13,Garcia-SorianoK20} designed algorithms that given access to the adjacency matrix of $G$, 
run in $O(n/\eps)$ time and output a $3$-multiplicative plus $(\eps \cdot n^2)$-additive approximation to correlation clustering. Moreover, impossibility results by~\cite{BonchiGK13,Bressan19} prove that these algorithms are effectively 
optimal in a sense that one needs\footnote{Throughout, we use $\Omgt(f(n))$ and $\Ot(f(n))$ to suppress dependence on $\poly\log{(n)}$ factors.} $\Omgt(n^2)$ additive error whenever working with $\Ot(n)$ time algorithms in the adjacency matrix access model. These results however leave open the possibility of other natural access models to the input such as adjacency list access, employed extensively both in theory and practice. 

For sublinear-space algorithms, the `sweet spot' for correlation clustering is considered \emph{semi-streaming} algorithms~\cite{FeigenbaumKMSZ05} that have space complexity $\Ot(n)$ which is proportional to the answer itself~\cite{ChierichettiDK14,AhnCGMW21,CohenAddadLMNP21}. The first semi-streaming algorithm for this problem is due to~\cite{ChierichettiDK14} and obtains $(3+\eps)$-approximation in $O(\frac{\log^{2}(n)}{\eps})$ passes. 
This algorithm was improved by~\cite{AhnCGMW21} to $3$-approximation in $O(\log\log{n})$ passes. Most recently,~\cite{CohenAddadLMNP21} presented a novel algorithm
with $O(1)$-approximation in $O(1)$ passes\footnote{While the constant in number of passes in~\cite{CohenAddadLMNP21} is not stated explicitly by the authors, it appears to be $6$ passes.}. These results however come short of providing any non-trivial
guarantees for \emph{single-pass} algorithms, which are by far the most studied and practically appealing variants of (semi-)streaming algorithms\footnote{Beside being quantitatively more efficient, single-pass algorithms are qualitatively more appealing because they can process data generated ``on the fly'' without ever having to store it even once (e.g., in applications in network monitoring).}.

In this work, we answer this fundamental question in the affirmative by designing highly efficient sublinear-time and sublinear-space algorithms for $O(1)$-approximation of correlation clustering in these models: An $\Ot(n)$-time algorithm assuming
adjacency list access model, and a semi-streaming algorithm in a single pass.

\subsection{Our Contributions}
\label{subsec:rst-overview}

Our first main result is a sublinear-time algorithm that instead of adjacency matrix in prior work~\cite{BonchiGK13,Bressan19,Garcia-SorianoK20}, works with the adjacency list of $(+)$-labeled edges and bypass the strong impossibility 
results of~\cite{BonchiGK13,Bressan19}. Formally, 

\begin{theorem}
\label{res:main-time-alg}
There exists a randomized algorithm that given the adjacency list of the $(+)$-labeled subgraph of any labeled graph, with high probability outputs an $O(1)$-approximation of correlation clustering in $O(n\log^2{n})$ time and $O(n\log{n})$ query.
\end{theorem}

To our knowledge, prior to our work, no $o(n^2)$ time algorithm for multiplicative-approximation of correlation clustering was known (under any access model). 
We shall formally define the access model in~\Cref{res:main-time-alg} in~\Cref{sec:sublinear-models} but basically it  involves providing the algorithm with query access to the $(+)$-edges incident on each vertex individually. This seems to be a natural
access from a practical point of view in many applications. For instance, in the applications of coreference~\cite{CohenR02} and cross-lingual link detection~\cite{GaelZ07}, the `natural' labels available are often the positive ones (e.g. the `co-occurance' and the `article similarity'), and the negative labels are usually artificially-inserted. Moreover, in~\Cref{app:models}, we further study other natural sublinear-time access models 
such as adjacency list access to the labeled graph itself or $(-)$-labeled subgraph instead and prove that no multiplicative approximation is possible in these models in $o(n^2)$ time. This highlights our model 
as the more theoretically-natural one for this problem also. 

Our second main result is a single-pass semi-streaming algorithm for correlation clustering. 

\begin{theorem}
\label{res:main-space-alg}
There exists a randomized algorithm that with high probability computes an $O(1)$-approximation of correlation clustering in $O(n\log{n})$ space and a single pass over the edges of any given labeled graph. 
\end{theorem}

To our knowledge, no $o(n^2)$ space streaming algorithms was known for this problem in a single pass before our work. The only single-pass algorithm for this problem that we are aware of is due to~\cite{AhnCGMW21} that 
requires $\Ot(n+m)$ space on graphs with $m$ $(-)$-labeled edges which can  be $\Omega(n^2)$ space\footnote{Note that from a purely streaming point of view, one can entirely store a graph with $m$ $(-)$-labeled edges in $\Ot(n+m)$ space 
(even in a dynamic stream; see~\Cref{app:models}), and then solve the problem exactly on the stored graph at the end of the stream in exponential time.}. 
In~\Cref{app:models}, we further show that our algorithm in~\Cref{res:main-space-alg} can be extended to other streaming models such as when only $(+)$- or $(-)$-labeled edges are arriving, or even to dynamic streams, still in $\Ot(n)$ space. 

\subsection{Our Techniques} 

The earlier work on sublinear algorithms for correlation clustering in~\cite{BonchiGK13,ChierichettiDK14,Garcia-SorianoK20,AhnCGMW21} were all based on implementing
the so-called \emph{Pivot} method of~\cite{AilonCN08}  via sublinear algorithms. The Pivot method is  based on computing a \emph{random-order maximal independent set} of $(+)$-labeled edges and achieves a $3$-approximation. This method however does not seem particularly suitable for either sublinear-time or (single-pass) streaming algorithms: it is known that computing \emph{any} type of maximal independent set (let alone the one required by the Pivot method) requires $\Omega(n^2)$ time given access to both adjacency list or matrix of the input graph~\cite{AssadiCK19a,AssadiS19} as well as $\Omega(n^2)$ space in single-pass streams~\cite{AssadiCK19a,CormodeDK19}. 

In a recent elegant work,~\cite{CohenAddadLMNP21} presented an interesting new insight on the problem. Their approach is based on trimming down the edges of the graph in multiple steps into $\Ot(n)$ edges that can be stored in the memory
and finding connected components of this trimmed graph. The authors then show that placing these connected components into their own clusters achieves an $O(1)$-approximation to the problem. The proof of this part 
is done via a charging scheme that exploits the fact that vertices not in the same connected component have ``different neighborhoods'' while vertices inside the components are ``tightly connected''. 

In this work, we first observe that this general strategy of partitioning a graph into different-neighborhood vs tightly-connected subgraphs is reminiscent of a classical approach
in graph coloring literature referred to as \emph{sparse-dense decompositions}. These decompositions have their 
root in the work of~\cite{MolloyR98,Reed98,Reed99,Reed99b} (see also~\cite{MolloyR10,MolloyR14}) in graph theory and more recently have been at the core of several breakthrough results on graph 
coloring in distributed~\cite{HarrisSS16,ChangLP18,HalldorssonKMT21} and sublinear algorithms~\cite{AssadiCK19a,AlonA20}. A typical sparse-dense decomposition partitions the graph into \emph{sparse} vertices that have many non-edges in their neighborhood, and a collection of \emph{almost-cliques} that are subgraphs which are close to a clique in a property testing sense. It is thus natural to wonder whether such decompositions can be used in place of the trimming step of~\cite{CohenAddadLMNP21}, specially as some earlier work in~\cite{AssadiCK19a} have already shown ways of finding these decompositions via different sublinear algorithms. 

The first challenge in implementing this strategy is that these decompositions  are almost exclusively tailored toward maximum-degree $\Delta$ of the graph, in the sense that their sparse vertices include all vertices with degree, say, $<0.9\Delta$, and their almost-cliques are only $\approx \Delta$-cliques. While this is quite natural for graph coloring problems such as $(\Delta+1)$-coloring and alike, such a decomposition would not be particularly helpful for correlation clustering. The only exception  
that we are aware of is a recent decomposition of~\cite{AlonA20} for the so-called $(\deg+1)$-coloring problem which actually generates different types of sparse vertices and almost-cliques that are proportional to degree of individual vertices. 

It turns out however that the decomposition of~\cite{AlonA20} is \emph{too rigid} to be used in the context of the correlation clustering and the charging framework of~\cite{CohenAddadLMNP21} (we elaborate more on this in~\Cref{sec:decomposition}). 
On top of that, the decomposition of~\cite{AlonA20} is primarily used as a structural result in~\cite{AlonA20} and its only known algorithmic implementation requires using several instantiations of the algorithm of~\cite{AssadiCK19a}, which does not result in simple nor particularly efficient algorithms for the decomposition\footnote{We should emphasize that main results of both~\cite{AssadiCK19a,AlonA20} rely on \emph{existence} of such a decomposition and do 
not require an algorithm for finding the decomposition (although~\cite{AssadiCK19a} give such algorithms also). This is very different from our purpose of using the decomposition in this paper as it is only useful to us if it can be find algorithmically.}. 

Our main technical ingredient in this paper is then to design a new sparse-dense decomposition that remedies this situation. We state our decomposition informally here and postpone the detailed and lengthy definitions to~\Cref{thm:decomposition} 
and~\Cref{thm:decomposition-structural} (see also~\Cref{sec:prelim} for any missing notation).

\begin{mdframed}[backgroundcolor=lightgray!40,topline=false,rightline=false,leftline=false,bottomline=false,innertopmargin=2pt]
\paragraph{A (Yet Another) Sparse-Dense Decomposition:}~ 

\noindent
 For any small constant $\eps>0$, vertices of any graph $G=(V,E)$ (not necessarily a labeled graph)  can be decomposed into the following sets: 
\begin{itemize}
\item \emph{Sparse vertices}: each sparse vertex $v$ has approximately $\eps \cdot \deg{(v)}$ neighbors $u$ such that $N(v)$ and $N(u)$ differ in approximately $\eps \cdot \max\set{\deg(v),\deg(u)}$ vertices\footnote{Beside the recovery algorithm, this is the guarantee that 
is  different from~\cite{AlonA20} and needed for correlation clustering.}. 
\item \emph{Dense vertices}: each dense vertex $v$ belongs to an almost-clique of size approximately $(1 \pm \eps) \cdot \deg(v)$, where an almost-clique is a subgraph of $G$ that can be turned into an actual clique by changing approximately $\eps$-fraction of edges of each one of its vertices. 
\end{itemize}
\noindent
Moreover, there is an algorithm that samples $O\paren{n\log(n)}$ edges of $G$ (from a certain non-uniform distribution) and uses degrees of vertices of $G$ to compute this decomposition in $O(n\log^2{n})$ time. 
\end{mdframed}
 We remark that 
 our way of defining and forming the decomposition is quite different from all recent algorithmic approaches for sparse-dense decompositions in~\cite{HarrisSS16,ChangLP18,AssadiCK19a,AlonA20,HalldorssonKMT21}. Instead, to be able to provide the per-vertex guarantee needed by our decomposition, we follow the classical work of~\cite{Reed98} that seems to give a better handle on the properties of the decomposition.
 As a result, we also give the \emph{first} efficient implementation of this type of decompositions via a sampling algorithm that is easily
 implementable in various computational models including sublinear algorithms studied in this paper. 

At this point, our task of designing sublinear algorithms is simple. Firstly, we show that given 
the decomposition of the $(+)$-labeled subgraph of the input, there is a natural way of forming an $O(1)$-approximation correlation clustering (see~\Cref{thm:cc-alg}), following the approach of~\cite{CohenAddadLMNP21}. 
Basically, sparse vertices of the decomposition are so costly even for optimum solution that one might as well place them in singleton clusters; on the other hand, each almost-clique of dense vertices is so closely connected that the best strategy is to 
cluster them together. Secondly, the sampling algorithm that creates this decomposition is simple enough that can easily be implemented via simple sublinear algorithms (see~\Cref{thm:sub-time-alg,thm:sub-space-alg} and~\Cref{app:models}).

In conclusion, we found the application of sparse-dense decompositions to correlation clustering (and graph clustering) quite natural and hope our work
paves the path for further study of this connection. Moreover, unlike almost all aforementioned work that use sparse-dense decompositions as a subroutine in much more complicated algorithms and proofs, here the main bulk of work is in designing the decomposition itself; as such, this application can perhaps find its way as a gentle(r) introduction to sparse-dense decompositions. 

\subsection{Related Work}
\label{subsec:related-work}
Correlation clustering is one of the most well-studied clustering problems. Apart from the classical settings where the edges are either $(-)$ or $(+)$, known results have also been developed under general graphs, where the edge weights are real numbers and the graphs are not necessarily complete. On this front, the work of \cite{EmanuelF03} gives an $O(\log(n))$-approximation algorithm in polynomial time. The NP-hardness result on labeled (complete) graphs automatically applies to general graphs, and it is further shown that the approximation even on weighted \emph{complete} graphs is APX-hard \cite{EmanuelF03,CharikarGW05}.

Beyond the disagreement minimization objective, another popular optimization target is agreement maximization, which aims to maximize the $(+)$ edges in the same clusters and $(-)$ edges across different clusters. Computing the exact solution of agreement maximization is also NP-hard. However, it admits a PTAS, rendering the objective more tractable for approximation \cite{bansal2004correlation}. Furthermore, for general graphs, the work of \cite{CharikarGW05,Swamy04} give algorithms that achieve $0.766$ approximation in polynomial time. More recently,~\cite{AhmadiKS19} proposed a new min-max objective, whose goal is to minimize the maximum number of disagreement edges inside each cluster. It is further shown in \cite{AhmadiKS19} that such an objective admits a worst-case $O(\log(n))$ approximation in polynomial time.

The quest of sublinear correlation clustering algorithms also goes outside sublinear-time and streaming models. For instance, the `local' correlation clustering introduced by~\cite{BonchiGK13} can approximation the cluster of a \emph{single} vertex in constant time, and the cluster of each vertex is consistent with the `global' clustering. Moreover, efficient sublinear algorithms are explored under the umbrella of parallel computing, especially the Mapreduce-type Massively Parallel Computation (MPC) models. On this front,~\cite{ChierichettiDK14} designed an algorithm that achieves $O(1)$-approximation in $O(\log(n))$ parallel rounds. The result was recently improved by~\cite{CohenAddadLMNP21} to constant many parallel rounds.


\newcommand{\cost}[2]{\ensuremath{\textnormal{\textsf{cost}}_{\, #1}(#2)}\xspace}
\newcommand{\scost}[1]{\ensuremath{\textnormal{\textsf{cost}}(#1)}\xspace}

\renewcommand{\OPT}{\mathcal{O}\xspace}
\renewcommand{\ALG}{\mathcal{A}\xspace}

\newcommand{\CC}{\mathcal{C}}

\newcommand{\sym}{\,\triangle\,}

\section{Preliminaries}\label{sec:prelim}

 \paragraph{Set notation.} For two sets $A$ and $B$, we use $A \sym B := (A - B) \cup (B - A)$ to denote the symmetric difference of $A$ and $B$. We say that a collection $\SS$ of sets is \emph{laminar} if 
 for any two sets $A,B \in \SS$, either $A \cap B = \emptyset$ or $A \subseteq B$ or $B \subseteq A$. For any laminar collection $\SS$, we say that a set $S \in \SS$ is a \emph{root} if $S$ is not a proper subset of any other 
 set in $\SS$. 
 
\paragraph{Graph notation.} For any graph $G=(V,E)$, and vertex $v \in V$, we use $N(v)$ to denote the neighbors of $v$ and $E(v)$ to denote the edges incident on $v$. We say that a pair $(u,v)$ is a \emph{non-edge} in the graph $G$ 
 if there is not an edge between $u$ and $v$ in $G$. 
 
 \paragraph{Concentration inequalities.} We use the following standard forms of Chernoff bound in our proofs. 
\begin{proposition}[Chernoff Bound; cf.~\cite{ASBook}]
\label{prop:chernoff}
Let $X_{1}, X_{2}, \cdots, X_{m}$ be $m$ independent random variables in $[0,1]$. Define $X=\sum_{i=1}^{m} X_{i}$. Then, for every $\delta>0$ and $t \geq 1$, 
\begin{align*}
\Pr\paren{\card{X - \expect{X}} \geq \delta\cdot\expect{X}} \leq 2 \cdot \exp\paren{-\frac{\delta^{2}}{2+\delta}\cdot\expect{X}} \quad \text{and} \quad \Pr\paren{\card{X - \expect{X}} \geq t} \leq 2 \cdot \exp\paren{-\frac{2t^2}{m}}.
\end{align*}
\end{proposition}

We also use the following form of Bernstein's inequality. 

\begin{proposition}[Bernstein's inequality; cf.~\cite{Vershyninbook}]\label{prop:bernstein} 
Let $X_1,\ldots,X_m$ be $m$ independent random variables such that $\expect{X_i} = 0$ and $\card{X_i} < M$ for all $i \in [m]$. Then, for any $t \geq 1$, 
\[
	\Pr\paren{\sum_{i=1}^m X_i \geq t} \leq \exp\paren{-\frac{t^2}{2 \sum_{i=1}^m \expect{X_i^2} + \nicefrac23 \cdot M \cdot m}}. 
\]
\end{proposition}
 \subsection{Problem Definition}\label{sec:problem}
 Throughout, by a \emph{labeled} graph $G=(V,E)$, we mean a complete graph with edges in $E$ labeled in $\set{-1,+1}$. We use $G^+$ and $G^-$ to denote 
 the subgraphs of $G$ consisting of only $(+)$-edges and $(-)$-edges, respectively. We extend this definition analogously to neighbor-sets $N^+(v)$ and $N^-(v)$, and edge-sets $E^+(v)$ and $E^-(v)$, for every $v \in V$.  
 
 Suppose we are given a labeled graph $G=(V,E)$. Let $\CC$ be any clustering of vertices of $G$ into disjoints clusters $C_1,\ldots,C_k$. For any vertex $v \in V$, 
 we use $\CC(v)$ to denote the cluster $C_i \in \CC$ that $v$ belongs to. For any edge $e=(u,v)$, we define the \emph{cost} of $e$ in the clustering $\CC$ as: 
\begin{align}
	\cost{\CC}{e} = 
	\begin{cases} 
	1 & \text{if $e \in G^+$ and $\CC(u) \neq \CC(v)$} \\ 
	1 & \text{if $e \in G^-$ and $\CC(u) = \CC(v)$} \\
	0 & \text{otherwise}
	\end{cases}. \label{eq:cost}
\end{align}
In words, cost of $(+)$-edge  is $1$ if its endpoints are clustered differently, and cost of a $(-)$-edge is $1$ if its endpoints are clustered together. 
The \emph{total cost} of a clustering $\CC$ is then:
\begin{align}
	\scost{\CC} = \sum_{e \in G} \cost{\CC}{e}. \label{eq:scost}
\end{align}
The goal in the correlation clustering problem is  to find a clustering $\CC$ that minimizes~\Cref{eq:scost}. 

\subsection{Sublinear Algorithms Models} \label{sec:sublinear-models}

In this paper, we focus on two of the most canonical models of sublinear algorithms, namely, sublinear-time algorithms, and (sublinear-space) streaming algorithms. These models are defined formally as follows. 

\paragraph{Sublinear-time algorithms.} When working with sublinear-time algorithms, it is important to specify the exact data model as the algorithm does not even have time to read the input once. In this paper, 
we assume the algorithms are given access to the \emph{adjacency list} of the $(+)$-graph $G^+$ of the input labeled graph $G$. This means that the algorithm can query the following information in $O(1)$ time: 
\begin{enumerate}[label=$\roman*).$]
	\item \emph{Degree queries:} What is $\degp{v}$ of a given vertex $v \in V$? 
	\item \emph{Neighbor queries:} What is the $i$-th vertex in $N^+(v)$ of $v \in V$ for $i \leq \degp{v}$? 
\end{enumerate}
The goal is to return a correlation clustering of $G$ (under cost function of~\Cref{eq:scost}) in a limited time. 

A remark about this model is in order. The standard query model for graph problems provides access to the adjacency list (or matrix) of $G$ itself (and not that of $G^+$). But in the context of labeled graphs, 
adjacency list of $G$ itself provides little information: degree queries are entirely uninformative (always return $n-1$) and 
neighbor queries only reveal the label of the edge between vertex $v$ to some other vertex $u$, similar to access to the adjacency matrix. Alternatively, we could have also considered 
access to the adjacency list of the $(-)$-graph $G^-$ instead, which at least is more informative than that of $G$. 

Nevertheless, we prove that neither 
model allows for any non-trivial sublinear-time algorithm for correlation clustering with \emph{any multiplicative} approximation guarantees (see~\Cref{app:time-lower}). 
In light of this impossibility result, and our sublinear-time algorithms, we believe the model we consider for this problem is most natural from the perspective of sublinear-time algorithms. 

\paragraph{Semi-streaming algorithms.} Semi-streaming algorithms focus on minimizing the space usage as opposed to time. In this model, 
the vertices of input labeled graph $G=(V,E)$ are known and the edges $E$ arrive one by one in a stream together with their labels. The goal is to read this stream in the given order only once\footnote{Or a few times in case of multi-pass algorithms -- our algorithm in this paper however is single-pass.} and use only $O(n \cdot \poly\!\log{(n)})$ space measured in machine words of size $O(\log{n})$ bits. At the end of the stream, the algorithm should return a correlation clustering of $G$ (under cost function of~\Cref{eq:scost}). 

Our streaming model is the same as the one studied by earlier work on this problem. But one can again wonder what would happen if only edges of $G^+$ or $G^-$ are being streamed instead of $G$. It turns out unlike the sublinear-time model, 
these different choices do not matter much for our purpose. In~\Cref{app:stream-upper}, we show that our algorithm can be extended to hand either of these cases, plus other natural variants such as {dynamic} (insertion-deletion) streams 
at the cost of increasing the space by at most $\poly\!\log{(n)}$ factor.


\newcommand{\low}[1]{N_{low}(#1)}

\newcommand{\Isolated}[2]{\ensuremath{{\textnormal{\texttt{Isolated}}_{#2}}}(#1)}
\newcommand{\isolated}[1]{\Isolated{#1}{\eps}}

\newcommand{\Dense}[3]{\ensuremath{{\textnormal{\texttt{Dense}}_{#2,#3}}}(#1)}
\newcommand{\Low}[2]{\ensuremath{\textnormal{\texttt{Low}}_{#2}}(#1)}
\renewcommand{\low}[1]{\ensuremath{\textnormal{\texttt{Low}}_{\eps}}(#1)}

\newcommand{\Kernel}[3]{\ensuremath{\textnormal{\texttt{Kernel}}_{#2,#3}}(#1)}
\newcommand{\kernel}[1]{\Kernel{#1}{\eps}{\delta}}

\newcommand{\Candid}[3]{\CC_{#2,#3}(#1)}

\newcommand{\candid}{\Candid{G}{\eps}{\delta}}

\newcommand{\NS}[1]{\ensuremath{N_{\textnormal{sample}}(#1)}}

\renewcommand{\deg}[1]{\ensuremath{\textnormal{{deg}}(#1)}}

\newcommand{\sample}{\ensuremath{\textnormal{\texttt{Sample}}}\xspace}

\newcommand{\Vsparse}{V_{\text{sparse}}}
\newcommand{\Vlight}{V_{\text{light}}}
\newcommand{\Vlowsparse}{V_{\text{low-sparse}}}

\newcommand{\tC}[1]{\widetilde{C}(#1)}

\newcommand{\Esample}{E_{\textnormal{sample}}}


\section{A (Yet Another) Sparse-Dense Decomposition} 
\label{sec:decomposition}

We present our sparse-dense decomposition in this section. We state the theorem in a form that allows for its recovery via sublinear algorithms in  subsequent sections -- however, we opted to present the recovery algorithm 
in a model-independent manner as this general form can also be applicable in other models of computation not considered in this paper.

\begin{theorem}[Sparse-Dense Decomposition (algorithmic version)]\label{thm:decomposition}
	There are absolute constants $\eps_0, \eta_0 > 0$ such that the following is true. For every $\eps < \eps_0$, vertices of any given graph $G=(V,E)$ can be partitioned into the following sets: 
\begin{itemize}

\item \textbf{Sparse vertices} $\Vsparse$: Any vertex $v\in \Vsparse$ has at least $\eta_0 \cdot \eps \cdot \deg{v}$ neighbors $u$ such that: 
\[
	\card{N(v) \sym N(u)} \geq \eta_0 \cdot \eps \cdot \max\set{\deg{u},\deg{v}}. 
\]
\item \textbf{Dense vertices partitioned into \underline{almost-cliques} $K_1,\ldots,K_k$}: For every $i \in [k]$, each $K_i$ has the following properties. Let 
 $\Delta(K_{i})$ be the maximum degree (in $G$) of the vertices in $K_i$, then: 
\begin{enumerate}[label=$\roman*)$.]
\item\label{dec:non-neighbors}  Every vertex $v\in K_{i}$ has at most $\eps\cdot \Delta(K_{i})$ \emph{non-neighbors} inside ${K_{i}}$;
\item\label{dec:neighbors} Every vertex $v\in K_{i}$ has at most $\eps\cdot \Delta(K_{i})$ \emph{neighbors} outside ${K_{i}}$;
\item\label{dec:size} Size of each $K_i$ satisfies $(1-\eps)\cdot \Delta(K_{i}) \leq |K_{i}| \leq (1+\eps)\cdot \Delta(K_{i})$.
\end{enumerate}
\end{itemize}
Moreover, there is an absolute constant $c > 0$ and an algorithm that given access to only the following information about  $G$, with high probability, computes this decomposition of $G$ in $O(\eps^{-2} \cdot n\log^2{n})$ time:
\begin{itemize}
	\item \textbf{Degree information:} Set of all vertices $v \in V$ plus their degrees $\deg{v}$; 
	\item \textbf{Random edge samples:} A collection of sets $\NS{v}$ of 
	\[
	t= c \cdot \eps^{-2} \cdot \log{n}
	\]
	neighbors of each vertex $v \in V$ chosen independently and uniformly at random (with repetition); 
	\item \textbf{Random vertex samples:} A set $\sample$~of vertices wherein each $v \in V$ is included independently with probability 
	\[
		p_v := \min\set{\frac{c \cdot \log{n}}{\deg{v}},1},
	\]
	together with all the neighborhood $N(v)$ of each sampled vertex $v \in \sample$. 
\end{itemize} 
(The probability of success of the algorithm is over the random choice of  edge and vertex samples.)  
\end{theorem}

 The sparse vertices in~\Cref{thm:decomposition} are such that ``many'' of their neighbors have a ``different'' neighborhood than themselves. Thus, even though we refer to them as `sparse' to be consistent 
with prior sparse-dense decompositions, these vertices do not necessarily have a sparse neighborhood as in standard decompositions but rather have a $2$-hop neighborhood that is very different than their $1$-hop neighborhood. 

The almost-cliques on the other hand, as the name suggests, are basically induced subgraphs of $G$ on ``similar degree'' vertices that can be turned into an actual clique by changing a small fraction of edges in their neighborhood. This part 
is also different from typical decompositions in that the almost-cliques are allowed to have varying sizes tailored to degrees of individual vertices as opposed to a single size based on maximum degree. The only other sparse-dense decomposition
with such guarantees that we know of is that of~\cite{AlonA20}. However, both in terms of precise guarantees and also the construction, our~\Cref{thm:decomposition} is quite different from~\cite{AlonA20}. To be specific: 
\begin{itemize}
	\item The sparse vertices in~\Cref{thm:decomposition} have an \emph{individual} guarantee on their ``different'' neighbors, while~\cite{AlonA20} makes an \emph{aggregate} guarantee for the entire neighborhood of a sparse vertex. 
	\item The construction of~\cite{AlonA20} is based on the notion of \emph{balanced-} and \emph{friend-edges}, which does not seem to allow for the fine-grained guarantees required by our decomposition. Instead, the proof of~\Cref{thm:decomposition}
	involves a more direct approach based on classical sparse-dense decompositions. 
	\item Finally, the decomposition in~\cite{AlonA20} is a structural result while ours is constructive via the sampling algorithm, which is needed for recovering this decomposition via sublinear algorithms. 
 \end{itemize}
\noindent
The rest of this section is dedicated to the proof of this theorem. 

\paragraph{Remark:} The proofs in this section require establishing various properties and claims which are generally simple but some require lengthy calculations. To keep the flow of arguments, we postpone the proofs of more straightforward
properties to~\Cref{app:decomposition} (marked with $\bigstar$) and only include the main or  subtler ones here.

\subsection{Preliminary Definitions and Properties}\label{sec:decomposition-prelim}

Throughout this section, we let $G=(V,E)$ be any arbitrary undirected graph on $n$ vertices, and $\eps,\delta \in (0,1)$ be sufficiently small parameters. 
To define our decomposition, we need to start with the definition of a couple of different types of vertices in the following. 

\smallskip
\noindent
\paragraph{Light vertices.} 
The first type of vertices are  the ones that have ``many'' ``high degree'' neighbors. 
\begin{definition}[{Light vertices}]\label{def:light}
	For a vertex $v \in V$, we define the set of \underline{\emph{\eps-low-degree}} neighbors of $v$ as:
	\[
		\low{v} := \set{u \in N(v) \mid \deg{u} \leq (1+\eps) \cdot \deg{v}},
	\]
	that is, the neighbors of $v$ that have degree at most $(1+\eps)$ times larger than $v$.
	
	\noindent	
	We call a vertex $v \in V$ \underline{\emph{$(\eps,\delta)$-light}} if size of $\low{v}$ is at most $(1-\delta) \cdot \deg{v}$. 
\end{definition}

The key  property of light vertices for us is that they have ``many'' neighbors with a ``different'' neighborhood than that of the light vertex. This is  because light vertices have many higher degree neighbors. 

\begin{property}[$\bigstar$]\label{p:light}
	Any light vertex $v$ has at least $\delta \cdot \deg{v}$ neighbors $u$ such that 
	\[
		\card{N(u) - N(v)} \geq \frac\eps{(1+\eps)} \cdot \deg{u} = \frac\eps{(1+\eps)} \cdot \max\set{\deg{u},\deg{v}}. 
	\]
\end{property}

\paragraph{Low-sparse vertices.} The next set of vertices are the ones that have ``many'' non-edges between their ``low degree'' neighbors. Formally, 

\begin{definition}[{Low-sparse vertices}]\label{def:isolated}
	We define the set of \underline{\emph{$\eps$-isolated}} neighbors of $v$ as 
	\[
		\isolated{v} = \set{u \in N(v) \mid \card{\low{v} - N(u)} \geq \eps \cdot \deg{v}},
	\]
	that is, vertices $u$ that $v$ has at least $\eps \cdot \deg{v}$ neighbors in $\low{v}$ that are \emph{not} neighbor to $u$. 
	
	\noindent
	We call a vertex $v \in V$ \underline{\emph{$(\eps,\delta)$-low-sparse}} iff it has at least $\delta \cdot \deg{v}$ $\eps$-isolated neighbors in $\low{v}$, i.e., 
	\[
		\card{\isolated{v} \cap \low{v}} \geq \delta \cdot \deg{v}. 
	\]  
	(These vertices are called low-sparse as the subgraph induced on their low-degree neighbors is (rather) sparse.)
\end{definition}

The key property of low-sparse vertices for us is that again they have ``many'' neighbors with a ``different'' neighborhood than that of the light vertex. This is because low-sparse vertices have many non-neighbors among their low degree neighbors. 

\begin{property}[$\bigstar$]\label{p:isolated}
	Any low-sparse vertex $v$ has at least $\delta \cdot \deg{v}$ neighbors $u$ such that 
	\[
		\card{N(v) - N(u)} \geq \eps \cdot \deg{v} \geq \frac{\eps}{(1+\eps)} \cdot \max\set{\deg{u},\deg{v}}. 
	\]
\end{property}

\paragraph{Dense vertices.} Finally, we pack all the remaining vertices into one definition. 

\begin{definition}[{Dense vertices}]\label{def:dense}
	Any vertex $v \in V$ which is neither $(\eps,\delta)$-light nor $(\eps,\delta)$-low-sparse is called a \underline{\emph{$(\eps,\delta)$-dense}} vertex. We use $\Dense{G}{\eps}{\delta}$ to denote the set
	of $(\eps,\delta)$-dense vertices in the graph $G$. 
\end{definition}

The main part of the decomposition is to handle dense vertices. Unlike light and low-sparse vertices, the main property of dense vertices for our decomposition is that we can ``bundle'' them together to form almost-cliques -- this is the main step of the decomposition and is handled in the next 
subsection. But before we move on, we first list some key properties of dense vertices that will be needed for the next step. 

The first property of dense vertices is that there are ``few'' non-edges between their low degree neighbors, as well as ``few'' edges going out of their low degree neighbors. Both of these are intuitively true as the induced subgraph of 
dense vertices on their low degree vertices is not sparse (because they are not low-sparse) and they have many low degree neighbors (because they are not light). 

\begin{property}[$\bigstar$]\label{p:dense}
	For every dense vertex $v \in \Dense{G}{\eps}{\delta}$:
	\begin{enumerate}[label=$(\roman*)$]
		\item the number of non-edges inside $\low{v}$ is at most $\frac{\eps+\delta}{2} \cdot \deg{v}^2$;
		\item the number of edges going out of $\low{v}$ is at most $2\,(\eps+\delta) \cdot \deg{v}^2$;
	\end{enumerate} 
\end{property}

We can also show that neighborhood of most vertices in $\low{v}$ has a large intersection with $\low{v}$ itself, in other words, the subgraph induced on $\low{v}$ is ``almost a clique'' (this should not be confused with the definition of 
almost-cliques we use in our decomposition). To this end, we have our final definition. 

\begin{definition}[Kernel]\label{def:kernel}
	For every dense vertex $v \in \Dense{G}{\eps}{\delta}$, we define \underline{\emph{kernel}} of $v$ as: 
	\[
		\kernel{v} := \low{v} - \isolated{v}.
	\]
	(These are low-degree neighbors of $v$ that share ``many'' neighbors with other low-degree neighbors of $v$.)
\end{definition}

The following property formalizes our discussion before the definition of kernel. 

\begin{property}[$\bigstar$]\label{p:kernel}
For every dense vertex $v \in \Dense{G}{\eps}{\delta}$, $\kernel{v}$ satisfies the following properties: 
	\begin{enumerate}[label=$(\roman*)$]
		\item $\kernel{v}$ is a subset of $\low{v}$ with size at least $(1-2\delta) \cdot \deg{v}$; 
		\item every vertex $u \in \kernel{v}$ has at least $(1-\eps-\delta) \cdot \deg{v}$ neighbors in $\low{v}$.
	\end{enumerate} 
\end{property}

Moreover, we prove that kernel vertices are also dense for a slightly larger parameters $\eps$ and $\delta$. This is again because kernel of a dense vertex is ``almost a clique''. 

\begin{property}[$\bigstar$]\label{p:kernel-dense}
	For every $(\eps,\delta)$-dense vertex $v$, any vertex $u \in \kernel{v}$ is $(4\eps+2\delta,2\eps+2\delta)$-dense. 
\end{property}

Finally, we argue  kernel vertices for a dense vertex are ``almost monotone'' in the parameter $\eps$. Formally, 
\begin{property}[$\bigstar$]\label{p:kernel-monotone}
	For every $(\eps,\delta)$-dense vertex $v$, any vertex in $\Kernel{v}{\eps}{\delta}$ also belongs to $\Kernel{v}{\eps'}{\delta'}$ for any $\eps' > \eps+\delta$ and arbitrary $\delta' > 0$. 
\end{property}

An illustration of the definition and properties of the above sets of vertices can be found in \Cref{fig:decompose-illus}.

\begin{figure}[h!]
\centering
\tikzset{every picture/.style={line width=0.75pt}} 


\begin{tikzpicture}[x=0.75pt,y=0.75pt,yscale=-0.75, xscale=0.75, every node/.style={scale=0.75}]

\draw   (58.07,181.42) -- (599,181.42) -- (599,319) -- (58.07,319) -- cycle ;
\draw   (448,181.42) -- (599,181.42) -- (599,319) -- (448,319) -- cycle ;
\draw    (58.07,181.42) -- (316.14,49.1) ;
\draw   (229,255.8) .. controls (229,247.07) and (236.07,240) .. (244.8,240) -- (522.2,240) .. controls (530.93,240) and (538,247.07) .. (538,255.8) -- (538,319) .. controls (538,319) and (538,319) .. (538,319) -- (229,319) .. controls (229,319) and (229,319) .. (229,319) -- cycle ;
\draw   (257.14,294.1) .. controls (257.14,284.76) and (264.27,277.2) .. (273.07,277.2) .. controls (281.87,277.2) and (289,284.76) .. (289,294.1) .. controls (289,303.43) and (281.87,311) .. (273.07,311) .. controls (264.27,311) and (257.14,303.43) .. (257.14,294.1) -- cycle ;
\draw  [color={rgb, 255:red, 0; green, 0; blue, 255 }  ,draw opacity=1 ][dash pattern={on 5.63pt off 4.5pt}][line width=1.5]  (207.86,288.42) .. controls (219.53,263.35) and (267.73,261.05) .. (315.53,283.29) .. controls (363.33,305.53) and (392.62,343.88) .. (380.96,368.95) .. controls (369.29,394.03) and (321.09,396.32) .. (273.29,374.09) .. controls (225.49,351.85) and (196.2,313.5) .. (207.86,288.42) -- cycle ;
\draw  [color={rgb, 255:red, 255; green, 0; blue, 0 }  ,draw opacity=1 ][dash pattern={on 1.69pt off 2.76pt}][line width=1.5]  (444,187) -- (444,236) -- (222,237) -- (222,314) -- (63,314) -- (63,186) -- (236,188) -- cycle ;
\draw   (560.14,229.1) .. controls (560.14,219.76) and (567.27,212.2) .. (576.07,212.2) .. controls (584.87,212.2) and (592,219.76) .. (592,229.1) .. controls (592,238.43) and (584.87,246) .. (576.07,246) .. controls (567.27,246) and (560.14,238.43) .. (560.14,229.1) -- cycle ;
\draw   (316.14,49.1) .. controls (316.14,39.76) and (323.27,32.2) .. (332.07,32.2) .. controls (340.87,32.2) and (348,39.76) .. (348,49.1) .. controls (348,58.43) and (340.87,66) .. (332.07,66) .. controls (323.27,66) and (316.14,58.43) .. (316.14,49.1) -- cycle ;
\draw    (153,181) -- (316.14,49.1) ;
\draw    (515,181) -- (348,49.1) ;
\draw    (599,181.42) -- (348,49.1) ;
\draw    (226,181) -- (316.14,49.1) ;
\draw    (448,181.42) -- (348,49.1) ;

\draw (157.38,223.56) node [anchor=north west][inner sep=0.75pt]   [align=left] {$ $};
\draw (311.07,150.88) node [anchor=north west][inner sep=0.75pt]    {$N( v)$};
\draw (327.19,44.44) node [anchor=north west][inner sep=0.75pt]    {$v$};
\draw (451.8,185.05) node [anchor=north west][inner sep=0.75pt]    {$N( v) -\text{Low}_{\epsilon }( v)$};
\draw (568.89,222.58) node [anchor=north west][inner sep=0.75pt]    {$w$};
\draw (234.37,244.14) node [anchor=north west][inner sep=0.75pt]    {$\text{Isolated}_{\epsilon }( v)$};
\draw (62.8,186.05) node [anchor=north west][inner sep=0.75pt]    {$\text{Low}_{\epsilon }( v)$};
\draw (266.89,287.58) node [anchor=north west][inner sep=0.75pt]    {$u$};
\draw (312.8,344.05) node [anchor=north west][inner sep=0.75pt]  [color={rgb, 255:red, 0; green, 0; blue, 255 }  ,opacity=1 ]  {$N( u)$};
\draw (315,87) node [anchor=north west][inner sep=0.75pt]   [align=left] {......};

\end{tikzpicture}

\caption{\label{fig:decompose-illus} The $\low{v}$ and $\isolated{v}$ vertices in $N(v)$. For each vertex $w \in \low{v}$, there is $\deg{w}> (1+\eps)\cdot \deg{v}$. For each vertex $u \in \isolated{v}$, it is \emph{not} a neighbor to at least $\eps\cdot \deg{v}$ vertices in $\low{v}$. If $v$ is \emph{dense}, the area circumscribed by the dotted lines is the $\kernel{v}$.}

\end{figure}
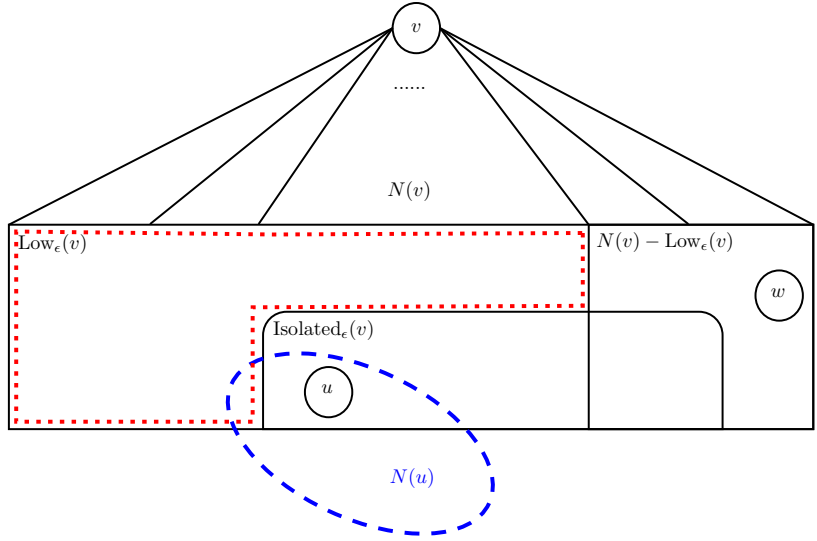

\subsection{Partitioning Dense Vertices into Almost-Cliques}

We now get to the main part of the decomposition which involves bundling the dense vertices into disjoint almost-cliques. We emphasize that these almost-cliques, in addition to partitioning \emph{all} of dense vertices, may also include \emph{some} light or low-sparse vertices. In order to do this, we define a \emph{candidate set} $C_v$ of vertices for every dense vertex $v$, and then show how to pick a subset of these candidate sets to form the  almost-cliques. 

Before we get to the definition of these sets however an important remark is in order. In order for us to be able to eventually recover the decomposition via a sampling algorithm (as in the second part of~\Cref{thm:decomposition}), we need our criteria in the definition of candidate sets to be somewhat relaxed. As such, the definition we get for the collection of candidate sets is  \emph{not unique}, but the properties we prove subsequently hold for any valid choice of these sets according to our definition.

\begin{definition}[Candidate sets]\label{def:candidate-sets}
	For any $(\eps,\delta)$-dense vertex $v$, a \emph{candidate set} $C_v$ is a set of vertices that satisfy the following rules: 

\begin{enumerate}[label=$(\arabic*)$]
	\item \underline{Every} vertex $u$ with the following property should be included in $C_v$: 
	\[
	\card{N(u) \cap \low{v}} \geq (1-6\eps-6\delta) \cdot \deg{v} \quad \text{and} \quad \deg{u} \leq (1+2\eps+2\delta) \cdot \deg{v}.
	\]
	\item \underline{No} vertex $u$ with the following property can be included in $C_v$: 
	\[
	\card{N(u) \cap \low{v}} < (1-7\eps-7\delta) \cdot \deg{v} \quad \text{or} \quad \deg{u} > (1+2\eps+2\delta) \cdot \deg{v}.
	\]
	\end{enumerate}
	(The exact choice of which vertices to include or not in $C_v$ is arbitrary as long as it satisfies the given rules.)
	
	\noindent
	We use $\Candid{G}{\eps}{\delta} := \set{C_v \mid v \in \Dense{G}{\eps}{\delta}}$ to denote the collection of (a choice of) candidate sets. 
\end{definition}

Let us emphasize again that the ``rule-based'' definition of candidate sets is to allow for recovering a valid choice of  these sets via the specified samples in~\Cref{thm:decomposition}; if one is only interested in the existence of the decomposition, a single threshold can be picked
instead in Rules (1) and (2) which collapses them into a single rule and results in a unique choice for each set $C_v$\footnote{In other words, we could define 
$C_v := \set{u \in V \mid \card{N(u) \cap \low{v}} \geq (1-6\eps-6\delta) \cdot \deg{v} ~ \text{and} ~ \deg{u} \leq (1+2\eps+2\delta) \cdot \deg{v}}$.}. 

In the following, we first start by stating the  individual properties of each set $C_v$ and vertices inside it, and then switch to the collective properties of the collection $\candid$ that allows to partition dense vertices. 

\paragraph{Individual properties of candidate sets $C_v$.} We first argue that many vertices in $\low{v}$, in particular $\kernel{v}$ (defined in~\Cref{p:kernel}), will be added to $C_v$. This is because vertices in $\kernel{v}$ have a ``large'' intersection
with $\low{v}$ and thus should be included by Rule (1). 

\begin{property}[$\bigstar$]\label{p:1} 
	$\kernel{v}$ belongs to $C_v$ and thus $\card{\low{v} \cap C_v} \geq (1-2\delta) \cdot \deg{v}$. 
\end{property}
The next property ensures that every vertex in $C_v$ has many neighbors in $C_v$. This is because only vertices with ``large enough'' intersection with $\low{v}$ are included in $C_v$ by Rule (2) and $C_v$ and $\low{v}$ themselves
have a large intersection by the previous property. 

\begin{property}[$\bigstar$]\label{p:2}
	Every vertex $u \in C_v$ satisfies $\card{N(u) \cap C_v} \geq (1-7\eps-9\delta) \cdot \deg{v}$. 
\end{property}

The next property ensures that every vertex in $C_v$ has ``few'' neighbors outside of $C_v$. This is because vertices in $C_v$ have a degree proportional to $v$ by Rule (2) and by the previous property, most of their neighbors 
should be inside $C_v$ instead. 

\begin{property}[$\bigstar$]\label{p:3}
	Every vertex $u \in C_v$ satisfies $\card{N(u) - C_v} \leq (9\eps+11\delta) \cdot \deg{v}$. 
\end{property}

We also show that number of vertices in $C_v$ that are not part of $\low{v}$ is small. This is because the number of edges going out of $\low{v}$ is small but Rule (2) requires every vertex included in $C_v$ to ``consume'' many of these edges. 

\begin{property}[$\bigstar$]\label{p:4}
	 $\card{C_v - \low{v}} \leq (3\eps+3\delta) \cdot \deg{v}$.  
\end{property}

Combining the previous two properties, we can also show that each vertex in $C_v$ has few non-neighbors inside of $C_v$. This is because vertices in $C_v$ have a large intersection with $\low{v}$, the same as $C_v$. 

\begin{property}[$\bigstar$]\label{p:4'}
	Every vertex $u \in C_v$ satisfies $\card{C_v - N(u)} \leq (10\eps+10\delta) \cdot \deg{v}$. 
\end{property}

We can also ensure that each candidate set $C_v$ includes the vertex $v$ itself. This is simply because $N(v)$ intersect with all of $\low{v}$ by definition and thus $v$ should be included by Rule (1). 

\begin{property}[$\bigstar$]\label{p:5}
	Every dense vertex $v$ belongs to its candidate set $C_v$. 
\end{property}

At this point, it is easy to see that each set $C_v$ satisfies all the required properties for an almost-clique we require in our decomposition (for a proper setting of parameters). Formally, 
\begin{property}\label{p:C-almost-clique}
	For $\Delta(C_v) := \max_{u \in C_v} \deg{u}$, we have, 
	\begin{enumerate}[label=$\roman*)$.]
\item  Every vertex $u \in C_v$ has at most $(10\eps+10\delta)\cdot \Delta(C_v)$ \emph{non-neighbors} inside ${C_v}$;
\item Every vertex $u \in C_v$ has at most $(9\eps+11\delta) \cdot \Delta(C_v)$ \emph{neighbors} outside $C_v$;
\item $(1-2\eps-4\delta)\cdot \Delta(C_v) \leq |C_v| \leq (1+3\eps+3\delta)\cdot \Delta(C_v)$.
\end{enumerate}
\end{property}
\begin{proof}
	Each part can be proven as follows: 
\begin{enumerate}[label=$\roman*)$.]
\item Follows from~\Cref{p:4'} and since $\deg{v} \leq \Delta(C_v)$ as $v \in S_v$ by~\Cref{p:5};
\item Follows from~\Cref{p:3} similar to part $i).$ above;  
\item By~\Cref{p:1}, $\card{C_v} \geq (1-2\delta) \cdot \deg{v} \geq (1-2\eps - 4\delta) \cdot \Delta(C_v)$ by the upper bound on $\Delta(C_v)$. Moreover, by~\Cref{p:4}, $\card{C_v} \leq (3\eps + 3\delta) \cdot \deg{v} \leq (3\eps + 3\delta) \cdot \Delta(C_v)$. \qed
\end{enumerate}
\end{proof}

Note that despite having proved~\Cref{p:C-almost-clique}, we are still far from being done: it is not yet clear that these candidate sets \emph{partition} dense vertices and allow for 
constructing \emph{disjoint} almost-cliques needed in the decomposition. This is the content of the next part. 

\paragraph{Collective properties of candidate sets $C_v$.}  So far, we only examined the candidate sets $C_v$ in isolation. We now consider these sets in conjunction with each other. 

The following is the main property of candidate sets. 
Roughly speaking, it states that if a vertex $u$ does not have a ``too large degree'' compared to $v$ and $C_u \cap C_v$ is non-empty, then $u$  itself is also in $C_v$. The proof of this property is more involved than the 
rest but the main idea is as follows. Assuming $C_u \cap C_v$ is non-empty forces $\low{u}$ and $\low{v}$ to intersect ``heavily'' with each other (given combination of several of properties established in the previous part); this in turn forces $u$ to have enough intersection with $\low{v}$ also to join $C_v$ (its degree already satisfy the needed bounds). 

\begin{property}\label{p:6}
	If $C_u \cap C_v \neq \emptyset$ and $\deg{u} \leq (1+2\eps+2\delta) \cdot \deg{v}$, then $u$ also belongs to $C_v$. 
\end{property}
\begin{proof}
	Fix a vertex $w \in C_u \cap C_v$. By~\Cref{p:2}, we have that
	\begin{align}
		\card{N(w) \cap C_u} \geq (1-7\eps-9\delta) \cdot \deg{u} \qquad \text{and} \qquad \card{N(w) \cap C_v} \geq (1-7\eps-9\delta) \cdot \deg{v}. \label{eq:loose-bound2}
	\end{align}
	Combining this with~\Cref{p:4}, we have, 
	\begin{equation}
	\label{eq:NwL} 
	\begin{aligned}
		\card{N(w) \cap \low{u}} &\geq \card{N(w) \cap C_u} - \card{C_u - \low{u}} \geq (1-10\eps-12\delta) \cdot \deg{u}; \\ 
		\card{N(w) \cap \low{v}} &\geq \card{N(w) \cap C_v} - \card{C_v - \low{v}} \geq (1-10\eps-12\delta) \cdot \deg{v}.
	\end{aligned}
	\end{equation}
	Moreover, since $w$ is in both $C_u$ and $C_v$, by Rule (2), we know that 
	\begin{align}
		\deg{w} \leq (1+2\eps+2\delta) \cdot \deg{u} \quad \text{and} \quad \deg{w} \leq (1+2\eps+2\delta) \cdot \deg{v}. \label{eq:degrees} 
	\end{align}
	Given these bounds, we have 
	\begin{align}
		\card{\low{u} \cap \low{v}} &\geq \card{N(w) \cap \low{u}} - \card{N(w) - \low{v}} \notag \\
		&\geq \card{N(w) \cap \low{u}} + \card{N(w) \cap \low{v}} - \deg{w} \notag \\
		&\geq \frac{(1-10\eps-12\delta)}{(1+2\eps+2\delta)} \cdot \deg{w} + \frac{(1-10\eps-12\delta)}{(1+2\eps+2\delta)} \cdot \deg{w} - \deg{w} \tag{by~\Cref{eq:NwL} and~\Cref{eq:degrees}} \\
		&\geq (1-22\eps-26\delta) \cdot \deg{w} \notag \\
		&\geq (1-27\eps-33\delta) \cdot \max\set{\deg{u},\deg{v}} \tag{by~\Cref{eq:loose-bound2}, $\deg{w} = \card{N(w)} \geq (1-5\eps-7\delta) \cdot \max\set{\deg{u},\deg{v}}$} \\
		&\geq \frac{5}{6} \cdot \max\set{\deg{u},\deg{v}}, \label{eq:loose-bound}
	\end{align}
	for $\eps,\delta < 1/360$. This gives us a loose lower bound on the size of intersection of $\low{u}$ and $\low{v}$. In the following, we build on this lower bound to refine it into a much sharper bound for our purpose. 
	
	Define $F(u,v)$ as the set of edges between $(\low{u} \cap \low{v})$ and $(\low{u} - \low{v})$. We have:
	\begin{itemize}
		\item On one hand, all edges of $F(u,v)$ are going from inside of $\low{v}$ to outside of $\low{v}$. Since $v$ is a dense vertex, by part $(ii)$ of~\Cref{p:dense}, we should have 
		\[
		\card{F(u,v)} \leq (2\eps+2\delta) \cdot \deg{v}^2.
		\] 
		\item On the other hand, $F(u,v)$ contains all edges between vertices in $\low{u}$. Since $u$ is a dense vertex, by part $(i)$ of~\Cref{p:dense},  number of non-edges inside $\low{u}$ is at most $(\eps+\delta)/2 \cdot \deg{u}^2$. Thus, 
		\[
		\card{F(u,v)} \geq \card{\low{u} \cap \low{v}}  \cdot \card{\low{u} - \low{v}} - (\eps+\delta)/2 \cdot \deg{u}^2.
		\] 
	\end{itemize}
	Combining the above two bounds we have, 
	\[
		\card{\low{u} \cap \low{v}}  \cdot \card{\low{u} - \low{v}} \leq (\nicefrac52) \cdot (\eps+\delta) \cdot \max\set{\deg{u},\deg{v}}^2. 
	\]
	Using this with the (loose) lower bound of~\Cref{eq:loose-bound}, we have 
	\begin{align*}
		\card{\low{u} - \low{v}} &\leq (\nicefrac52) \cdot (\eps+\delta) \cdot \max\set{\deg{u},\deg{v}}^2 / (\nicefrac56\cdot \max\set{\deg{u},\deg{v}}) \\
		&= 3 \cdot (\eps+\delta) \cdot \max\set{\deg{u},\deg{v}}.
	\end{align*}
	This in turn gives a much stronger bound (compared to~\Cref{eq:loose-bound}) that, 
	\begin{align*}
		\card{\low{u} \cap \low{v}} &= \deg{u} - \card{\low{u} - \low{v}} \\
		&\geq \deg{u} - 3 \cdot (\eps+\delta) \cdot \max\set{\deg{u},\deg{v}}. 
	\end{align*}
	At the same time, we also know that 
	\begin{align*}
		\deg{u} &\geq (\frac{1}{1+2\eps+2\delta}) \cdot \deg{w} \tag{by~\Cref{eq:degrees}} \\
			&\geq (\frac{1-7\eps-9\delta}{1+2\eps+2\delta}) \cdot \max\set{\deg{u},\deg{v}} \tag{by~\Cref{eq:loose-bound2}} \\
			&\geq (1-9\eps-11\delta) \cdot \deg{v} \\
			&> (3/4) \cdot \deg{v},
	\end{align*}
	for $\eps,\delta < 1/80$. Combining the above two equations, we get that 
	\[
		\card{\low{u} \cap \low{v}} \geq\deg{u} - 3 \cdot (\eps+\delta) \cdot (\nicefrac{4}{3}) \cdot \deg{u} = (1-4\eps-4\delta)) \cdot \deg{u} \geq (1-6\eps-6\delta)) \cdot \deg{v},
	\]
	as $\deg{u} \leq (1+2\eps+2\delta) \cdot \deg{v}$ in this property. Given that $\low{u} \subseteq N(u)$, we have, 
	\[
		\card{N({u}) \cap \low{v}} \geq  (1-6\eps-6\delta) \cdot \deg{v} \quad \text{and} \quad \deg{u} \leq (1+2\eps+2\delta) \cdot \deg{v}.
	\]
	Thus, by Rule (1), $u$ should also belong to $C_v$, concluding the proof. 
\end{proof}

We now use this property to argue that the collection $\candid$ forms a \emph{laminar} family: for any two vertices $u$ and $v$, either $C_u$ and $C_v$ do not intersect at all or one of them is a subset of the other one.

\begin{property}[$\bigstar$]\label{p:7}
	If $C_u \cap C_v \neq \emptyset$ and $\deg{u} \leq \deg{v}$, then $C_u \subseteq C_v$. 
\end{property}

\subsection{Existence of the Decomposition} 

The established properties in previous parts are enough to prove the \emph{existence} of the decomposition. For our~\Cref{thm:decomposition}, we need a \emph{constructive} version of the decomposition (via the information provided in the theorem statement); this will be done in the subsequent section which also require some further relaxing of the obtained bounds. Nevertheless, we present the following theorem 
on the existence of the decomposition in its full generality as a standalone result as it is of its own independent interest. 

\begin{theorem}[Sparse-Dense Decomposition -- existential version]\label{thm:decomposition-structural}
	For every sufficiently small $\eps,\delta>0$, vertices of any given graph $G=(V,E)$ can be partitioned into the following sets: 
\begin{itemize}

\item \textbf{Light vertices} $\Vlight$: Any  $v\in \Vlight$ has at least $\delta \cdot \deg{v}$ neighbors $u$ such that: 
\[
	\card{N(u) - N(v)} \geq \frac{\eps}{(1+\eps)} \cdot \max\set{\deg{u},\deg{v}}. 
\]
\item \textbf{Low-sparse vertices} $\Vlowsparse$: Any  $v\in \Vlowsparse$ has at least $\delta \cdot \deg{v}$ neighbors $u$ such that: 
\[
	\card{N(v) - N(u)} \geq \frac{\eps}{(1+\eps)} \cdot \max\set{\deg{u},\deg{v}}. 
\]
\item \textbf{Dense vertices partitioned into \underline{almost-cliques} $K_1,\ldots,K_k$}: For every $i \in [k]$, each $K_i$ has the following properties. Let 
 $\Delta(K_{i})$ be the maximum degree (in $G$) of the vertices in $K_i$, then: 
\begin{enumerate}[label=$\roman*)$.]
\item  Every vertex $v\in K_{i}$ has at most $(10\eps+10\delta) \cdot \Delta(K_{i})$ \emph{non-neighbors} inside ${K_{i}}$;
\item Every vertex $v\in K_{i}$ has at most $(9\eps+11\delta)\cdot \Delta(K_{i})$ \emph{neighbors} outside ${K_{i}}$;
\item Size of each $K_i$ satisfies $(1-2\eps-4\delta)\cdot \Delta(K_{i}) \leq |K_{i}| \leq (1+3\eps+3\delta)\cdot \Delta(K_{i})$.
\end{enumerate}
\end{itemize}  
\end{theorem}
\begin{proof}
	The decomposition is constructed as follows: 
	\begin{itemize}
		\item Compute the collection of candidate sets $\candid$. Recall that by~\Cref{p:7}, this  is a laminar collection. We will return all the roots of this laminar collection as the almost-cliques $K_1,\ldots,K_k$. 
		By~\Cref{p:C-almost-clique}, each $K_i$ satisfies the almost-clique property of the theorem statement, and by laminarity, these almost-cliques are disjoint. 
		
		\item By~\Cref{p:5}, every $(\eps,\delta)$-dense vertex $v$ belongs to $C_v$ and thus by the previous part, it belongs to one almost-clique. Thus, all remaining vertices at this point are $(\eps,\delta)$-light or $(\eps,\delta)$-sparse. 
		We can then partition them accordingly as required by the theorem statement; the corresponding bounds now follow from~\Cref{p:light} and~\Cref{p:isolated} for light and low-sparse vertices, respectively.  
	\end{itemize}
	This concludes the proof. 
\end{proof}

\subsection{An Algorithm for Recovering the Decomposition} 
\label{subsec:sample-decompose}

We now design an algorithm that for any graph $G=(V,E)$, given only the information specified in~\Cref{thm:decomposition}, can find our desired decomposition of $G$. As a quick reminder, 
the information provided to our algorithm is $(i)$ the degrees of all vertices, $(ii)$ $O(\eps^{-2} \cdot \log{n})$ random neighbors $\NS{v}$ for each vertex $v \in V$, and $(iii)$ a set $\sample$ of vertices where in each vertex $v$ is included 
with probability $O(\eps^{-2} \cdot \log{n}/\deg{v})$ together with $N(v)$ for $v \in \sample$ (see~\Cref{thm:decomposition} for more details). 


The algorithm for finding the decomposition consists of two steps: $(i)$ (approximately) identifying dense vertices in $\sample$ and $(ii)$ forming almost-cliques for the identified dense vertices in a way that it allows 
for the required decomposition of the entire graph. We start with the first step which is also the easier of the two. At the end, we also analyze the runtime of the algorithm which completes the proof of~\Cref{thm:decomposition}. 

\subsubsection*{Part (I): (Approximately) Identifying Dense Vertices} 

\noindent
In this part, we design a tester for approximately identifying the dense vertices in $\sample$. 

\begin{lemma}\label{lem:classify}
	There is an algorithm that given the samples specified in~\Cref{thm:decomposition} outputs a set $D$ in $O(\eps^2 \cdot n\log{n})$ time such that with high probability: 
	\begin{enumerate}[label=$\roman*).$]
		\item every $(\eps,\eps)$-dense vertex is in $D$;
		\item every vertex in $D$ is $(7\eps,2\eps)$-dense vertices. 
	\end{enumerate}
\end{lemma}
\begin{proof}
Recall that a dense vertex is a one which is not light nor low-sparse. Thus, to prove~\Cref{lem:classify}, we need to be able to rule out vertices which are light or low-sparse. We do each part in the following. 

\paragraph{Checking if $v \in \sample$ is $(\eps,\eps)$-light.} The check is quite easy: we know $N(v)$ and $\deg{u}$ for each $u \in N(v)$. The tester can easily compute $\low{v}$ by checking 
degrees of vertices in $N(v)$ and output $v$ is light iff $\card{\low{v}} < (1-\eps) \cdot \deg{v}$.

\paragraph{Checking if $v \in \sample$ is $(\eps,\eps)$-low-sparse (approximately).} On the other hand, checking whether a vertex $v$ is low-sparse or not is not that easy. Being low-sparse not only depends on $N(v)$, but rather $N(u)$ for each $u \in N(v)$,
an information that is not provided to our algorithm. As such, we need allow for approximation in this step. Moreover, even with approximation, our tester only works for vertices which are already ruled out as being $(\eps,\eps)$-light
by the previous part. 

 The tester we use  is as follows (we emphasize that in the definition of the tester and the proofs below, we need to frequently switch and mix-and-match different values for parameter $\eps$ of low-sparse vertices):

\begin{tbox}
\begin{itemize}
	\item For any vertex $u \in \Low{v}{7\eps}$, return `$u$ seems isolated for $v$' if 
	\[
		\deg{u} < (1-2\eps) \cdot \deg{v} \quad \text{or} \quad \card{\Low{v}{7\eps} \cap \NS{u}} < (1-4\eps) \cdot t.
	\]
	Return `$v$ seems low-sparse' if at least $2\eps \cdot \deg{v}$ vertices seem isolated for $v$. 
\end{itemize}
\end{tbox}
We now analyze this tester in the next two claims. In the following, let $I(v)$ denote the set of vertices in $\low{v}$ that seem isolated for $v$ by the tester. 
The following two claims are simple corollaries of definitions of low-sparse vertices, the threshold chose in the algorithm above, and Chernoff bound.

\begin{claim}[$\bigstar$]\label{clm:should-in} 
Any vertex $u$ in $\Low{v}{\eps} - \isolated{v}$ \underline{will not be} included in $I(v)$ with high probability.
\end{claim}

\begin{claim}[$\bigstar$]\label{clm:should-out} 
Any vertex $u$ in $\Low{v}{7\eps} \cap \Isolated{v}{7\eps}$ \underline{will be} included in $I(v)$ with high probability.
\end{claim}

Let $S$ denote the vertices in $\sample$ that seem low-sparse by our tester. We have, 

\begin{itemize}
	\item Let $v$ be a $(\eps,\eps)$-dense vertex. By~\Cref{clm:should-in}, no vertex from $\low{v} - \isolated{v}$ will be included in $I(v)$ with high probability. Thus, 
	\begin{align*}
		\card{I(v)} &\leq \card{\Low{v}{7\eps} \cap \isolated{v}} \\
		&\leq \card{\Low{v}{7\eps} - \low{v}} + \card{\low{v} \cap \isolated{v}} < \eps \cdot \deg{v} + \eps \cdot \deg{v},
	\end{align*}
	where the first term of last inequality is because $v$ is not $(\eps,\eps)$-light (thus $\card{\low{v}} > (1-\eps) \cdot \deg{v}$), and the second term is because $v$ is not $(\eps,\eps)$-sparse. 
	Thus at most $2\eps\cdot\deg{v}$ can be included in $I(v)$ and hence $v$ itself will not seem low-sparse to the tester. 
	
	Consequently, $S$ does not include any $(\eps,\eps)$-dense vertex. 
	
	\item Conversely, let $v$ be a $(7\eps,2\eps)$-sparse vertex. By~\Cref{clm:should-out}, all vertices in $\Low{v}{7\eps} \cap \Isolated{v}{7\eps}$ will be included in $I(v)$ with high probability. Thus, 
	\begin{align*}
		\card{I(v)} \geq \card{\Low{v}{7\eps} \cap \Isolated{v}{7\eps}} \geq 2\eps \cdot \deg{v}, 
	\end{align*}
	by~\Cref{def:isolated}. Thus, at least $2\eps \cdot \deg{v}$ will be included in $I(v)$ and hence $v$ will seem low-sparse to the tester. 
	
	Consequently, $S$ includes all $(7\eps,2\eps)$-sparse vertices. 
	
\end{itemize}

\paragraph{Final tester for dense vertices.} Finally, the tester for dense vertices will be as follows. For any vertex $v \in \sample$, we first run the tester for light vertices and ignore $v$ if the tester returns it is light. Then, 
we run the tester for low-sparse vertices to get the set $S$ of vertices that seem low-sparse. We ignore vertices of $S$ also and let $D$ be the remaining vertices. As such, with high probability, 
\begin{itemize}
	\item Every $(\eps,\eps)$-dense vertices will be included in $D$; 
	\item No vertex which is $(7\eps,2\eps)$-low-sparse or $(\eps,\eps)$-light can be in $D$, thus  vertices in $D$ are $(7\eps,2\eps)$-dense. 
\end{itemize}
This concludes the proof of~\Cref{lem:classify} (the bound on the runtime follows trivially from the testers for light and low-sparse vertices). \Qed{\Cref{lem:classify}}

\end{proof}

\subsubsection*{Part (II): Forming Almost-Cliques} 

We now consider the main part of the argument which is on forming the almost-cliques in the decomposition. Throughout this part, we define: 
\[
	\eps' := 7\eps \quad \text{,} \quad \delta' := 4\eps \quad \text{and} \quad \eps'' := \frac{\eps}{7} \quad \text{,} \quad \delta'' := \frac{\eps}{4}. 
\]
Our goal is to form $(\eps',\delta')$-candidate sets $C_v$ for vertices $v \in D$ and then pick a proper subset of them as our almost-cliques. Since every vertex in $D$ is $(7\eps,2\eps)$-dense, they will all be $(\eps',\delta')$-dense as well and thus we can hope to form the required candidate sets. We start by designing an algorithm for  the candidate tests. 

\paragraph{Forming $(\eps',\delta')$-candidate sets $C_v$ for \underline{all} $v \in V$.} To avoid dependency issues, and for the next part of the analysis, we design and analyze an algorithm that given any $(\eps',\delta')$-dense vertex $v$, find a correct choice of $C_v$ for $v$ with high probability. We emphasize that will only be able to run the algorithm for $v \in \sample$ (for which we know $N(v)$ entirely), but the algorithm and resulting candidate sets are correct for all  $v \in V$. 
\begin{tbox}
\begin{itemize}
	\item Define $\tC{v}$ as the set of all vertices $u \in V$ that satisfy the following properties: 
	\[
		\card{\NS{u} \cap \Low{v}{\eps'}} \geq (1-6\eps'-6\delta' - \eps) \cdot t \cdot \frac{\deg{v}}{\deg{u}} \quad \text{and} \quad \deg{u} \leq (1+2\eps'+2\delta') \cdot \deg{v}.
	\]
	Return $\tC{v}$ as a choice of candidate set $C_v$. 
\end{itemize}
\end{tbox}

The following lemma establishes the correctness of the algorithm. The proof is a simple application of Chernoff bound. 

\begin{lemma}[$\bigstar$]\label{lem:form-Cv}
	Let $v$ be any $(\eps',\delta')$-dense vertex in $V$. Then, with high probability, 
	\begin{enumerate}[label=$(\roman*)$]
		\item Every vertex $u \in V$ satisfying the following  is included in $\tC{v}$: 
		\[
			\card{N({u}) \cap \Low{v}{\eps'}} \geq (1-6\eps'-6\delta' ) \cdot \deg{v} \quad \text{and} \quad \deg{u} \leq (1+2\eps'+2\delta') \cdot \deg{v};
		\]
		\item No vertex $u \in V$ satisfying the following  is included in $\tC{v}$:
		\[
			\card{N({u}) \cap \Low{v}{\eps'}} < (1-7\eps'-7\delta' ) \cdot \deg{v} \quad \text{and} \quad \deg{u} > (1+2\eps'+2\delta') \cdot \deg{v}.
		\]
	\end{enumerate}
	Thus, $\tC{v}$ is a valid choice of $(\eps',\delta')$-candidate set $C_v$ by~\Cref{def:candidate-sets}.  
\end{lemma}

Finally, we can get to forming the desired almost-cliques. 

\paragraph{Forming almost-cliques.}  
A careful reader may have noticed that up until this part of argument, we never used the randomness in the choice of $\sample$. We will do that in this part. 
Consider the collection of candidate sets $\CC = \Candid{G}{\eps'}{\delta'}$ and recall that as we proved earlier, this is a laminar collection. Our strategy earlier in the proof of~\Cref{thm:decomposition-structural} was to return the root sets $C_v$ of the collection $\CC$. However, we will only be able to do so if 
we have sampled $v \in \sample$ (and it further makes its way to $v \in D$), a guarantee that cannot hold in general. 

Consequently, we use a different strategy in this part by further relaxing of our requirements. Let $S$ denote the set of all $(\eps'',\delta'')$-dense vertices in $V$ (note that given $\eps'' < \eps < \eps'$ and $\delta'' < \delta' < \delta$, these are 
in a sense the ``densest'' vertices we consider in the graph). Our goal is to find a collection of almost-cliques that cover all vertices in $S$; so, all remaining vertices will be sufficiently ``not dense'' for us to place them outside almost-cliques. 
The key step here is to prove that such a dense vertex belongs to the candidate set of ``many'' vertices in the graph, thus, it is still likely for us to sample one of those vertices at least in $\sample$ and thus include the dense vertex in the corresponding 
almost-clique also. This is formalized in the following lemma. 

\begin{lemma}[$\bigstar$]\label{lem:dense-finding-1}
	Suppose $v$ is an $(\eps'',\delta'')$-dense vertex. Then, $v \in C_u$ for every $u \in \Kernel{v}{\eps''}{\delta''}$ (where $C_v$ is also computed as a $(\eps',\delta')$-candidate set of $u$ by the algorithm in the previous part). 
\end{lemma}

We also prove that for each $(\eps'',\delta'')$-dense vertex, there is at least one vertex from $\Kernel{v}{\eps''}{\delta''}$ that is sampled in $\sample$ with high probability. 
This is because size of $\Kernel{v}{\eps''}{\delta''}$ is $\Omega(\deg{v})$ and all those vertices have degree $O(\deg{v})$; as we sample each vertex proportional to its degree, with high probability, 
we sample at least one vertex of $\Kernel{v}{\eps''}{\delta''}$. 

\begin{lemma}[$\bigstar$]\label{lem:dense-finding-2}
	With high probability, for every $(\eps'',\delta'')$-dense vertex $v$, there is at least one vertex $u \in \Kernel{v}{\eps''}{\delta''}$ that is sampled in $\sample$. 
\end{lemma}

The algorithm for forming the almost-cliques is then as follows: 
\begin{tbox}
\begin{itemize}
	\item For every vertex $v \in D$, form its candidate set $C_v$. The collection $\set{C_v}_{v \in D}$ forms a laminar family by~\Cref{p:7}. Pick all roots of this collection as our almost-cliques $K_1,\ldots,K_k$ (note that these almost-cliques
	contain vertices not in $D$). 
\end{itemize}
\end{tbox}

The following two lemmas establish the correctness of this part. 

\begin{lemma}\label{lem:almost-clique-good}
	Conditioned on  the high probability events of~\Cref{lem:classify,lem:form-Cv,lem:dense-finding-2}, the collection of almost-cliques $K_1,\ldots,K_k$ satisfy the almost-clique properties of~\Cref{thm:decomposition} for some parameter $\Theta(\eps)$. 
\end{lemma}
\begin{proof}
	Any vertex in $D$ is $(\eps',\delta')$-dense by~\Cref{lem:classify}, thus by~\Cref{p:C-almost-clique}, each set $C_v$ satisfies the properties of almost-cliques for parameters $\eps' = 7\eps$ and $\delta'=4\delta$. Plugging in these bounds ensure that each 
	$K_i$ is individually an almost-clique. Moreover, since the collection $\set{C_v}_{v \in D}$ is laminar and we are picking root sets of this collection, the resulting almost-cliques will be disjoint. Re-scaling $\eps$ by a constant factor, finalizes the proof. 
\end{proof}

\begin{lemma}\label{lem:sparse-good}
	Conditioned on  the high probability events of~\Cref{lem:classify,lem:form-Cv,lem:dense-finding-2}, any vertex $v$ not in $K_1 \cup \ldots \cup K_k$ satisfies the sparse vertex properties of~\Cref{thm:decomposition} for some parameter $\Theta(\eps)$. 
\end{lemma}
\begin{proof}
	Let $v$ be any $(\eps'',\delta'')$-dense vertex. By~\Cref{lem:dense-finding-1}, $v$ belongs to $C_u$ for all $u \in \Kernel{v}{\eps''}{\delta''}$. By~\Cref{lem:dense-finding-2}, at least one vertex $u \in \Kernel{v}{\eps''}{\delta''}$ is sampled in $\sample$. 
	By~\Cref{p:kernel-dense}, the vertex $u$ is also $(\eps,\delta)$-dense. By~\Cref{lem:classify}, $u$ should also belong to $D$. Finally, since we are returning roots of the laminar family $\set{C_w}_{w \in D}$, all vertices in $C_u$ will belong to a 
	single almost-clique. Thus, $v$ will also be included in one of almost-cliques $K_1,\ldots,K_k$. 
	
	As such, any remaining vertex is either $(\eps'',\delta'')$-light or $(\eps'',\delta'')$-low-sparse. In both cases, by~\Cref{p:light} and~\Cref{p:isolated}, $v$ satisfies the required properties of sparse vertices in~\Cref{thm:decomposition}. 
	Re-scaling $\eps$ by a constant factor, finalizes the proof. 
\end{proof}

We are almost done with the proof of \Cref{thm:decomposition} as \Cref{lem:almost-clique-good,lem:sparse-good} ensure that the output of our algorithm  with high probability satisfies the desired decomposition. 
The very last step is to analyze the runtime of the recovery algorithm also which is done in the next part. 

\subsubsection*{Sample Size and Runtime Analysis of the Recovery Algorithm}

The first step to analyze the runtime of the algorithm is to bound the size of its input, namely, the total size of edge-samples and edges incident on vertex-samples. This will also be helpful in subsequent sections 
when designing our sublinear algorithms. It turns out that the proof of this lemma is not entirely trivial as direct applications of Chernoff bound seem to not achieve (asymptotically) optimal bounds (our proof uses Bernstein's inequality instead). 

\begin{lemma}\label{lem:sample-helper-lem}
	Define
	\[
		\Esample := \set{e = (u,v) \mid \text{$u \in \NS{v}$ for $v \in V$ or $u \in N(v)$ for $v \in \sample$}}
	\]
	as the set of edges given to the algorithm in~\Cref{thm:decomposition}. With high probability, $\card{\Esample} = O(\eps^{-2} \, n\log{n})$. 
\end{lemma}
\begin{proof}
	The total number of edge-samples is bounded by $O(\eps^{-2}\,n\log{n})$ deterministically. We thus only need to focus on bounding the contribution of edges for vertices $v \in \sample$ for which we store all of $N(v)$. 
	Moreover, in the following, we can focus on all vertices $v \in V$ whose degree is at least $2 \, ({c \cdot \log{n}})$ as the total number of edges incident on all other vertices is $O(n\log{n})$ at most. 
	
	Define the random variable $X_v$ for each $v \in V$ to be $\deg{v}$ if $v \in \sample$ and $0$ otherwise. Let $X:= \sum_{v \in V} X_v$ denote the total number of edges we store for the sampled vertices. We only need to bound $X$ then. 
	Given that each vertex $v$ belongs to $\sample$ with probability $p_v = ({c \cdot \log{n}})/{\deg{v}}$, we have, 
	\[
		\expect{X} = \sum_{v \in V} \expect{X_v} = \sum_{v \in V} p_v \cdot \deg{v} = (c \cdot  n \cdot \log{n}). 
	\]
	We need to prove the concentration of $X$ which is sum of independent random variables $\set{X_v \mid v \in V}$. 
	But given that the range of these variables is $[0,\deg{v}]$, a direct application of Chernoff-Hoeffding bounds is not good enough (for very dense graphs). 
	Instead, we use Bernstein's inequality (\Cref{prop:bernstein}). 
	
	Define zero-mean random variables $Y_v := X_v - \expect{X_v}$ for each vertex $v \in V$ and let $Y := \sum_{v \in V} Y_v$. So we have $Y = X - (c \cdot  n \cdot \log{n})$ and bounding $Y$ will allow us to 
	bound $X$ as well. We have, 
	\begin{align*}
		\expect{Y_v^2} &= \expect{(X_v-\expect{X_v})^2} \\
		&= \expect{\Paren{X_v-(c \cdot  n \cdot \log{n})}^2} \tag{by the calculation above} \\
		&\leq \expect{X_v}^2 \tag{by our assumption on $\deg{v} \geq 2 \, ({c \cdot \log{n}})$} \\
		&= p_v \cdot \deg{v}^2 \tag{$X_v$ is $\deg{v}$ w.p. $p_v$ and otherwise $0$} \\
		&= \deg{v} \cdot (c \cdot \log{n}).
	\end{align*}
	Moreover, for each $v$, we have $\card{Y_v} \leq \deg{v} < n$ and there are (at most) $n$ variables involved. 
	Plugging in these bounds in Bernstein's inequality (\Cref{prop:bernstein})  implies that, 
	\begin{align*}
		\Pr\paren{X > (5c \cdot n \cdot \log{n})} &= \Pr\paren{Y > (4c \cdot n \cdot \log{n})} \\
		&\leq \exp\paren{-\frac{(16 c^2 \cdot n^2 \cdot \log^2{n})}{2 \sum_{v \in V} \deg{v} \cdot (c \cdot \log{n}) + \nicefrac23 \cdot n^2}} \\
		&\leq  \exp\paren{-\frac{(16 c^2 \cdot n^2 \cdot \log^2{n})}{4 n^2 \cdot (c \cdot \log{n})}} \tag{as $\sum_{v \in V} \deg{v} \leq n^2$ and a loose upper bound for second term} \\
		&\leq \exp\paren{-{4c \cdot \log{n}}{}} < n^{-4c}.
	\end{align*}
	This concludes the proof. 
\end{proof}

Let us now analyze the runtime of each part of the algorithm assuming the high probability event of~\Cref{lem:sample-helper-lem} and all the ones conditioned on in the previous part happen. 

\paragraph{Approximately identifying dense vertices.} This part involves the following two components: 
\begin{enumerate}[label=$\roman*).$]
\item The tester for $(\eps,\eps)$-light vertices in $\sample$: this is done simply in $O(\card{\sample}) = O(n)$ time by checking degree of each vertex in $O(1)$ time. 
\item The (approximate) tester for $(\eps,\eps)$-dense vs $(7\eps,2\eps)$-sparse vertices in $\sample$: Testing each vertex $v \in V$ requires spending $O(\deg{v})$ time to visit each of its neighbors, and for each neighbor, $O(t) = O(\eps^{-2}\log{n})$
time to check intersection of its neighborhood with $\Low{v}{7\eps}$ (both sets can be sorted in linear-time using counting sort as  we can rename the vertices to $\set{1,\ldots,n}$ and thus their intersection can be computed in $O(t)$ time). 
This step takes $O(\sum_{v \in \sample} \deg{v} \cdot t) = O(\eps^{-2} \cdot n\log^2{n})$ time. 
\end{enumerate}

\paragraph{Forming almost-cliques.} This part also involves two components: 
\begin{enumerate}[label=$\roman*).$]
\item Forming $(\eps',\delta')$-candidate sets $C_v$ for $v \in \sample$: Recall that even though we defined the candidate sets for all $v \in V$, that was only for the analysis and the algorithm only needs (and can) compute these sets for 
vertices of $\sample$. This step involves including vertices $u \in V$ with large intersection of $\NS{u}$ and $\Low{v}{\eps'}$. This can be done efficiently as follows. 

The total number of edges (in $G$) going out of $\Low{v}{\eps'}$ is $O(\deg{v}^2)$
by~\Cref{p:dense}. On the other hand, since any vertex $u$ that we check for $v$ has to have $\deg{u} = \Omega(\deg{v})$ and we sample $t = \Theta(\eps^{-2} \cdot \log{n})$ edges uniformly at random from its neighborhood to 
get $\NS{u}$, we get that the total number of edge-samples going out of $\Low{v}{\eps'}$ is $O(\deg{v} \cdot t)$ with high probability. As constructing $C_v$ takes time proportional to visiting all these edges, this takes $O(\deg{v} \cdot t)$ time for vertex $v$. 
Thus, the overall runtime of this step is $O(\sum_{v \in \sample} \deg{v} \cdot t) = O(\eps^{-2} \cdot n\log^2{n})$. 

\item Forming almost-cliques: Given that the candidate sets form a laminar collection, we can find all their roots in $O(n)$ time as follows. Originally, for each vertex $v$ write the name of candidate sets $C_u$ that $v$ belongs to. Moreover, for each 
candidate set $C_u$ write down its size. Then, go over vertices: for each vertex $v$, pick the largest candidate set $C_w$ that contains $v$ (breaking the ties arbitrarily); since $C_u$'s form a laminar collection, this candidate set $C_w$ would be a root of the collection. Thus, we can include $C_w$ as an almost-clique and remove all its vertices from $V$. Then, continue like this until all vertices are processed. Finally, all remaining vertices are output as part of sparse vertices. This step takes $O(n)$ time this way. 
\end{enumerate}
\noindent
To conclude, the  runtime of the algorithm is $O(\eps^{-2} \cdot n\log^2{n})$ time. This concludes the proof of~\Cref{thm:decomposition}.


\newcommand{\charge}[2]{\ensuremath{\textnormal{\textsf{charge}}(#1,#2)}\xspace}

\newcommand{\EE}{\mathcal{E}}

\newcommand{\constSp}{\ensuremath{C_{1}}}

\newcommand{\uncharge}[2]{\ensuremath{\textnormal{\textsf{uneven-charge}}(#1,#2)}\xspace}
\newcommand{\unHcharge}[2]{\ensuremath{\textnormal{\textsf{uneven-H-charge}}(#1,#2)}\xspace}
\newcommand{\unScharge}[2]{\ensuremath{\textnormal{\textsf{uneven-S-charge}}(#1,#2)}\xspace}

\newcommand{\spcharge}[2]{\ensuremath{\textnormal{\textsf{sparse-charge}}(#1,#2)}\xspace}
\newcommand{\spHcharge}[2]{\ensuremath{\textnormal{\textsf{sparse-H'-charge}}(#1,#2)}\xspace}
\newcommand{\spScharge}[2]{\ensuremath{\textnormal{\textsf{sparse-S-charge}}(#1,#2)}\xspace}

\newcommand{\ccharge}[2]{\ensuremath{\textnormal{\textsf{clique-charge}}(#1,#2)}\xspace}
\newcommand{\occharge}[2]{\ensuremath{\textnormal{\textsf{outside-clique-charge}}(#1,#2)}\xspace}
\newcommand{\iccharge}[2]{\ensuremath{\textnormal{\textsf{inside-clique-charge}}(#1,#2)}\xspace}

\newcommand{\cscharge}[2]{\ensuremath{\textnormal{\textsf{clique-sep-charge}}(#1,#2)}\xspace}
\newcommand{\clcharge}[2]{\ensuremath{\textnormal{\textsf{clique-charge}}(#1,#2)}\xspace}

\newcommand{\chargeset}[1]{\ensuremath{\textnormal{{ChargeSet}}(#1)}\xspace}

\newcommand{\degpiK}[2]{\textnormal{indeg}^{\!+}(#1,K_#2)}
\newcommand{\NpiK}[2]{N^{+}_{\textnormal{in}}(#1,K_#2)}
\newcommand{\EpiK}[2]{E^{+}_{\textnormal{in}}(#1,K_#2)}

\newcommand{\degpoK}[2]{\textnormal{outdeg}^{\!+}(#1,K_#2)}
\newcommand{\NpoK}[2]{N^{+}_{\textnormal{out}}(#1,K_#2)}
\newcommand{\EpoK}[2]{E^{+}_{\textnormal{out}}(#1,K_#2)}

\newcommand{\degniK}[2]{\textnormal{indeg}^{\!-}(#1,K_#2)}
\newcommand{\NniK}[2]{N^{-}_{\textnormal{in}}(#1,K_#2)}
\newcommand{\EniK}[2]{E^{-}_{\textnormal{in}}(#1,K_#2)}

\newcommand{\degnoK}[2]{\textnormal{outdeg}^{\!+}(#1,K_#2)}
\newcommand{\NnoK}[2]{N^{-}_{\textnormal{out}}(#1,K_#2)}
\newcommand{\EnoK}[2]{E^{-}_{\textnormal{out}}(#1,K_#2)}

\section{Correlation Clustering via the Sparse-Dense Decomposition}\label{sec:cc-approx}

We are now ready to present a correlation clustering scheme based on the decomposition results in \Cref{sec:decomposition}, applied to the underling $G^+$ graph. Our approach is to 
simply place the sparse vertices in separate singleton clusters and treat each almost-clique of the dense vertices as one disjoint cluster. Formally, 

\begin{theorem}
\label{thm:cc-alg}
	Suppose $G=(V,E)$ is any labeled graph and $V = \Vsparse \sqcup K_1 \sqcup \ldots \sqcup K_k$ is an $\eps$-sparse-dense decomposition of $G^+$ for  $\eps > 0$ according to~\Cref{thm:decomposition}. 
	Let $\ALG$ be the following clustering: 
	\begin{itemize}
		\item Any vertex $v \in \Vsparse$ is placed in a singleton cluster, i.e., $\ALG(v) = \set{v}$; 
		\item Any almost-clique $K_i$ forms a separate cluster, i.e., for any $v \in K_i$, $\ALG(v) = \set{u \mid u \in K_i}$. 
	\end{itemize}
	Then, $\ALG$ is an $O(\eps^{-2})$-approximation correlation clustering of $G$. 
\end{theorem}

The intuition behind the proof of~\Cref{thm:cc-alg} is simple: the $(+)$-neighborhood of sparse vertices is so different from that of their neighbors that no matter how we cluster them, we will need to pay a cost proportional to their degree; so we might as well cluster them individually. On the other hand, the almost-cliques are so tightly connected to each other and so loosely connected to outside by their $(+)$-edges that they simply form the best cluster possible themselves; so we cluster them that way also. 

We formalize this intuition in this section. Our analysis of~\Cref{thm:cc-alg} is inspired by the recent work of~\cite{CohenAddadLMNP21}. The main difference is in using the decomposition of~\Cref{thm:decomposition} instead of the rather ad-hoc and ``multi-step'' partitioning in~\cite{CohenAddadLMNP21} which is crucial for our sublinear algorithms (the decomposition allows us to also give a more modular proof by focusing on each part of the partition separately). 

Throughout this section, fix $\OPT$ to be a fixed optimal clustering of $G$ and recall that $\ALG$ denotes the clustering returned by~\Cref{thm:cc-alg}. Similar to~\cite{CohenAddadLMNP21}, we use a \textbf{charging scheme}: To any vertex $z \in V$ and any edge $f$ incident on $z$, i.e., $f \in E^+(z) \sqcup E^-(z)$, we assign a value $\charge{z}{f}$ as follows: 
\begin{tbox}
Charging scheme for the analysis of~\Cref{thm:cc-alg}. 

\smallskip

\begin{enumerate}[label=$(\roman*)$]
	\item Initially, $\charge{z}{f} = 0$ for all $z \in V$ and $f \in E(z)$; 
	\item For any edge $e \in E$ with $\cost{\ALG}{e} = 1$, we will \underline{find a collection of vertex-edge pairs}, called the \emph{charge-set} of $e$: 
	\[
		\chargeset{e} \subseteq \set{(z,f) \mid z \in V, \text{~~} f \in E(z), \text{~and~} \cost{\OPT}{f} = 1}. 
	\]
	For simplicity of notation, we define $\chargeset{e} = \emptyset$ if $\cost{\ALG}{e} = 0$. 
	\item We will then increase $\charge{z}{f}$ for all $(z,f) \in \chargeset{e}$ by $\card{\chargeset{e}}^{-1}$. 
\end{enumerate}
\end{tbox}

The main part in this charging scheme is to find proper charge-sets for all edges. The following lemma establishes our desired property of the charging scheme. 

\begin{lemma}\label{lem:charging}
	Suppose there is a choice of $\chargeset{e}$ for edges $e \in E$ in the charging scheme  such that for all $z \in V$ and $f \in E(z)$, we have $\charge{z}{f} \leq \alpha$ for some $\alpha \geq 1$.  Then, 
	\[
		\scost{\ALG} \leq 2\alpha \cdot \scost{\OPT}. 
	\]
\end{lemma}
\begin{proof}
	We have, 
	\begin{align*}
		\scost{\ALG} &= \sum_{e \in E} \cost{\ALG}{e} \tag{by the definition of the total cost in~\Cref{eq:scost}} \\
		&= \sum_{e \in E} \sum_{\substack{(z,f) \in \\ \chargeset{e}}} \card{\chargeset{e}}^{-1} \tag{the inner sum is $1$ if $\cost{\ALG}{e}=1$ and $0$ otherwise} \\
		&= \sum_{z \in V} \sum_{\substack{f \in E(z) \\ \text{~and~} \\ \cost{\OPT}{f}=1}} \charge{z}{f}  \tag{by summing over charges of all vertex-edge pairs} \\
		&\leq \sum_{z \in V} \alpha \cdot \card{f \in E(z) \text{~and~} \cost{\OPT}{f}=1} \tag{by the guarantee of the lemma statement} \\
		&= \sum_{f \in E} 2\alpha \cdot \cost{\OPT}{f} \tag{as each edge will be added twice (one by each endpoint)} \\
		&= 2\alpha \cdot \scost{\OPT}. 
	\end{align*}
	This concludes the proof. 
\end{proof}

By~\Cref{lem:charging}, we only need to find charge-sets of the given edges so that $\charge{z}{f}$ is small for all vertex-edge pairs $(z,f)$. This is done for edges of sparse and dense vertices separately in the next subsections. 

\paragraph{A helper lemma.} Before getting to the main part of the proof, we will  prove a helper lemma that simplifies our task of finding charge-sets for $(+)$-edges in the later parts of the analysis. 
Roughly speaking, it states that if we have a collection of edges whose endpoints have sufficiently different neighborhood, then we can find a charge-set for all the given edges without increasing charge of any vertex-edge pair by much. 

\begin{lemma}\label{lem:helper-charge}
Let $\theta \in (0,1)$ be a parameter and $\EE$ be any collection of edges in $E^+$ in the input labeled graph $G$ such that for all $\xi = (\alpha,\beta) \in \EE$, 
\begin{align}
	{N^{+}(\alpha) \sym N^{+}(\beta)} \geq \theta \cdot \max\set{\degp{\alpha},\degp{\beta}}. 
\end{align}
Then, there is a choice of $\chargeset{\xi}$ for all $\xi \in \EE$ such that $\charge{z}{f} = O(\theta^{-1})$ for all vertex-edge pairs $(z,f)$ in $G$. 
\end{lemma}
\begin{proof}
	We define $\chargeset{\xi}$ for any $\xi \in \EE$ as follows: 
\begin{itemize}
	\item Type-$1$ charges: when $\cost{\OPT}{\xi} = 1$. In this case, we simply set $\chargeset{\xi} = \set{(\alpha,\xi)}$ itself.
	\item Type-$2$ charges: when $\cost{\OPT}{\xi} = 0$. This is the more challenging case. Note that in this case, we have that $\OPT(\alpha) = \OPT(\beta)$, and let us denote this cluster as $\OPT_{\alpha\beta}$. Consider any vertex 
	$w \in {N^{+}(\alpha) \sym N^{+}(\beta)}$: 
	\begin{itemize}
	\item \underline{Case A}: $w \in N^{+}(\alpha)$ and $w \in N^{-}(\beta)$. In this case, there is $\cost{\OPT}{(w, \beta)} = 1$ if $\OPT(w) = \OPT_{\alpha\beta}$, and $\cost{\OPT}{(w, \alpha)} = 1$ if $\OPT(w)\neq \OPT_{\alpha\beta}$.
	\item \underline{Case B}: $w \in N^{+}(\beta)$ and $w \in N^{-}(\alpha)$. In this case, there is $\cost{\OPT}{(w, \alpha)} = 1$ if $\OPT(w) = \OPT_{\alpha\beta}$, and $\cost{\OPT}{(w, \beta)} = 1$ if $\OPT(w)\neq \OPT_{\alpha\beta}$.
	\end{itemize}
	Therefore, in both cases, there is exactly one edge $f(w) \in \set{(w,\alpha),(w,\beta)}$ such that $\cost{\OPT}{f(w)} = 1$. Let $z(w) \in \set{\alpha,\beta}$ be the vertex other than $w$ incident on $f(w)$. We add all pairs $(z(w),f(w))$ to $\chargeset{\xi}$, i.e., 
	\[
		\chargeset{\xi} = \set{(z(w),f(w)) \mid w \in N^+(\alpha) \sym N^+(\beta)}. 
	\]
	Given the  bound on the size of $N^+(\alpha) \sym N^+(\beta)$, we have that $\card{\chargeset{\xi}} \geq  \theta \cdot \max\set{\degp{\alpha},\degp{\beta}}$. 
\end{itemize}

An illustration of the type-$2$ charges can be found in \Cref{fig:type-2-charging-illus}.

Let us now bound the distributed charges. We have three different choices for $(z,f)$ that can belong to $\chargeset{\xi}$ 
for some edge $\xi \in \EE$  as follows (a graphical exemplification can be found in \Cref{fig:bounded-charge-illus}):  

\begin{itemize}
	\item \underline{A pair $(\alpha,\xi)$ charged by a type-1 charge, where $\xi \in \EE$ and $\alpha$ is an endpoint of $\xi$:} 
	
	In this case $\charge{\alpha}{\xi} = 1$ because there is only a single edge $\xi$ that can  make such a charge. 
	
	\item \underline{A pair $(\alpha,f(w))$ charged by a type-2 charge, where $w \in {N^{+}(\alpha) \sym N^{+}(\beta)}$ and $z(w) = \alpha$:} 
	
	For any such charge, we increase $\charge{\alpha}{f(w)}$ by 
	\[
	\card{\chargeset{\xi}}^{-1} \leq (\theta \cdot \degp{\alpha})^{-1}.
	\]
	 At the same time, such a charge 
	can only be made by edges from  $\alpha$ to $\beta \in N^{+}(\alpha)$ (so that $(\alpha,\beta) \in E^+$), which are $\degp{\alpha}$ many. Thus, the total charge made in this case  leads to $\charge{\alpha}{f(w)} = O(\theta^{-1})$. 
	
	\item \underline{A pair $(\beta,f(w))$ charged by a type-2 charge, where $w \in {N^{+}(\alpha) \sym N^{+}(\beta)}$ and $z(w) = \beta$:} 
	
	For any such charge, we increase $\charge{\beta}{f(w)}$ by 
	\[
	\card{\chargeset{\xi}}^{-1} \leq (\theta \cdot \degp{\beta})^{-1}.
	\]
	 At the same time, such a charge can only be made by edges from  $\beta$ to $\alpha \in N^{+}(\beta)$ (so that $(\alpha,\beta) \in E^+$), which are $\degp{\beta}$ many. 
	 Thus, the total charge made in this case  leads to $\charge{\beta}{f(w)} = O(\theta^{-1})$. 
\end{itemize}
	
This concludes the proof of the lemma. 
\end{proof}


\begin{figure}[h!]
\centering
\tikzset{every picture/.style={line width=0.75pt}} 

\tikzset{every picture/.style={line width=0.75pt}} 

\begin{tikzpicture}[x=0.75pt,y=0.75pt,yscale=-0.9, xscale=0.9, every node/.style={scale=0.9}]

\draw   (94.28,157) .. controls (94.28,147.13) and (101.84,139.13) .. (111.16,139.13) .. controls (120.49,139.13) and (128.04,147.13) .. (128.04,157) .. controls (128.04,166.87) and (120.49,174.88) .. (111.16,174.88) .. controls (101.84,174.88) and (94.28,166.87) .. (94.28,157) -- cycle ;
\draw [line width=1.5]    (112.46,338.5) -- (111.16,174.88) ;
\draw   (371.01,118.23) .. controls (371.01,109.26) and (378.28,102) .. (387.24,102) -- (539.18,102) .. controls (548.14,102) and (555.4,109.26) .. (555.4,118.23) -- (555.4,166.9) .. controls (555.4,175.86) and (548.14,183.13) .. (539.18,183.13) -- (387.24,183.13) .. controls (378.28,183.13) and (371.01,175.86) .. (371.01,166.9) -- cycle ;
\draw   (95.58,356.38) .. controls (95.58,346.5) and (103.14,338.5) .. (112.46,338.5) .. controls (121.79,338.5) and (129.34,346.5) .. (129.34,356.38) .. controls (129.34,366.25) and (121.79,374.25) .. (112.46,374.25) .. controls (103.14,374.25) and (95.58,366.25) .. (95.58,356.38) -- cycle ;
\draw   (377.51,157) .. controls (377.51,147.13) and (385.07,139.13) .. (394.39,139.13) .. controls (403.71,139.13) and (411.27,147.13) .. (411.27,157) .. controls (411.27,166.87) and (403.71,174.88) .. (394.39,174.88) .. controls (385.07,174.88) and (377.51,166.87) .. (377.51,157) -- cycle ;
\draw [color={rgb, 255:red, 0; green, 0; blue, 255 }  ,draw opacity=1 ][line width=1.5]    (128.04,157) -- (377.51,157) ;
\draw [color={rgb, 255:red, 0; green, 0; blue, 255 }  ,draw opacity=1 ][line width=1.5]  [dash pattern={on 5.63pt off 4.5pt}]  (377.51,157) -- (112.46,338.5) ;

\draw (93.93,275.63) node [anchor=north west][inner sep=0.75pt]   [align=left] {$ $};
\draw (103.16,110.25) node [anchor=north west][inner sep=0.75pt]    {$\alpha $};
\draw (104.46,382.5) node [anchor=north west][inner sep=0.75pt]    {$\beta $};
\draw (428.09,109.57) node [anchor=north west][inner sep=0.75pt]   [align=left] {$\displaystyle {\textstyle N^{+}( \alpha ) -N^{+}( \beta )}$};
\draw (421.64,146.38) node [anchor=north west][inner sep=0.75pt]    {$w$};
\draw (84.09,230.38) node [anchor=north west][inner sep=0.75pt]    {$\xi $};
\draw (167,104) node [anchor=north west][inner sep=0.75pt]   [align=left] {cost$\displaystyle _{\mathcal{O}}$($\displaystyle ( \alpha ,w)$)=1\\if $\displaystyle \mathcal{O}$($\displaystyle w) \neq \mathcal{O}_{\alpha \beta }$ };
\draw (250.61,249.56) node [anchor=north west][inner sep=0.75pt]  [rotate=-359.73] [align=left] {cost$\displaystyle _{\mathcal{O}}$($\displaystyle ( \beta ,w)$)=1\\if $\displaystyle \mathcal{O}$($\displaystyle w) =\mathcal{O}_{\alpha \beta }$ };

\end{tikzpicture}
\caption{\label{fig:type-2-charging-illus} Illustration of the type-$2$ charge conditioning on $\alpha$ and $\beta$ are in the same cluster $\OPT_{\alpha\beta}$. For each vertex $w \in N^{+}(\alpha)-N^{+}(\beta)$, if $w$ is in $\OPT_{\alpha\beta}$, the negative edge $(\beta,w)$ induces a cost of $1$ in $\OPT$; otherwise, if $w$ is in a different cluster, the positive edge $(\alpha,w)$ induces a cost of $1$ in $\OPT$. Vertices in $N^{+}(\beta)-N^{+}(\alpha)$ works in the same way.}
\end{figure}
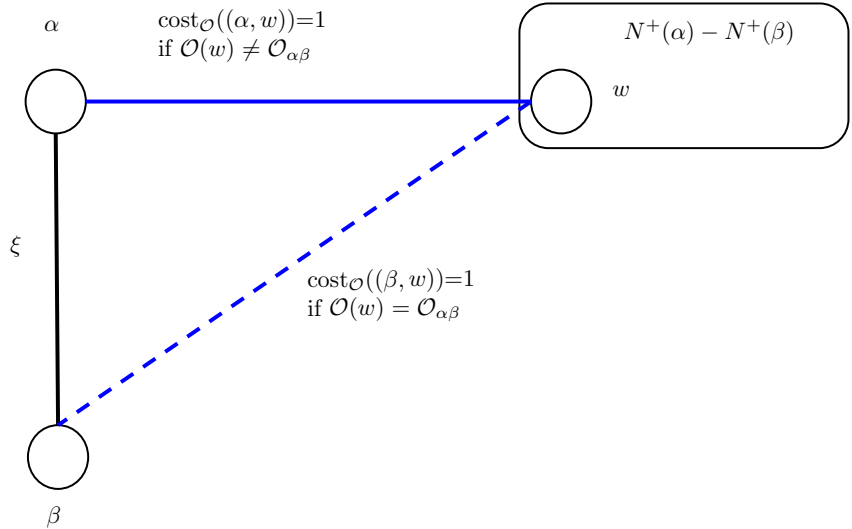

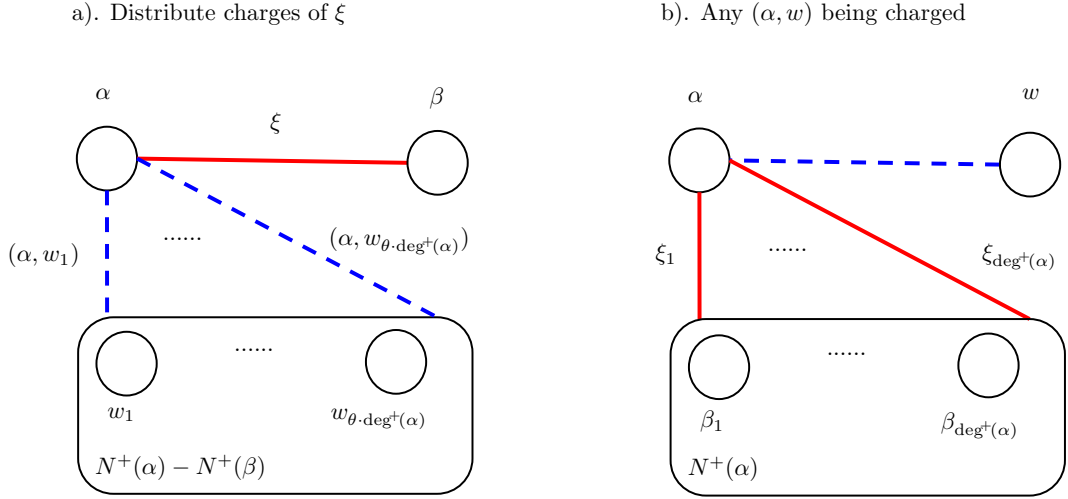
\begin{figure}[!h]
\centering
\tikzset{every picture/.style={line width=0.75pt}} 

\tikzset{every picture/.style={line width=0.75pt}} 

\begin{tikzpicture}[x=0.75pt,y=0.75pt,yscale=-0.9, xscale=0.9, every node/.style={scale=0.9}]

\draw   (54.28,157) .. controls (54.28,147.13) and (61.84,139.13) .. (71.16,139.13) .. controls (80.49,139.13) and (88.04,147.13) .. (88.04,157) .. controls (88.04,166.87) and (80.49,174.88) .. (71.16,174.88) .. controls (61.84,174.88) and (54.28,166.87) .. (54.28,157) -- cycle ;
\draw [color={rgb, 255:red, 255; green, 0; blue, 0 }  ,draw opacity=1 ][line width=1.5]    (239.58,159.38) -- (88.04,157) ;
\draw   (55.01,265.8) .. controls (55.01,254.86) and (63.88,246) .. (74.81,246) -- (256.2,246) .. controls (267.14,246) and (276,254.86) .. (276,265.8) -- (276,325.2) .. controls (276,336.14) and (267.14,345) .. (256.2,345) -- (74.81,345) .. controls (63.88,345) and (55.01,336.14) .. (55.01,325.2) -- cycle ;
\draw   (239.58,159.38) .. controls (239.58,149.5) and (247.14,141.5) .. (256.46,141.5) .. controls (265.79,141.5) and (273.34,149.5) .. (273.34,159.38) .. controls (273.34,169.25) and (265.79,177.25) .. (256.46,177.25) .. controls (247.14,177.25) and (239.58,169.25) .. (239.58,159.38) -- cycle ;
\draw [color={rgb, 255:red, 0; green, 0; blue, 255 }  ,draw opacity=1 ][line width=1.5]  [dash pattern={on 5.63pt off 4.5pt}]  (71.16,174.88) -- (71.24,246) ;
\draw [color={rgb, 255:red, 0; green, 0; blue, 255 }  ,draw opacity=1 ][line width=1.5]  [dash pattern={on 5.63pt off 4.5pt}]  (88.04,157) -- (256.2,246) ;
\draw   (65.28,272) .. controls (65.28,262.13) and (72.84,254.13) .. (82.16,254.13) .. controls (91.49,254.13) and (99.04,262.13) .. (99.04,272) .. controls (99.04,281.87) and (91.49,289.88) .. (82.16,289.88) .. controls (72.84,289.88) and (65.28,281.87) .. (65.28,272) -- cycle ;
\draw   (216.28,271) .. controls (216.28,261.13) and (223.84,253.13) .. (233.16,253.13) .. controls (242.49,253.13) and (250.04,261.13) .. (250.04,271) .. controls (250.04,280.87) and (242.49,288.88) .. (233.16,288.88) .. controls (223.84,288.88) and (216.28,280.87) .. (216.28,271) -- cycle ;
\draw   (386.28,158) .. controls (386.28,148.13) and (393.84,140.13) .. (403.16,140.13) .. controls (412.49,140.13) and (420.04,148.13) .. (420.04,158) .. controls (420.04,167.87) and (412.49,175.88) .. (403.16,175.88) .. controls (393.84,175.88) and (386.28,167.87) .. (386.28,158) -- cycle ;
\draw [color={rgb, 255:red, 0; green, 0; blue, 255 }  ,draw opacity=1 ][line width=1.5]  [dash pattern={on 5.63pt off 4.5pt}]  (571.58,160.38) -- (420.04,158) ;
\draw   (571.58,160.38) .. controls (571.58,150.5) and (579.14,142.5) .. (588.46,142.5) .. controls (597.79,142.5) and (605.34,150.5) .. (605.34,160.38) .. controls (605.34,170.25) and (597.79,178.25) .. (588.46,178.25) .. controls (579.14,178.25) and (571.58,170.25) .. (571.58,160.38) -- cycle ;
\draw [color={rgb, 255:red, 255; green, 0; blue, 0 }  ,draw opacity=1 ][line width=1.5]    (403.16,175.88) -- (403.24,247) ;
\draw [color={rgb, 255:red, 255; green, 0; blue, 0 }  ,draw opacity=1 ][line width=1.5]    (420.04,158) -- (588.2,247) ;
\draw   (387.01,266.8) .. controls (387.01,255.86) and (395.88,247) .. (406.81,247) -- (588.2,247) .. controls (599.14,247) and (608,255.86) .. (608,266.8) -- (608,326.2) .. controls (608,337.14) and (599.14,346) .. (588.2,346) -- (406.81,346) .. controls (395.88,346) and (387.01,337.14) .. (387.01,326.2) -- cycle ;
\draw   (397.28,274) .. controls (397.28,264.13) and (404.84,256.13) .. (414.16,256.13) .. controls (423.49,256.13) and (431.04,264.13) .. (431.04,274) .. controls (431.04,283.87) and (423.49,291.88) .. (414.16,291.88) .. controls (404.84,291.88) and (397.28,283.87) .. (397.28,274) -- cycle ;
\draw   (548.28,273) .. controls (548.28,263.13) and (555.84,255.13) .. (565.16,255.13) .. controls (574.49,255.13) and (582.04,263.13) .. (582.04,273) .. controls (582.04,282.87) and (574.49,290.88) .. (565.16,290.88) .. controls (555.84,290.88) and (548.28,282.87) .. (548.28,273) -- cycle ;

\draw (93.93,275.63) node [anchor=north west][inner sep=0.75pt]   [align=left] {$ $};
\draw (63.16,116.25) node [anchor=north west][inner sep=0.75pt]    {$\alpha $};
\draw (250.46,115.5) node [anchor=north west][inner sep=0.75pt]    {$\beta $};
\draw (63.09,321.57) node [anchor=north west][inner sep=0.75pt]   [align=left] {$\displaystyle {\textstyle N^{+}( \alpha ) -N^{+}( \beta )}$};
\draw (161.09,129.38) node [anchor=north west][inner sep=0.75pt]    {$\xi $};
\draw (70.04,295) node [anchor=north west][inner sep=0.75pt]    {$w_{1}$};
\draw (141,262) node [anchor=north west][inner sep=0.75pt]   [align=left] {......};
\draw (195.04,296) node [anchor=north west][inner sep=0.75pt]    {$w_{\theta \cdot \degp{\alpha}}$};
\draw (395.16,118.25) node [anchor=north west][inner sep=0.75pt]    {$\alpha $};
\draw (582.46,116.5) node [anchor=north west][inner sep=0.75pt]    {$w$};
\draw (425.93,277.63) node [anchor=north west][inner sep=0.75pt]   [align=left] {$ $};
\draw (395.09,322.57) node [anchor=north west][inner sep=0.75pt]   [align=left] {$\displaystyle {\textstyle N^{+}( \alpha )}$};
\draw (402.04,297) node [anchor=north west][inner sep=0.75pt]    {$\beta _{1}$};
\draw (473,264) node [anchor=north west][inner sep=0.75pt]   [align=left] {......};
\draw (537.04,298) node [anchor=north west][inner sep=0.75pt]    {$\beta _{\degp{\alpha}}$};
\draw (377.09,203.38) node [anchor=north west][inner sep=0.75pt]    {$\xi _{1}$};
\draw (560.09,203.38) node [anchor=north west][inner sep=0.75pt]    {$\xi _{\degp{\alpha}}$};
\draw (50,68) node [anchor=north west][inner sep=0.75pt]   [align=left] {a). Distribute charges of $\displaystyle \xi $};
\draw (380,68) node [anchor=north west][inner sep=0.75pt]   [align=left] {b). Any $\displaystyle ( \alpha ,w)$ being charged};
\draw (14,203) node [anchor=north west][inner sep=0.75pt]   [align=left] {$\displaystyle ( \alpha ,w_{1})$};
\draw (194,194) node [anchor=north west][inner sep=0.75pt]   [align=left] {$\displaystyle ( \alpha ,w_{\theta \cdot \degp{\alpha}})$};
\draw (101,200) node [anchor=north west][inner sep=0.75pt]   [align=left] {......};
\draw (440,206) node [anchor=north west][inner sep=0.75pt]   [align=left] {......};

\end{tikzpicture}

\caption{\label{fig:bounded-charge-illus} Illustration of the bound of charges on each vertex-edge pair. Focus on pairs that include $\alpha$, we assume w.log. that $\degp{\alpha}\geq \degp{\beta}$, and use $N^{+}(\alpha) - N^{+}(\beta)$ as a special case of $N^{+}(\alpha) \sym N^{+}(\beta)$. There are two types of edges: edges that distribute charges (red solid lines) and edges that are charged (blue dashed lines). For each edge $\xi=(\alpha,\beta)$ that distributes charges,  there are at least $\theta\cdot \degp{v}$ edges to be charged. Therefore, it suffices to distributed $\frac{1}{\theta\cdot \degp{v}}$ charges to each of the edges being charged. For each edge $(\alpha,w)$ to be charged, only edges indent to $\alpha$ can distribute a charge, which is at most $\degp{v}$ many. Therefore, every vertex-edge pair that include $\alpha$ is charged at most $\theta^{-1}$.}
\end{figure}

\subsubsection*{Part I: Sparse Vertices}

We now analyze the cost of  edges incident on the sparse vertices $\Vsparse$. We define $\spcharge{z}{f}$ as the contribution from the sparse vertices to $\charge{z}{f}$. We show that,
\begin{lemma}\label{lem:charging-sparse}
	There exist sets $\chargeset{e}$ for every $\eps$-sparse vertex $v$ and $e \in E(v)$ such that 
	\[
	\text{for all vertex-edge pairs $(z,f)$:} \quad \spcharge{z}{f} = O(\eps^{-2}). 
	\] 
\end{lemma}

\begin{proof}
Since we place each $v \in \Vsparse$ in a separate cluster,  $\cost{\ALG}{e}=0$ for every $e\in E^{-}(v)$. Hence, we can focus only on $e \in E^{+}(v)$ where $\cost{\ALG}{e}=1$. 

Fix any vertex $v\in \Vsparse$. By the guarantee of~\Cref{thm:decomposition} for $G^+$, there is a subset $S(v)$ of $N^+(v)$ with size at least $|S(v)| \geq \eta_0 \cdot \eps\cdot \degp{v}$ such that for every vertex $u\in S(v)$, 
\begin{align*}
\card{N^{+}(v) \sym N^{+}(u)} \geq \eta_0 \cdot \eps \cdot \max\set{\degp{u},\degp{v}}. 
\end{align*}


Define the set of $(+)$-edges between $v$ and the vertices in $S(v)$ as $E^+(v,S)$. By the size lower bound of $S(v)$, we can bound the charge of $\spcharge{z}{f}$ by the charge of $E^{+}(v,S)$ with a multiplicative $O({1}/{\eps})$ factor. Formally, let $\spScharge{z}{f}$ denote the part of $\spcharge{z}{f}$ contributed by $E^+(v,S)$. One can arbitrarily set the charge-set of every 
$1/(\eps \cdot \eta_0)$ $(+)$-edges incident on $v$ to be the same as a single edge in $E^+(v,S)$. Thus, 
\begin{equation}
\label{equ:spcharge-spScharge}
\text{for all vertex-edge pairs $(z,f)$:} \quad \spcharge{z}{f} \leq \frac{1}{\eps \cdot \eta_0}\cdot \spScharge{z}{f}.
\end{equation}

Therefore, it suffices to upper bound the charge contributed by $E^+(v,S)$. We construct the charge set and upper bound the charge by \Cref{lem:helper-charge} as follows.

\begin{itemize}
\item We set $\EE$ as the set of all edges between vertices $v \in \Vsparse$ and $S(v) \subseteq N^+(v)$. 
\item For any $v \in \Vsparse$ and $u \in S(v)$, 
\[
\card{N^+(u) \sym N^+(v)} \geq \eta_0 \cdot \eps \cdot \max\set{\degp{u},\degp{v}},
\]
so we can set $\theta = \eta_0 \cdot \eps$. 

\end{itemize}
Therefore, by \Cref{lem:helper-charge},  there exists a charge-set for all edges in $\EE$ such that for any vertex-edge pair $(z,f)$, $\spScharge{z}{f} = O(\eps^{-1})$ (as $\eta_0$ is an absolute constant). 
Combining this with~\Cref{equ:spcharge-spScharge} gives us the desired $O(\eps^{-2})$ bound. 
\end{proof}

\subsubsection*{Part II: Almost-Cliques} 


We now turn to the analysis of the charge contributed by almost-cliques. Here, there are two types of edges to consider: $(+)$-edges between different almost-cliques (intra cluster edges) and $(-)$-edges inside each single almost-clique (inter cluster edges); 
all remaining edges are already handled in the previous part. The following two lemmas handle these cases. 


We first bound the charge of $(+)$-edges between almost-cliques, referred to as $\occharge{z}{f}$, in the following lemma.  

\begin{lemma}\label{lem:charging-out-cliques}
	There exists  $\chargeset{e}$ for every $(+)$-edge between two different almost-cliques such that
	\[
	 \textnormal{for all vertex-edge pairs $(z,f)$:} \quad \occharge{z}{f} = O(1).
	\] 
\end{lemma}

\noindent
And next, we bound the charge of $(-)$-edges inside each almost-clique, referred to as $\iccharge{z}{f}$.  

\begin{lemma}\label{lem:charging-in-cliques}
	There exists  $\chargeset{e}$ for every $(-)$-edge inside any single almost-clique such that 
	\[
	 \textnormal{for all vertex-edge pairs $(z,f)$:} \quad \iccharge{z}{f} \leq 1. 
	\] 
\end{lemma}

We now prove each of these two lemmas. 

\begin{proof}[Proof of~\Cref{lem:charging-out-cliques}]
The proof of this lemma is similar to the one for sparse vertices in~\Cref{lem:charging-sparse}. We show that neighborhood of endpoints of $(+)$-intra cluster edges are very different, simply because they belong to different 
almost-cliques, and then apply~\Cref{lem:helper-charge}.

Let $(u,v)$ be a $(+)$-intra cluster edge where $u$ belongs to an almost-clique $K_i$ and $v$ belongs to another almost-clique $K_j$. Without loss of generality, let us assume $\Delta(K_i) \geq \Delta(K_j)$ (where $\Delta(\cdot)$ is the maximum
$(+)$-degree of the almost-clique defined in~\Cref{thm:decomposition}). By properties of almost-cliques in~\Cref{thm:decomposition}, 
\begin{align*}
	\card{N^+(u) - N^+(v)} &\geq \card{(N^+(u) \cap K_i) - (N^+(v) - K_j)} \tag{$K_i$ and $K_j$ are disjoint} \\
	&\geq \card{K_i} - \eps \cdot \Delta(K_i) - \eps \cdot \Delta(K_j) \tag{$u$ has $\leq \eps \cdot \Delta(K_i)$ non-edges in $K_i$ and $v$ has $\leq \eps \cdot \Delta(K_j)$ edges outside $K_j$}\\
	&\geq \Delta(K_i) - \eps \cdot \Delta(K_i) - \eps \cdot \Delta(K_i) - \eps \cdot \Delta(K_j) \tag{size of $K_i$ is $\geq (1-\eps) \cdot \Delta(K_i)$} \\
	&\geq (1-3\eps) \cdot \max\set{\degp{u},\degp{v}}  \tag{as $\Delta(K_i) \geq \Delta(K_j)$ and they both are maximum degree of respective almost-cliques}. 
\end{align*}
As a result, for every $(+)$-edge $(u,v)$ between two different almost-cliques, we have, 
\begin{align}
	\card{N^+(u) \sym N^+(v)} \geq (1-3\eps) \cdot \max\set{\degp{u},\degp{v}} \geq \frac{1}{2} \cdot \max\set{\degp{u},\degp{v}} \label{eq:ac-charge-1}. 
\end{align}
We can now apply \Cref{lem:helper-charge} as follows. We set $\EE$ as the set of all $(+)$-edges $(u,v)$ between vertices $u \in K_{i}$ and $v \in K_{j}$ for any two different almost-cliques $K_i$ and $K_j$. 
Then, by~\Cref{eq:ac-charge-1}, we have $\theta = 1/2$ in~\Cref{lem:helper-charge}.  Thus, $\occharge{z}{f} = O(1)$ for all  valid vertex-edge pairs $(z,f)$ as desired. \Qed{\Cref{lem:charging-out-cliques}}

\end{proof}


\begin{proof}[Proof of~\Cref{lem:charging-in-cliques}]
Proof of this lemma follows a different strategy compared to the previous ones as we now have to handle  $(-)$-edges (inside each almost-clique) as opposed to $(+)$-edges. The main idea of proof is to show that
the best clustering one can do for almost-cliques individually (in absence of all other edges) is to just cluster them one by one exactly as in $\ALG$.

Fix an almost-clique $K_i$. We need to find a charge-set for each $(-)$-edge inside $K_i$ as these are the edges with cost $1$ in $\ALG$. 
%
%
Recall that $\Delta(K_i)$ is the maximum $(+)$-degree of any vertex in $K_i$. We have, 
\begin{align}
	\card{\set{e=(u,v) \in E^- \mid u,v \in K_i}} = \frac12 \cdot \sum_{v \in K_i} \card{N^-(v) \cap K_i} \leq \frac\eps2 \cdot \card{K_i} \cdot \Delta(K_i) \label{eq:cost-Ki},
\end{align}
as number of $(-)$-edges of each vertex in $K_i$ inside $K_i$ is at most $\eps \cdot \Delta(K_i)$ by~\Cref{dec:non-neighbors} of~\Cref{thm:decomposition}. 

Consider the optimal clustering $\OPT$ of the graph $G$.  Let $O_1,\ldots,O_\ell$ for $\ell \geq 1$ denote the set of clusters in $\OPT$ that have non-zero intersection with $K_i$. 
 We have two cases: 
\begin{itemize}
	\item \textbf{Case (1): when $\card{O_j \cap K_i} \leq (1-3\eps) \cdot \Delta(K_i)$ for all $j \in [\ell]$.} Any vertex $v \in K_i$ has at least $\card{K_i} - \eps \cdot \Delta(K_i)$ $(+)$-edges
	inside $K_i$ by~\Cref{thm:decomposition}. Consequently, for any $O_j$, any vertex $v \in O_j \cap K_i$ has 
	\begin{align*}
		\card{N^+(v) - O_j} &\geq \card{N^+(v) \cap K_i} - \card{O_j \cap K_i} \\
		&\geq \card{K_i} - \eps \cdot \Delta(K_i) - \card{O_j \cap K_i}\tag{by the discussion above} \\
		&\geq (1-2\eps) \cdot \Delta(K_i) - (1-3\eps) \cdot \Delta(K_i) \tag{by~\Cref{thm:decomposition} and our assumption on size of $O_j$} \\
		&=\eps \cdot \Delta(K_i). 
	\end{align*}
	All the edges in $N^+(v) - O_j$ are now inter cluster $(+)$-edges and thus have cost $1$ in $\OPT$. Hence, 
	\[
		\card{\set{e=(u,v) \in E^+ \mid u,v \in K_i,~\cost{\OPT}{e}=1}} \geq \frac12 \sum_{v \in K_i} \card{N^+(v) - O_j} \geq \frac\eps2 \cdot \card{K_i} \cdot \Delta(K_i). 
	\]
	Consequently, we can pick the charge-set of edges in LHS of~\Cref{eq:cost-Ki} by an arbitrary one-to-one mapping to (a subset of) the set in the LHS of above equation. This ensures that $\iccharge{z}{f} \leq 1$ in this case as desired. 

	\item \textbf{Case (2): when $\card{O_{j^*} \cap K_i} > (1-3\eps) \cdot \Delta(K_i)$ for some $j^* \in [\ell]$.} For any  $v$ in some $O_j\cap K_i$ for $j \neq j^*$, 
	\begin{align*}
		\card{N^+(v) \cap O_{j^*}} &\geq \card{N^+(v) \cap K_i} - \card{K_i} + \card{K_i \cap O_{j^*}} \\
		&\geq \card{K_i} - \eps \cdot \Delta(K_i) - (1+\eps) \cdot \Delta(K_i) + (1-3\eps) \cdot \Delta(K_i) \tag{by~\Cref{dec:non-neighbors} and \Cref{dec:size} of~\Cref{thm:decomposition} and our assumption on size of $O_j$} \\
		&\geq (1-6\eps) \cdot \Delta(K_i) \tag{by~\Cref{dec:size} of~\Cref{thm:decomposition}}. 
	\end{align*}
	All the edges in $N^+(v) \cap O_{j^*}$ are now inter cluster $(+)$-edges and thus have cost $1$ in $\OPT$. Hence, for any vertex $v \in O_j$, 
	\[
		\card{\set{e=(u,v) \mid u \in N^+(v) \cap O_{j^+},~\cost{\OPT}{e}=1}} \geq (1-6\eps) \cdot \Delta(K_i) \geq \eps \cdot \Delta(K_i),
	\]
	for $\eps < 1/7$. On the other hand, the only edges of $v$ with $\cost{\ALG}{e}=1$ in LHS of~\Cref{eq:cost-Ki} are $(-)$-edges of $v$ in $K_i$ which are at most $\eps \cdot \Delta(K_i)$. Thus, we can arbitrarily 
	pick the charge-set of edges of $v$ in LHS of~\Cref{eq:cost-Ki} from the LHS above. Finally, for $(-)$-edges with both endpoints in $O_{j^*}$, we can simply pick their charge-set to be themselves as $\OPT$ also has cost 
	$1$ for them. This ensures that $\iccharge{z}{f} \leq 1$ in this case as well. 
\end{itemize} 
This concludes the proof. \end{proof}

\medskip

In conclusion, the proof of~\Cref{thm:cc-alg} now follows by~\Cref{lem:charging-sparse} (for handling all edges with at least one sparse endpoint) 
and~\Cref{lem:charging-out-cliques,lem:charging-in-cliques} (for handling all edges between dense vertices).


\newcommand{\degpt}[2]{\textnormal{deg}^{\!+}_{#1}\!(#2)}
\newcommand{\Vhigh}{V_\textnormal{high}}
\newcommand{\Vlow}{V_\textnormal{low}}

\newcommand{\degt}[1]{\textnormal{deg}_{t}(#1)}

\section{Sublinear Algorithms for Correlation Clustering}
\label{sec:alg}

With the sparse-dense decomposition result from \Cref{sec:decomposition} and the correlation clustering scheme in \Cref{sec:cc-approx} that built upon it, we can now describe our sublinear algorithms. We start with our sublinear-time algorithm. 
\begin{theorem}[Formalization of \Cref{res:main-time-alg}]
\label{thm:sub-time-alg}
There exists a randomized algorithm that given a labeled graph $G=(V,E)$, specified via adjacency list of its $(+)$-subgraph $G^+$, 
with high probability finds an $O(1)$-approximation to the correlation clustering problem on $G$ in $O(n\log{n})$ query and $O(n\log^2{n})$ time. 
\end{theorem}

We remind the reader  about our discussion earlier in~\Cref{sec:sublinear-models} on the necessity of access to the adjacency list of $G^+$ as opposed to $G$ in~\Cref{thm:sub-time-alg} (see also~\Cref{app:time-lower}).

The second sublinear algorithm we present is a sublinear-space streaming algorithm.

\begin{theorem}[Formalization of \Cref{res:main-space-alg}]
\label{thm:sub-space-alg}
There exists a randomized single-pass semi-streaming algorithm that given a labeled graph $G=(V,E)$, specified via a stream of edges of $G$ together with their labels, 
with high probability finds an $O(1)$-approximation to the correlation clustering problem on $G$ in $O(n\log{n})$ space. Moreover, the algorithm has $O(\log{n})$ processing time per each element of the stream and $O(n\log^2{n})$ post-processing time. 
\end{theorem}

As we discussed in~\Cref{sec:sublinear-models} (and is formally shown in~\Cref{app:stream-upper}), this algorithm can be extended to various other streaming scenarios, such as when only edges of $G^+$ or $G^-$ are streamed, or dynamic (insertion-deletion) and sliding window streams, by increasing the space  with at most a $\poly\!\log{(n)}$ factor.

\subsection{A Sublinear-Time Algorithm: Proof of \Cref{thm:sub-time-alg}}
\label{subsec:sub-time-alg}

The algorithm is a direct implementation of the recovery algorithm of~\Cref{thm:decomposition} for finding the decomposition plus the scheme of~\Cref{thm:cc-alg}. We also need to show that we can provide the recovery algorithm 
of~\Cref{thm:decomposition} with proper information it needs. This is done as follows. 

\begin{Algorithm}\label{alg:cc-sub-time}
{A sublinear-time algorithm for correlation clustering.} 

\begin{itemize}
\item \textbf{Input:} A labeled graph $G=(V,E)$ specified via adjacency list access to $G^+$.
\end{itemize}

\begin{enumerate}[label=$(\roman*)$]
\item Let $\eps > 0$ be a sufficiently small \underline{constant} as prescribed by~\Cref{thm:decomposition}. 

\item \label{line:degree-query} For each vertex $v\in V$, use degree queries to get positive degree $\degp{v}$ of $v$.

\item \label{line:N-query-sample} For each vertex $v$, use neighbor queries to sample $t = (c \cdot \log{n})/\eps^2$ neighbors of $v$ from $N^+(v)$ with repetition to get $\NS{v}$ (for the absolute constant $c > 0$ in~\Cref{thm:decomposition}). 

\item \label{line:S-query-sample} Sample each vertex with probability $p_v := \min\set{\frac{c\log(n)}{\degp{v}},1}$ and call this set $\sample$. Use neighbor queries to get $N^+(v)$ for $v \in \sample$.

\item \label{line:decompose-query} Run the algorithm of~\Cref{thm:decomposition} for sparse-dense decomposition with parameter $\eps$ and  the inputs $\{\NS{v}\}_{v \in V}$ and $\set{N^+(v)}_{v \in \sample}$ to its recovery algorithm. 

\item \label{line:clustering-query} Output clustering $\ALG$ based on the resulting $\Vsparse \sqcup K_1 \sqcup \ldots \sqcup K_k$ as prescribed in \Cref{thm:cc-alg}. 

\end{enumerate}
\end{Algorithm}

We now prove the correctness and the bounds on query and time complexity of~\Cref{alg:cc-sub-time}.

\paragraph{Correctness.} The information provided to the recovery algorithm of~\Cref{thm:decomposition} by~\Cref{alg:cc-sub-time} is exactly as prescribed in the theorem (for the underlying graph $G^+$).  
As such, with high probability, the resulting decomposition is a valid sparse-dense decomposition specified by~\Cref{thm:decomposition}. Conditioned on this event, 
by~\Cref{thm:cc-alg}, the returned answer is an $O(1)$-approximation to the correlation clustering on $G$. 

\paragraph{Query complexity.} The total number of queries made by~\Cref{alg:cc-sub-time} is equal to $n$ degree queries plus the number of edge-samples and neighbors of all vertex-samples. 
By~\Cref{lem:sample-helper-lem}, this is $O(n\log{n})$ edges. 

\paragraph{Runtime analysis.} The runtime of Lines \ref{line:degree-query},\ref{line:N-query-sample}, and \ref{line:S-query-sample} is equal to the query complexity of the algorithm and is thus $O(n\log{n})$ with high probability. 
The runtime of~\Cref{line:decompose-query} is equal to the recovery algorithm of Line~\Cref{thm:decomposition} which is $O(n\log^2{n})$ with high probability. The runtime of Line~\ref{line:clustering-query} is equal to $O(n)$ 
as it only involves a direct partitioning of $n$ vertices as specified in the statement of~\Cref{thm:cc-alg}. This is $O(n\log^2{n})$ time in total. 

\smallskip
This concludes the proof of~\Cref{thm:sub-time-alg}.

\subsection{A  Semi-Streaming Algorithm: Proof of \Cref{thm:sub-space-alg}}
\label{subsec:sub-space-alg}

We now give a single-pass semi-streaming algorithm for correlation clustering in insertion-only streams (over the edges of the input labeled graph $G$).~\Cref{app:stream-upper} contains further extensions of this algorithm to other
streaming models. 

Our semi-streaming algorithm is also a direct implementation of our recovery algorithm in~\Cref{thm:decomposition} and the scheme of~\Cref{thm:cc-alg} (by focusing on edges of $G^+$ in the stream and simply skipping any edge of $G^-$). 
However, compared to the previous section, 
for this algorithm we have to be a bit careful in how we exactly provide the required information to the recovery algorithm of~\Cref{thm:decomposition}. This is primarily because, in a single pass over the stream, 
we will not know degrees of vertices beforehand so that we can sample the set $\sample$ store all their neighbors appropriately. Nevertheless, we show that simple ideas in reservoir sampling~\cite{Vitter85} can be used to address this problem. 
Thus we first start by designing a subroutine for obtaining $\sample$ and $N(v)$ for $v \in \sample$ and then show use to obtain our final semi-streaming algorithm. 

\paragraph{Sampling vertices and storing their neighbors.} We present the following lemma and algorithm for sampling vertices inversely proportional to their degree and storing all neighbors of sampled vertices (We note that we shall apply
the following lemma to the underlying graph $G^+$). The idea behind this lemma seems standard to us and we present it here for completeness. 

\begin{lemma}\label{lem:vertex-sample}
	Let $\beta_0 > 0$ be a sufficiently large constant. 
	There is a semi-streaming algorithm that given any arbitrary graph $G=(V,E)$ (not necessarily a labeled graph) specified via a stream of its edges and a parameter $\beta > \beta_0$, 
	at every point of time $t$ during the stream: 
	\begin{enumerate}[label=$\roman*).$]
	\item Maintains a collection $S_t$ of vertices together with $N_t(v)$ for all $v \in S_t$ so that each vertex is sampled independently and with
	probability $\min\set{(\beta \cdot \log{n})/\degt{v},1}$ in $S_t$ (here, $N_t(v)$ and $\degt{v}$ refer to the set of neighbors of $v$ and degree of $v$ among the edges up to time $t$ in the stream);
	\item With high probability, uses space of $O(\beta \cdot n\log{n})$ throughout the stream and $O(1)$ time per update. 
	\end{enumerate} 
	(We note that the independence guarantee of the algorithm is across the vertices and \emph{not} time steps.) 
\end{lemma}
\begin{proof}

The algorithm is as follows. 

\begin{tbox}
Sampling algorithm of~\Cref{lem:vertex-sample}. 

\begin{enumerate}[label=$(\roman*)$]
\item Let $S_1 = V$, $N_1(v) = \emptyset$, $\textnormal{deg}_1(v) = 0$ for $v \in V$.
\item For each arriving edge $e_t = (u_t,v_t)$: 
\begin{enumerate}
	\item Update $N_t(w)$ and $\degt{w}$ for $w \in S_{t-1}$ by adding $u_t$ and $v_t$ to the neighborhood of respective vertices and increasing their degree and keeping other neighbor-sets intact. 
	\item Update $S_t$ from $S_{t-1}$ by keeping all vertices other than $u_t$ and $v_t$ in $S_t$ and {removing} $z \in \set{u_t,v_t}$ from $S_t$ with probability $1/\degt{z}$ if $\degt{z} > (\beta \cdot \log{n})$; if $z$ is removed 
	from $S_t$ we also discard $N_t(z)$ from the memory. 
\end{enumerate}
\end{enumerate}
\end{tbox}

Fix a vertex $v \in V$ and let $t_0(v)$ denote the first time step $t$ such that $\degt{v} > (\beta \cdot \log{n})$. For any time $t < t_0(v)$, we have $v$ in $S_t$ as it cannot be removed by the algorithm. For a time 
$t_1 \geq t_0(v)$, we have, 
\begin{align*}
	\Pr\paren{v \in S_{t_1}} &= \prod_{t=t_0}^{t_1} \Pr\paren{v \in S_{t} \mid v \in S_{t-1}} \\
	&= \prod_{t=t_0}^{t_1} (1-\frac{1}{\degt{v}}) \tag{as $v$ is removed from $S_t$ w.p. $1/\degt{v}$} \\
	&= {\frac{\textnormal{deg}_{t_0(v)}-1}{\degt{v}} } \tag{by a simple cancelation of intermediate terms} \\
	&= \frac{(\beta \cdot \log{n})}{\degt{v}}.
\end{align*}
Thus, each vertex $v$ belongs to $S_t$ with probability $\min\set{(\beta\,\log{n})/\degt{v},1}$. Moreover, the choice of inclusion or exclusion of different vertices in $S_t$ is independent, proving the first part of the lemma. 

For the second part, given the correctness of the first, at each time step $t$, the information stored by the algorithm consists a random subset $S_t$ of vertices where each vertex is included with probability $\min\set{(\beta\,\log{n})/\degt{v},1}$ plus
all edges incident on vertices. As such, we can apply the same argument of~\Cref{lem:sample-helper-lem} to the graph at time $t$ to get that the total number of stored edges is $O(\beta \cdot n\log{n})$ with high probability (the argument is virtually identical and we do not repeat it here). A union bound on at most ${{n}\choose{2}}$ steps concludes the proof (the bound of $O(1)$ on the update time is immediate). 
\end{proof}

\paragraph{The semi-streaming algorithm.} We can now present our semi-streaming algorithm for~\Cref{thm:sub-space-alg}. 

\begin{Algorithm}\label{alg:cc-sub-space}
{A single-pass semi-streaming algorithm for correlation clustering.} 

\begin{itemize}
\item \textbf{Input:} A labeled graph $G=(V,E)$ specified via an arbitrarily ordered stream of edges $E$. 
\end{itemize}

\begin{enumerate}[label=$(\roman*)$]
\item Let $\eps > 0$ be a sufficiently small \underline{constant} as prescribed by~\Cref{thm:decomposition}. 

\item \label{line:stream-degree} For each vertex $v\in V$, use a counter over edges of $E^+(v)$ to maintain $\degp{v}$. 

\item \label{line:stream-neighbor} For each vertex $v$, use reservoir sampling to sample $t = (c \cdot \log{n})/\eps^2$ neighbors of $v$ from $N^+(v)$ with repetition to get $\NS{v}$ (for the absolute constant $c > 0$ in~\Cref{thm:decomposition}). 

\item \label{line:stream-sample} Run the algorithm of~\Cref{lem:vertex-sample} on the graph $G^+$  with parameter $\beta = c$. Let $\sample$ be the final set of vertices 
maintained by the algorithm and note we have $N^+(v)$ for $v \in \sample$. 

\item \label{line:decompose-query} Run the algorithm of~\Cref{thm:decomposition} for sparse-dense decomposition with parameter $\eps$ and the inputs $\{\NS{v}\}_{v \in V}$ and $\set{N^+(v)}_{v \in \sample}$ to its recovery algorithm. 

\item \label{line:clustering-query} Output clustering $\ALG$ based on the resulting $\Vsparse \sqcup K_1 \sqcup \ldots \sqcup K_k$ as prescribed in \Cref{thm:cc-alg}. 

\end{enumerate}
\end{Algorithm}

\paragraph{Correctness.} The information provided to the recovery algorithm of~\Cref{thm:decomposition} by~\Cref{alg:cc-sub-time} is exactly as prescribed in the theorem (for the underlying graph $G^+$).  
As such, with high probability, the resulting decomposition is a valid sparse-dense decomposition specified by~\Cref{thm:decomposition}. Conditioned on this event, 
by~\Cref{thm:cc-alg}, the returned answer is an $O(1)$-approximation to the correlation clustering on $G$. 

\paragraph{Space complexity.} Line~\ref{line:stream-degree} requires storing $O(n)$ numbers. Line~\ref{line:stream-neighbor} requires storing $O(\log{n})$ neighbors of each vertex for $O(n\log{n})$ space in total. 
Line~\ref{line:stream-sample} uses $O(n\log{n})$ space with high probability by~\Cref{lem:vertex-sample}. The algorithms of~\Cref{thm:decomposition,thm:cc-alg} require space proportional to their input which is $O(n\log{n})$ in total. 
Thus, overall space of the algorithm is $O(n\log{n})$. 

\paragraph{Update time and post-processing time.} The update time  is $O(1)$ for Lines~\ref{line:stream-degree} and~\ref{line:stream-sample} and $O(\log{n})$ for Line~\ref{line:stream-neighbor}. 
 Thus, the update time is $O(\log{n})$. The post-processing time is $O(n\log^2{n})$ time by~\Cref{thm:decomposition}. 

\smallskip
This concludes the proof of~\Cref{thm:sub-space-alg}.

\bibliographystyle{alpha}
\bibliography{new}

\appendix

\clearpage


\newcommand{\wrap}[2]{\subsubsection*{Proof of~\Cref{#1}.} \emph{#2}}

\section{Missing Proofs of~\Cref{sec:decomposition}}\label{app:decomposition}

\wrap{p:light}
{Any light vertex $v$ has at least $\delta \cdot \deg{v}$ neighbors $u$ such that
	\[
		\card{N(u) - N(v)} \geq \frac\eps{(1+\eps)} \cdot \deg{u} = \frac\eps{(1+\eps)} \cdot \max\set{\deg{u},\deg{v}}. 
	\]
}
\begin{proof}
	We argue that set $N(v) - \low{v}$ has the required property. Consider $u \in N({v})-\low{v}$. We have,
	\[
	\deg{u} - \deg{v} \geq  \frac\eps{(1+\eps)} \cdot \deg{u} = \frac\eps{(1+\eps)} \cdot \max\set{\deg{u},\deg{v}},
	\]
	as $\deg{u} > \deg{v}$. As such, $N(u) - N(v)$ satisfies the equation in the property, simply because size of $N(u)$ is sufficiently larger than that of $N(v)$. Given that $\card{N(v) - \low{v}} > \delta \cdot \deg{v}$ by the definition of $v$ being a light vertex, 
	we are done. 
\end{proof}

\wrap{p:isolated}{
Any low-sparse vertex $v$ has at least $\delta \cdot \deg{v}$ neighbors $u$ such that
	\[
		\card{N(v) - N(u)} \geq \eps \cdot \deg{v} \geq \frac{\eps}{(1+\eps)} \cdot \max\set{\deg{u},\deg{v}}. 
	\]
}
\begin{proof}
	We argue that the set $\isolated{v} \cap \low{v}$ has the required property. Firstly, size of this set is at least $\delta \cdot \deg{v}$ by~\Cref{def:isolated}. Moreover, 
	for any $u \in \isolated{v} \cap \low{v}$, 
	\[
		\card{N(v) - N(u)} \geq \card{\low{v} - N(u)} \geq \eps \cdot \deg{v} \geq \frac{\eps}{(1+\eps)} \cdot \deg{u},
	\]
	where the last inequality is because $u \in \low{v}$ and thus $\deg{u} \leq (1+\eps) \cdot \deg{v}$ by~\Cref{def:light}.
\end{proof}

\wrap{p:dense}{
	For every dense vertex $v \in \Dense{G}{\eps}{\delta}$:
	\begin{enumerate}[label=$(\roman*)$]
		\item the number of non-edges inside $\low{v}$ is at most $\frac{\eps+\delta}{2} \cdot \deg{v}^2$;
		\item the number of edges going out of $\low{v}$ is at most $2\,(\eps+\delta) \cdot \deg{v}^2$;
	\end{enumerate} 
}
\begin{proof}[Proof of Part $(i)$.]
	We have, 
	\begin{align*}
		\text{\# of non-edges in $\low{v}$} &= \frac12 \cdot \sum_{u \in \low{v}} \hspace{-10pt} \text{\# of non-edges of $u$ in $\low{v}$} \\
		&= \frac12 \cdot \sum_{u \in \low{v} \cap \isolated{v}} \hspace{-35pt} \text{\# of non-edges of $u$ in $\low{v}$} \\
		&\hspace{1.5cm} + \frac12 \cdot \sum_{u \in \low{v} - \isolated{v}} \hspace{-35pt} \text{\# of non-edges of $u$ in $\low{v}$}.
	\end{align*}
	We bound the first sum by $\delta \cdot \deg{v}^2$ and the second with $\eps \cdot \deg{v}^2$ which concludes the proof.
	
	For the first term, 
	\[
		\sum_{u \in \low{v} \cap \isolated{v}} \hspace{-35pt} \text{\# of non-edges of $u$ in $\low{v}$} \leq \sum_{u \in \low{v} \cap \isolated{v}} \deg{v} \leq (\delta \cdot \deg{v}) \cdot \deg{v}, 
	\]
	where the last inequality is because $v$ is not low-sparse and thus $\card{\low{v} \cap \isolated{v}} \leq \delta \cdot \deg{v}$. 
	
	For the second term, 
	\[
		\sum_{u \in \low{v} - \isolated{v}} \hspace{-35pt} \text{\# of non-edges of $u$ in $\low{v}$} \leq \sum_{u \in \low{v} - \isolated{v}} \eps \cdot \deg{v} \leq \deg{v} \cdot (\eps \cdot \deg{v}),
	\] 
	where the first inequality is because $v$ is not isolated from vertex $u$ in the sum. 
\end{proof}
\begin{proof}[Proof of Part $(ii)$.]
	We have, 
	\begin{align*}
		\text{\# of edges going out of $\low{v}$} &= \sum_{u \in \low{v}} \hspace{-10pt} \text{\# of edges $u$ going out of $\low{v}$} \\
		&= \sum_{u \in \low{v} \cap \isolated{v}} \hspace{-35pt} \text{\# of edges of $u$ going out of $\low{v}$} \\
		&\hspace{1.5cm} + \sum_{u \in \low{v} - \isolated{v}} \hspace{-35pt} \text{\# of edges of $u$ going out of $\low{v}$} .
	\end{align*}
	We bound each of these sums. For the first term, 
	\begin{align*}
		\sum_{u \in \low{v} \cap \isolated{v}} \hspace{-35pt} \text{\# of edges of $u$ going out of $\low{v}$} &\leq \sum_{u \in \low{v} \cap \isolated{v}} \deg{u} \\
		&\leq  \sum_{u \in \low{v} \cap \isolated{v}} (1+\eps) \cdot \deg{v} \tag{by~\Cref{def:light} as $u \in \low{v}$} \\
		&\leq (\delta \cdot \deg{v}) \cdot (1+\eps) \cdot \deg{v} \tag{by~\Cref{def:isolated} as $v$ is not low-sparse}.
	\end{align*}
	
	For the second term, 
	\begin{align*}
		\sum_{u \in \low{v} - \isolated{v}} \hspace{-35pt}\text{\# of edges of $u$ going out of $\low{v}$} &= \sum_{u \in \low{v} - \isolated{v}} \card{N(u) - \low{v}} \\
		&= \sum_{u \in \low{v} - \isolated{v}} \card{N(u)} - \card{\low{v}} + \card{\low{v} - {N(u)}} \\
		&\leq \sum_{u \in \low{v} - \isolated{v}} \deg{u} - \deg{v} + \eps \cdot \deg{v} \tag{by~\Cref{def:isolated} as $v$ is not isolated from vertex $u$} \\
		&\leq \sum_{u \in \low{v} - \isolated{v}} 2\eps \cdot \deg{v} \tag{by~\Cref{def:light} as $u \in \low{v}$} \\
		&\leq \deg{v} \cdot (2\eps \cdot \deg{v)}. 
	\end{align*}
	Summing up the previous two bounds together, plus the fact $\eps < 1$, concludes the proof. 
\end{proof}

\wrap{p:kernel}{
For every dense vertex $v \in \Dense{G}{\eps}{\delta}$, $\kernel{v}$ satisfies the following properties: 
	\begin{enumerate}[label=$(\roman*)$]
		\item $\kernel{v}$ is a subset of $\low{v}$ with size at least $(1-2\delta) \cdot \deg{v}$; 
		\item every vertex $u \in \kernel{v}$ has at least $(1-\eps-\delta) \cdot \deg{v}$ neighbors in $\low{v}$.
	\end{enumerate} 
}
\begin{proof}
	By definition, $\kernel{v} \subseteq \low{v}$ and since $v$ is not low-sparse, by~\Cref{def:isolated}, we have 
	\[
	\card{\low{v} \cap \isolated{v}} < \delta \cdot \deg{v},
	\]
	and since $v$ is not light, by~\Cref{def:light}, we have $\card{\low{v}} > (1-\delta) \cdot \deg{v}$. Putting these two together proves part $(i)$. 
	
	Similarly, since for any $u \in \kernel{v}$, $u$ is not an $\eps$-isolated neighbor of $v$, by~\Cref{def:isolated}, 
	\[
		\card{N(u) \cap \low{v}}= \card{\low{v}} - \card{\low{v} - N(u)} \geq \card{\low{v}} - \eps \cdot \deg{v}.
	\]
	Again, since $v$ is  not light, we have $\card{\low{v}} > (1-\delta) \cdot \deg{v}$, which concludes the proof. 
\end{proof}

\wrap{p:kernel-dense}{
	For every $(\eps,\delta)$-dense vertex $v$, any vertex $u \in \kernel{v}$ is $(4\eps+2\delta,2\eps+2\delta)$-dense. 
}
\begin{proof}
	Fix a vertex $u \in \kernel{v}$; we shall prove that $u$ is neither light nor low-sparse (for the given parameters) which concludes the proof. 
	
	Consider a vertex $w \in \low{v}$. We have, 
	\[
		\deg{w} \leq (1+\eps) \cdot \deg{v} \leq \frac{1+\eps}{1-\eps - \delta} \cdot \deg{u} \leq (1+4\eps+2\delta) \cdot \deg{u};
	\]
	 the first inequality is by~\Cref{def:light} as $w \in \low{v}$, the second is by~\Cref{p:kernel} as $u \in \kernel{v}$, and the last one is by a simple calculation assuming $\eps,\delta < 1/8$. 
	 This in turn implies that any vertex in $N(u) \cap \low{v}$ belongs to $\Low{u}{4\eps+2\delta}.$ As such, we have, 
	\begin{align}
		\card{\Low{u}{4\eps+2\delta}} \geq \card{N(u) \cap \low{v}} \geq (1-\eps-\delta) \cdot \deg{v} \geq \frac{1-\eps-\delta}{1+\eps} \cdot \deg{u} \geq (1-(2\eps+\delta)) \cdot \deg{u}, \label{eq:kernel-in}
	\end{align}
	where the second inequality is by~\Cref{p:kernel}, and the third inequality is because $u \in \low{v}$. This implies that $u$ is \emph{not} an $(4\eps+2\delta,2\eps+\delta)$-light vertex. 

	We now prove that $u$ cannot be low-sparse either. Define 
	\[
		S(u) := N(u) \cap \low{v} - \isolated{v}.
	\]
	Note that, by the discussion above, $S(u)$ is a subset of $\Low{u}{4\eps+2\delta}$. 
	
	Firstly, we have that, 
	\begin{align*}
		\card{S(u)} &= \card{(N(u) \cap \low{v}) - (\low{v} \cap \isolated{v})} \\
		&\geq \card{N(u) \cap \low{v}} - \card{\low{v} \cap \isolated{v}} \\
		&\geq \card{N(u) \cap \low{v}} - \delta \cdot \deg{v} \tag{as $v$ is not $(\eps,\delta)$-isolated} \\
		&\geq (1-(2\eps+2\delta)) \cdot \deg{u} \tag{by~\Cref{eq:kernel-in}}. 
	\end{align*}
	Secondly, for any vertex $w \in S(u)$, 
	\begin{align*}
		\card{\Low{u}{4\eps+2\delta} - N(w)} &\leq \card{\low{v} - N(w)} + \card{\Low{u}{4\eps+2\delta} - \low{v}} \\
		&\leq \eps \cdot \deg{v} +  \card{\Low{u}{4\eps+2\delta} - \low{v}} \tag{by~\Cref{def:isolated} for $w \in \low{v}$}\\
		&= \eps \cdot \deg{v} +  \card{\Low{u}{4\eps+2\delta}} - \card{\Low{u}{4\eps+2\delta} \cap \low{v}} \\
		&= \eps \cdot \deg{v} +  \card{\Low{u}{4\eps+2\delta}} - \card{N(u) \cap \low{v}} \tag{as $N(u) \cap \low{v} \subseteq \Low{u}{4\eps+2\delta}$ as discussed above} \\
		&\leq \eps \cdot \deg{v} + \deg{u} - (1-(2\eps+\delta)) \cdot \deg{u} \tag{by~\Cref{eq:kernel-in}}  \\
		&\leq \frac{\eps}{1-\eps - \delta} \cdot \deg{u} + (2\eps+\delta) \cdot \deg{u} \tag{by~\Cref{p:kernel}, $\deg{u} \geq (1-\eps -\delta) \cdot \deg{v}$} \\
		&\leq (4\eps + \delta) \cdot \deg{u}. 
	\end{align*}
	This means that no vertex $S(u)$ is $(4\eps+\delta)$-isolated from $u$. Thus, by the bound on the size of $S(u)$, there are at most $(2\eps+2\delta)$ vertices in $\Low{u}{4\eps+2\delta}$ that 
	are $(4\eps+\delta)$-isolated from $u$. Thus, $u$ is also not a $(4\eps+2\delta,2\eps+2\delta)$-low-sparse vertex. 
\end{proof}

\wrap{p:kernel-monotone}{
	For every $(\eps,\delta)$-dense vertex $v$, any vertex in $\Kernel{v}{\eps}{\delta}$ also belongs to $\Kernel{v}{\eps'}{\delta'}$ for any $\eps' > \eps+\delta$ and arbitrary $\delta' > 0$. 
}
\begin{proof}
	Let $u$ be in $\kernel{v}$. By definition, we have 
	\[
	\deg{u} \leq (1+\eps) \cdot \deg{v} \qquad \text{and} \qquad \card{\low{v} - N(u)} < \eps \cdot \deg{v}.
	\]
	This obviously implies that $\deg{u} \leq (1+\eps') \cdot \deg{v}$ so $u \in \Low{v}{\eps'}$ also (a necessary condition to be in $\Kernel{v}{\eps'}{2\delta}$). However, 
	we also need to have that 
	\[
		\card{\Low{v}{\eps'} - N(u)} < \eps \cdot \deg{v},
	\]
	a guarantee which is not immediate since $\Low{v}{\eps'}$ can potentially contain many more vertices compared to $\low{v}$. However, since $v$ is also $(\eps,\delta)$-dense, 
	then we know that 
	\[
		\card{\Low{v}{\eps'} - \Low{v}{\eps}} < \delta \cdot \deg{v}.
	\]
	Combining the above two equations implies that $u \notin \Isolated{v}{\eps'}$, thus finalizing the proof. 
\end{proof}

\wrap{p:1}{
	$\kernel{v}$ belongs to $C_v$ and thus $\card{\low{v} \cap C_v} \geq (1-2\delta) \cdot \deg{v}$. 
}
\begin{proof}
	By~\Cref{p:kernel}, every vertex $u \in \kernel{v}$ satisfies $\card{N(u) \cap \low{v}} \geq (1-\eps-\delta) \cdot \deg{v}$. Thus by Rule (1), $u$ will be included in $C_v$. The second part of 
	the property now follows immediately from the size of $\kernel{p}$ and that it is a subset of $\low{v}$. 
\end{proof}

\wrap{p:2}{
	Every vertex $u \in C_v$ satisfies $\card{N(u) \cap C_v} \geq (1-7\eps-9\delta) \cdot \deg{v}$. 
}
\begin{proof}
	By Rule (2), every vertex $u \in C_v$ satisfies $\card{N(u) \cap \low{v}} \geq (1-7\eps-7\delta) \cdot \deg{v}$. As such, 
	\begin{align*}
		\card{N(u) \cap C_v} &\geq \card{N(u) \cap \low{v}} - \card{\low{v} - S_v} \\
		&\geq (1-7\eps-7\delta) \cdot \deg{v} - \card{\low{v}} + \card{\low{v} \cap S_v}  \\
		&\geq (1-7\eps-7\delta) \cdot \deg{v} - \deg{v} + (1-2\delta) \cdot \deg{v} \tag{by~\Cref{p:1}},
	\end{align*}
	which is $(1-7\eps-9\delta) \cdot \deg{v}$ as desired. 
\end{proof}

\wrap{p:3}{
	Every vertex $u \in C_v$ satisfies $\card{N(u) - C_v} \leq (9\eps+11\delta) \cdot \deg{v}$. 
}
\begin{proof}
	By Rule (2), every vertex $u \in C_v$ satisfies $\deg{u} \leq (1+2\eps+2\delta) \cdot \deg{v}$. Thus, 
	\[
		\card{N(u) - C_v} = \card{N(u)} - \card{N(u) \cap C_v} \leq (1+2\eps+2\delta) \cdot \deg{v} - (1-7\eps-9\delta) \cdot \deg{v} = (9\eps + 11\delta) \cdot \deg{v} 	
	\]
	as desired. 
\end{proof}

\wrap{p:4}{
	 $\card{C_v - \low{v}} \leq (3\eps+3\delta) \cdot \deg{v}$.  
}
\begin{proof}
	By~\Cref{p:dense}, there are at most $(2\eps+2\delta) \cdot \deg{v}^2$ edges going out of $\low{v}$. On the other hand, by Rule (2), for any  $u \in S_v - \low{v}$, there is a dedicated set of at least $(1-7\eps-7\delta) \cdot \deg{v}$ edges among 
	the outgoing edges of $\low{v}$. For $\eps,\delta < 1/42$, we have $(1-7\eps-7\delta) \cdot \deg{v} > 2/3 \cdot \deg{v}$. Thus, 
	\[
		\card{C_v - \low{v}} \leq (2\eps+2\delta) \cdot \deg{v}^2/(2/3 \cdot \deg{v}) = (3\eps+3\delta) \cdot \deg{v},
	\]
	concluding the proof. 
\end{proof}

\wrap{p:4'}{
	Every vertex $u \in C_v$ satisfies $\card{C_v - N(u)} \leq (10\eps+10\delta) \cdot \deg{v}$. 
}
\begin{proof}
	By Rule (2), every vertex $u \in C_v$ satisfies $\deg{u} \leq (1+2\eps+2\delta) \cdot \deg{v}$. Thus, 
	\begin{align*}
		\card{C_v-N(u)} &\leq \card{C_v - \low{v}} + \card{\low{v} - N(u)} \\
		&\leq  (3\eps+3\delta) \cdot \deg{v} + \card{\low{v}} - \card{\low{v} \cap N(u)} \tag{by~\Cref{p:4}} \\
		&\leq  (3\eps+3\delta) \cdot \deg{v} + \deg{v}  - (1-7\eps-7\delta) \cdot \deg{v} \tag{by Rule (2)} \\	
		&= (10\eps+10\delta) \cdot \deg{v},
	\end{align*}
	as desired. 
\end{proof}

\wrap{p:5}{
	Every dense vertex $v$ belongs to its candidate set $C_v$. 
}
\begin{proof}
	Since $v$ is a not a light vertex, $\low{v}$ has size at least $(1-\delta) \cdot \deg{v}$. Since $\low{v}$ is a subset of $N(v)$, by Rule (1), vertex $v$ should be included in $C_v$. 
\end{proof}

\wrap{p:7}{
	If $C_u \cap C_v \neq \emptyset$ and $\deg{u} \leq \deg{v}$, then $C_u \subseteq C_v$. 
}
\begin{proof}
	By~\Cref{p:6}, we know that $u \in C_v$. Now, consider any vertex $w \in C_u$. By Rule (2), we have, 
	\begin{align}
	\deg{w} \leq (1+2\eps+2\delta) \cdot \deg{u} \leq (1+2\eps+2\delta) \cdot \deg{v}. \label{eq:w-degree-1}
	\end{align}
	Given this degree bound, it remains to prove that $C_w \cap C_v \neq \emptyset$, so that we can apply~\Cref{p:6} and have $w$ also belongs to $S_v$. Given that $w$ can be any arbitrary vertex in $C_u$, this will imply that $C_u \subseteq C_v$. 
	
	Let us now prove that $C_w \cap C_v$ is non-empty. Firstly, since $u \in C_v$, by Rule (2), we have that 
	\begin{align}
		\card{N(u) - \low{v}} = \deg{u} - \card{N(u) \cap \low{v}} \leq \deg{u} - (1-7\eps-7\delta) \cdot \deg{v}. \label{eq:p7-1}
	\end{align}
	Similarly, since $w \in S_u$ and $\low{u} \subseteq N(u)$, by Rule (2), we also have that, 
	\begin{align}
		\card{N(w) \cap N(u)} \geq \card{N(w) \cap \low{u}} \geq (1-7\eps-7\delta) \cdot \deg{u}. \label{eq:p7-2}
	\end{align}
	Combining these, we have that
	\begin{align*}
		\card{N(w) \cap \low{v}} &\geq \card{N(w) \cap N(u)} - \card{N(u) - \low{v}} \\
		&\geq (1-7\eps-7\delta) \cdot \deg{u} - \deg{u} + (1-7\eps-7\delta) \cdot \deg{v} \tag{by~\Cref{eq:p7-2} for the first term and~\Cref{eq:p7-1} for the second} \\
		&\geq -(14\eps+14\delta) \cdot \deg{v} + (1-7\eps-7\delta) \cdot \deg{v} \tag{by Rule (2), $\deg{u} = \card{N(u)} \geq \card{N(u) \cap \low{v}} \geq (1-7\eps-7\delta) \cdot \deg{v}$}  \\
		&= (1-21\eps-21\delta) \cdot \deg{v}. 
	\end{align*}
	Using the above equation, we have 
	\begin{align*}
		\card{N(w) \cap C_v} &\geq \card{N(w) \cap \low{v}} - \card{\low{v} - S_v} \\
		&\geq (1-21\eps-21\delta) \cdot \deg{v} - 2\delta \cdot \deg{v} \tag{by~\Cref{p:1}, $\card{\low{v} \cap C_v} \geq (1-2\delta) \cdot \deg{v}$}  \\
		&= (1-21\eps-23\delta) \cdot \deg{v}. 
	\end{align*}
	Finally, the above equation in turn implies that
	\begin{align*}
		\card{C_w \cap C_v} &\geq \card{N(w) \cap C_v} - \card{N(w) - C_w} \\
		&\geq (1-21\eps-23\delta) \cdot \deg{v} - \card{N(w) - C_w} \tag{by the equation above} \\
		&\geq (1-21\eps-23\delta) \cdot \deg{v} - 2\delta \cdot \deg{w} \tag{by~\Cref{p:1}, $\card{N(w) \cap C_w} \geq (1-2\delta) \cdot \deg{w}$} \\
		&\geq (1-21\eps-23\delta) \cdot \deg{v} - 2\delta \cdot (1+2\eps+2\delta) \cdot \deg{v} \tag{by~\Cref{eq:w-degree-1}} \\
		&\geq (1-21\eps-25\delta) \cdot \deg{v} \\
		&> 0,
	\end{align*}
	as long as $\eps,\delta < 1/46$. As such $C_w \cap C_v$ is non-empty, which concludes the proof as argued earlier. 
\end{proof}

\wrap{clm:should-in}{
Any vertex $u$ in $\Low{v}{\eps} - \isolated{v}$ \underline{will not be} included in $I(v)$ with high probability.
}
\begin{proof}
Since $u$ is not $\eps$-isolated for $v$ and $v$ is not $(\eps,\eps)$-light, by~\Cref{def:isolated,def:light}, we know that 
\begin{align}
 \card{\Low{v}{\eps} \cap N({u})} =  \card{\low{v}} - \card{\low{v} - N({u})} \geq (1-2\eps) \cdot \deg{v}. \label{eq:should-in-1}
\end{align}
For any vertex $w \in \NS{u}$, let $X_w \in \set{0,1}$  be an indicator random variable for the event  `$w \in \low{v}$'. Define $X := \sum_{w \in \NS{u}} X_w$. We have, 
\begin{align*}
	\Ex\card{\low{v} \cap \NS{u}} &= \expect{X} = \sum_{w \in \NS{u}} \Ex{[X_w]} \\
	&\geq (1-2\eps) \cdot t \cdot \frac{\deg{v}}{\deg{u}} \tag{by the random choice of $\NS{u}$ and~\Cref{eq:should-in-1}} \\
	&\geq (1-3\eps) \cdot t \tag{as $u \in \low{v}$ and thus $\deg{u} \leq (1+\eps) \cdot \deg{v}$}. 
\end{align*}
 As such, by Chernoff bound (\Cref{prop:chernoff}),  we have 
 \begin{align*}
 \Pr\paren{\card{\low{v} \cap \NS{u}} < (1-4\eps) \cdot t} &\leq \Pr\paren{\card{X - \expect{X}} > \eps \cdot t} \\
 &\leq \exp\paren{-\frac{2\eps^2 \cdot t^2}{t}} \\
 &\leq \exp\paren{-2\eps^2 \cdot c \cdot \eps^{-2} \cdot \log{n}} = n^{-2c}. 
 \end{align*}
  Given  $\low{v} \subseteq \Low{v}{7\eps}$, and $\deg{u} \geq (1-2\eps) \cdot \deg{v}$ (\Cref{eq:should-in-1}),
 $u$ will not be included in $I(v)$. 
\end{proof}

\wrap{clm:should-out}{
Any vertex $u$ in $\Low{v}{7\eps} \cap \Isolated{v}{7\eps}$ \underline{will be} included in $I(v)$ with high probability.
}
\begin{proof}
Without loss of generality, in the following, we can assume $\deg{u} \geq (1-2\eps) \cdot \deg{v}$, as otherwise, $u$ will be included in $I(v)$ just because of the first condition of the test. 

Since $u$ is  $(7\eps)$-isolated, by~\Cref{def:isolated}, we know that 
\begin{align}
\card{\Low{v}{7\eps} \cap N({u})} \leq	 \deg{v} - \card{\Low{v}{7\eps} - N({u})} \leq (1-7\eps) \cdot \deg{v}. \label{eq:should-out-1}
\end{align}
For any vertex $w \in \NS{u}$, let $X_w \in \set{0,1}$  be an indicator random variable for the event  `$w \in \Low{v}{7\eps}$'. Define $X := \sum_{w \in \NS{u}} X_w$. We have, 
\begin{align*}
	\Ex\card{\Low{v}{7\eps} \cap \NS{u}} &= \expect{X} = \sum_{w \in \NS{u}} \Ex{[X_w]} \\
	&\leq (1-7\eps) \cdot t \cdot \frac{\deg{v}}{\deg{u}} \tag{by the random choice of $\NS{u}$ and~\Cref{eq:should-out-1}} \\
	&\geq (1-5\eps) \cdot t \tag{by the assumption that $\deg{u} \geq (1-2\eps) \cdot \deg{v}$ and $\eps < 1/2$}. 
\end{align*}
 As such, by Chernoff bound (\Cref{prop:chernoff}),  we have 
 \begin{align*}
 \Pr\paren{\card{\Low{v}{7\eps} \cap \NS{u}} \geq (1-4\eps) \cdot t} &\leq \Pr\paren{\card{X - \expect{X}} > \eps \cdot t} \\
 &\leq \exp\paren{-\frac{2\eps^2 \cdot t^2}{t}} \\
 &\leq \exp\paren{-2\eps^2 \cdot c \cdot \eps^{-2} \cdot \log{n}} = n^{-2c}. 
 \end{align*}
Thus, $u$ will be included in $I(v)$.
\end{proof}

\wrap{lem:form-Cv}{
	Let $v$ be any $(\eps',\delta')$-dense vertex in $V$. Then, with high probability, 
	\begin{enumerate}[label=$(\roman*)$]
		\item Every vertex $u \in V$ satisfying the following  is included in $\tC{v}$: 
		\[
			\card{N({u}) \cap \Low{v}{\eps'}} \geq (1-6\eps'-6\delta' ) \cdot \deg{v} \quad \text{and} \quad \deg{u} \leq (1+2\eps'+2\delta') \cdot \deg{v};
		\]
		\item No vertex $u \in V$ satisfying the following  is included in $\tC{v}$:
		\[
			\card{N({u}) \cap \Low{v}{\eps'}} < (1-7\eps'-7\delta' ) \cdot \deg{v} \quad \text{and} \quad \deg{u} > (1+2\eps'+2\delta') \cdot \deg{v}.
		\]
	\end{enumerate}
	Thus, $\tC{v}$ is a valid choice of $(\eps',\delta')$-candidate set $C_v$ by~\Cref{def:candidate-sets}.  
}
\begin{proof}
	Given the degrees of vertices is computed accurately by the algorithm, the conditions on the degrees are certainly satisfied in this lemma. We thus focus on the first conditions in each part. 
	
	Fix any vertex $u \in V$ and for any vertex $w \in \NS{u}$, define $X_w \in \set{0,1}$ as an indicator random variable for the event `$w \in \Low{v}{\eps'}$'. Define $X := \sum_{w \in \NS{u}} X_w$. 
	We have, 
	\begin{align}
		\Ex\card{\NS{u} \cap \Low{v}{\eps'}} = \expect{X} = \sum_{w \in \NS{u}} \expect{X_w} = t \cdot \frac{\card{N(u) \cap \low{v}}}{\deg{u}}. \label{eq:expectation-Cv}
	\end{align}
	
	Consider a vertex $u$ that should be included in part $(i)$. By~\Cref{eq:expectation-Cv}, we have,
	\[
		\Ex\card{\NS{u} \cap \Low{v}{\eps'}}  \geq (1-6\eps'-6\delta' ) \cdot t \cdot \frac{\deg{v}}{\deg{u}}
	\]
	As such, by Chernoff bound (\Cref{prop:chernoff}), 
	\begin{align*}
		\Pr\paren{\card{\NS{u} \cap \Low{v}{\eps'}} < (1-6\eps'-6\delta' - \eps) \cdot t \cdot \frac{\deg{v}}{\deg{u}}} &\leq \Pr\paren{\card{X-\expect{X}} > \eps \cdot t \cdot \frac{\deg{v}}{\deg{u}}} \\
		&\leq \exp\paren{-\frac{2t^2 \cdot \deg{v}}{t \cdot \deg{u}}} \\
		&\leq \exp\paren{-\eps^2 \cdot t} \tag{by the bound of $\deg{u} \leq (1+2\eps'+2\delta') \cdot \deg{v}$} \\
		&= \exp\paren{-\eps^2 \cdot c \cdot \eps^{-2} \cdot \log{n}} < n^{-c}. 
	\end{align*}
	As such, with high probability, all the vertices in part $(i)$ will be included in $\tC{v}$. 
	
	Now consider a vertex $u$ that should be not included in part $(ii)$. By~\Cref{eq:expectation-Cv}, we have,
	\[
		\Ex\card{\NS{u} \cap \Low{v}{\eps'}} \leq (1-7\eps'-7\delta' ) \cdot t \cdot \frac{\deg{v}}{\deg{u}}
	\]
	As such, by Chernoff bound (\Cref{prop:chernoff}), 
	\begin{align*}
		\Pr\paren{\card{\NS{u} \cap \Low{v}{\eps'}} \geq (1-6\eps'-6\delta' - \eps) \cdot t \cdot \frac{\deg{v}}{\deg{u}}} &< \Pr\paren{\card{X-\expect{X}} > \eps \cdot t \cdot \frac{\deg{v}}{\deg{u}}} \\
		&\leq \exp\paren{-\frac{2t^2 \cdot \deg{v}}{t \cdot \deg{u}}} \\
		&\leq \exp\paren{-\eps^2 \cdot t} \tag{by the bound of $\deg{u} \leq (1+2\eps'+2\delta') \cdot \deg{v}$} \\
		&= \exp\paren{-\eps^2 \cdot c \cdot \eps^{-2} \cdot \log{n}} < n^{-c}. 
	\end{align*}
	As such, with high probability, not vertex of part $(ii)$ will be included in $\tC{v}$. 
\end{proof}

\wrap{lem:dense-finding-1}{
	Suppose $v$ is an $(\eps'',\delta'')$-dense vertex. Then, $v \in C_u$ for every $u \in \Kernel{v}{\eps''}{\delta''}$ (where $C_v$ is also computed as a $(\eps',\delta')$-candidate set of $u$ by the algorithm in the previous part). 
}
\begin{proof}
	$\Kernel{v}{\eps''}{\delta''} \subseteq \Kernel{v}{\eps'}{\delta'}$ by~\Cref{p:kernel-monotone} since $\eps' > \eps'' + \delta''$, and $\Kernel{v}{\eps'}{\delta'} \subseteq C_v$ by~\Cref{p:1}. 
	At the same time, for any vertex $u \in \Kernel{v}{\eps''}{\delta''}$, we also have $u \in C_u$ also by~\Cref{p:5}. Thus, we have that $C_u \cap C_v \neq \emptyset$. We will now prove that $\deg{v} \leq (1+2\eps'+2\delta') \cdot \deg{u}$ 
	so that we can apply~\Cref{p:6} and get that $v \in C_u$. 
	
	By~\Cref{p:kernel-dense}, we have each $u \in \Kernel{v}{\eps''}{\delta''}$ has 
	\[
		\deg{u} \geq (1-\eps''-\delta'') \cdot \deg{v} \geq \frac{1}{(1+2\eps''+2\delta'')} \cdot \deg{v},
	\]
	as desired. 
\end{proof}

\wrap{lem:dense-finding-2}{
	With high probability, for every $(\eps'',\delta'')$-dense vertex $v$, there is at least one vertex $u \in \Kernel{v}{\eps''}{\delta''}$ that is sampled in $\sample$. 
}
\begin{proof}
	Size of $\Kernel{v}{\eps''}{\delta''}$ is at least $\deg{v}/2$ by~\Cref{p:kernel} and each vertex $u \in \Kernel{v}{\eps''}{\delta''}$ has degree at most $2 \cdot \deg{v}$ by definition. 
	Since we are sampling each vertex $u \in \Kernel{v}{\eps''}{\delta''}$ in $\sample$ independently with probability $p_u$, we have, 
	\begin{align*}
		\Pr\paren{\text{no vertex of $\Kernel{v}{\eps''}{\delta''}$ is sampled}} &\leq \prod_{u \in \Kernel{v}{\eps''}{\delta''}} (1-p_u) \\
		&\leq (1-p_u)^{\deg{v}/2} \leq \paren{1-\frac{c \cdot \log{n}}{2\,\deg{v}}}^{\deg{v}/2} \\
		&\leq \exp\paren{-\frac{c}{4} \cdot \log{n}} < n^{-c/4}.
	\end{align*}
	This concludes the proof. 
\end{proof}


\newcommand{\ORn}{\ensuremath{\textnormal{OR}_{N}}}
\newcommand{\istar}{\ensuremath{i^{*}}}
\newcommand{\tseq}[1]{\ensuremath{\langle #1\rangle \xspace}}

\section{Standard Variants of Sublinear-Time and Streaming Algorithms}
\label{app:models}
In this section, we provide the following complementary results beyond the main models we used for the algorithms.
\begin{itemize}
\item Impossibility results for sublinear-time algorithms such that neither the query access to a graph $G$ nor the adjacency list of the $(-)$-edges $G^{-}$ is sufficient for \emph{any} multiplicative approximation algorithm for correlation clustering with $o(n^{2})$ time. 
\item Additional arithmetic results to show that a variation of our sublinear-space algorithm also works under the dynamic graph streams. 
\end{itemize} 

Note that under the adjacency list model, any query lower bound automatically implies a time lower bound since each query takes $O(1)$ time. Therefore, we prove \emph{query} lower bounds for the impossibility results on sublinear-time algorithms.

\subsection{Impossibility Results for Sublinear-time Algorithms on Models}
\label{app:time-lower}
We first show two negative results under query models other than the adjacency list of $(+)$-graph we adopted for our sublinear-time algorithm. These results indicate that the choice of our model with the $(+)$-graph is a natural one for correlation clustering.

\subsubsection*{Lower bound for correlation clustering algorithms with adjacency list access of $G$}
Our first negative result shows that it is impossible to get \emph{any} multiplicative approximation for correlation clustering with $o(n^{2})$ queries with \emph{only} the query access of the adjacency list of $G$. Formally, we have the following result
\begin{proposition}
\label{prop:lb-adjacency-G}
Suppose an algorithm is given an input labeled graph $G$ specified via the adjacency list of the graph (but \emph{not} the $(+)$-subgraph). Then, any algorithm that finds an $\alpha$-approximation to the correlation clustering problem on $G$ for any finite $\alpha\geq 1$ with probability at least $\frac{99}{100}$ requires $\Omega\paren{n^2}$ queries.
\end{proposition}
\begin{proof}
We prove \Cref{prop:lb-adjacency-G} by a reduction from the simple $\ORn$ problem. 
\begin{problem}[$\ORn$]
Given $N$ Boolean variables $\{x_{i}\}_{i=1}^{N}$ such that $(N-1)$ of them are $0$ and one of them is either $0$ or $1$, output $f(\{x_{i}\}_{i=1}^{N}):=x_{1} \vee x_{2} \cdots \vee x_{N}$.
\end{problem}
It is well known that any algorithm that solves $\ORn$ with probability at least $\frac{49}{50}$ requires $\Omega\paren{N}$ queries \cite{BuhrmanW02}. We now show that this implies the lower bound in \Cref{prop:lb-adjacency-G}. To see this, consider the following graph $G$
\begin{tbox}
A family of instances to prove the lower bound in \Cref{prop:lb-adjacency-G}.

\smallskip
\begin{enumerate}
\item Order the pairs of vertices arbitrarily among $[{n \choose 2}]$. Pick a special index $\istar \in {n \choose 2}$. 
\item For all other indices \emph{except} $\istar$, add $(-)$ edges between the corresponding vertex pairs.
\item For the index $\istar$, add \emph{either} a $(+)$ or a $(-)$ edge between the corresponding vertex pair.
\end{enumerate}
\end{tbox}

We now show that if we can find an $\alpha$-approximation to the correlation clustering on the above instances probability at least $99/100$ in $o(n^2)$ queries, it would mean an algorithm to solve $\ORn$ in $o(N)$ queries with probability at least $49/50$, which forms a contradiction. Observe that 
\begin{itemize}
\item If the edge corresponds to $\istar$ is positive, the optimal clustering is to put the vertex pair corresponds to $\istar$ in a cluster, and all other vertices in other clusters;
\item On the other hand, if the edge corresponds to $\istar$ is negative, the optimal clustering is to put all the vertices to separate clusters.
\end{itemize}
Note that in both cases, the optimal clustering cost is $0$. Therefore, any algorithm that provides $\alpha$-multiplicative approximation for some finite $\alpha$ must \emph{recover the optimal clustering}. 

Now suppose such an approximation algorithm $\textsf{ALG}$ for correlation clustering exists and the success probability is at least $\frac{99}{100}$ and $o(n^2)$ queries. For any given $\ORn$ instance, one can use $\textsf{ALG}$ in the following way: construct a graph with $n$ vertices where $N = {n \choose 2}$, and treat each query on the $i$-th edge slot as the query of the $i$-th element $x_{i}$. Return queries as $(-)$ edge if $x_{i}=0$, and as $(+)$ edge if $x_{i}=1$. In this way, we can run $\textsf{ALG}$, and output the results of $\ORn$ based on the resulting correlation clustering. Specifically, we can output $f(\{x_{i}\}_{i=1}^{N})=0$ if the total number of clusters is $n$, and output $f(\{x_{i}\}_{i=1}^{N})=1$ if the total number of clusters is $n-1$. By the guarantee of $\textsf{ALG}$, this algorithm solves $\ORn$ with probability at least $\frac{99}{100}>\frac{49}{50}$ with $o(n^2)=o(N)$ queries. This is a contradiction. Therefore, any such correlation clustering algorithm must use at least $\Omega\paren{n^2}$ queries. \Qed{\Cref{prop:lb-adjacency-G}}

\end{proof}

We remark that our lower bound in this section is similar to the one proved by Bonchi et al. \cite{BonchiGK13}, but the subtle difference makes the two lower bounds incomparable. The lower bound of \cite{BonchiGK13} rules out any algorithm with $o(n^2)$ queries to achieve \emph{constant}  multiplicative approximation (even) with an extra \emph{constant} additive error. In the \emph{additive} sense, their lower bound is stronger than ours. However, our lower bound rules out any algorithm with $o(n^2)$ queries to achieve \emph{any} multiplicative approximation, even with polynomial or exponential factors. Therefore, in the \emph{multiplicative} sense, our lower bound is stronger.

\subsubsection*{Lower bound for correlation clustering algorithms with adjacency list of $G^{-}$}
We now turn to the stronger negative result, which shows that even if we are given the access of the adjacency list of the $(-)$-subgraph $G^{-}$, any algorithm to provide multiplicative approximation to correlation clustering with high probability still needs $\Omega\paren{n^2}$ queries. Formally, there is 
\begin{proposition}
\label{prop:lb-adjacency-G-minus}
Suppose an algorithm is given an input labeled graph $G$ specified via the adjacency list of the $(-)$-subgraph $G^-$. Then, any algorithm that finds an $\alpha$-approximation to the correlation clustering problem on $G$ for any finite $\alpha\geq 1$ with probability at least $\frac{99}{100}$ requires $\Omega\paren{n^2}$ queries.
\end{proposition}
\begin{proof}
We prove the statement by a reduction from the recovery of a perfect matching. Formally, the problem is defined as follows.
\begin{problem}[Matching Recovery]
\label{prb:matching-recover}
Suppose there are two players, Alice and Bob. Alice picks a random perfect matching among $n$ vertices, and mark the matching edges as \emph{special}. Furthermore, Alice adds other edges to make the graph complete. Bob is given the adjacency list of the the complete graph constructed by Alice, and Bob tries to find all the \emph{special} edges.
\end{problem}
By an argument from \cite{CzumajS10}, it can be shown that solving \Cref{prb:matching-recover} with probability at least $\frac{49}{50}$ requires $\Omega\paren{n^2}$ queries from Bob. Now, suppose we have a correlation clustering algorithm $\textsf{ALG}$ that returns an $\alpha$-approximation to correlation clustering for some finite $\alpha$. We use it to solve \Cref{prb:matching-recover}. Specifically, we can arrange the labels of the edges and use $\textsf{ALG}$ as follows.

\begin{tbox}
Solving \Cref{prb:matching-recover} with correlation clustering algorithm $\textsf{ALG}$.

\smallskip
\begin{enumerate}
\item Alice picks a random perfect matching as the \emph{special} edges, and label the \emph{special} edges as $(+)$. Furthermore, Alice add $(-)$-edges between all other vertex pairs.
\item Let $G$ denote the graph constructed by Alice, and give the adjacency list of $G$ to Bob.
\item Bob calls $\textsf{ALG}$, and simulate adjacency list of $G^-$ by answering $\textsf{ALG}$ in the following way:
\begin{itemize}
\item Degree queries: simply return $(n-1)$.
\item Pair queries: return `$(-)$ edge' if there is a $(-)$ edge; otherwise, return `no $(-)$ edge'.
\item Neighbor queries:
\begin{enumerate}
\item If the neighbor is connected by a $(-)$ edge, return the neighbor.
\item If the neighbor is connected by a $(+)$ edge, skip the neighbor and return the next vertex instead; for all the later queries, return the neighbor with the index increased by $1$.
\end{enumerate}
\end{itemize}
\end{enumerate}
\end{tbox}

Observe that the optimal solution of the correlation clustering in a graph $G$ constructed as above is to put the vertices that are connected by the $(+)$ edges in separate clusters. In this way, the optimal cost of correlation clustering on $G$ is $0$; and again, any algorithm that provides an $\alpha$-approximation to correlation clustering for some finite $\alpha$ must \emph{recover this clustering}. Therefore, if $\textsf{ALG}$ succeeds, which is with probability at least $\frac{99}{100}$, Bob can recover the special matching edges.

We now show that the query complexity in the above procedure is $o(n^2)$. Note that for the degree queries, pair queries, and neighbors queries for $(-)$ edges, the number of total queries in $o(n^2)$ by the guarantee of $\textsf{ALG}$. On the other hand, if the neighbor query includes a $(+)$ edge, we pay a \emph{single} extra query. This can happen at most $O(n)=o(n^2)$ times, which means the total query complexity is $o(n^2)$.

The above reduction gives an algorithm that can solve \Cref{prb:matching-recover} with probability $\frac{99}{100}$ and $o(n^2)$ queries from Bob, which forms a contradiction. Therefore, any such correlation clustering algorithm must use $\Omega\paren{n^2}$ queries. \Qed{\Cref{prop:lb-adjacency-G-minus}}

\end{proof}

\subsection{Algorithms for Other Streaming Models}
\label{app:stream-upper}
On the side of sublinear space algorithms, we have additional positive results on other models. We first note that an immediate observation is the streaming algorithm works with access to $G^{+}$:
\begin{proposition}
\label{prop:ub-streaming-G-plus}
Suppose a labeled graph $G=(V,\, E)$ is specified via a stream of edges of the positive subgraph $G^+$ of $G$. Then, \Cref{alg:cc-sub-space} with high probability finds an $O(1)$-approximation to the correlation clustering problem on $G$ in $O\paren{n\log(n)}$ space. 
\end{proposition}
\begin{proof}
The proof is immediate as \Cref{alg:cc-sub-space} only utilizes the positive edges to perform clustering. \Qed{\Cref{prop:ub-streaming-G-plus}}

\end{proof}
Next, we show that a variate of our streaming algorithm also works under a \emph{dynamic} stream with only (multiplicative) poly-logarithm space overhead. Since correlation clustering deals with labeled graphs, we first extend the definition of the dynamic streams to graphs as such.

\paragraph{The dynamic stream of labeled graphs.} For a labeled graph $G$, its dynamic stream consists a length-$T$ sequence of tuples $\tseq{\sigma_{1}, \sigma_{2},\cdots,\sigma_{T}}$, where each $\sigma_{i}$ is consist of
\[\sigma_{i} = (u_{i}, v_{i}, \Delta_{i}),\]
such that $u_{i}$ and $v_{i}$ are a pair of vertices, and $\Delta_{i}$ is the update of the label. There are only four types of updates allowed: insertion of $(+)$, removal of $(+)$, insertion of $(-)$, and removal of $(-)$. Furthermore, for the vertex pair $(u,v)$, any removal of the label can only appear \emph{after} the insertion, and any insertion cannot happen if there is already a labeled edge. Finally, by the end of the stream, every edge should have a label in $\{-1, +1\}$. 

We remark that the model is a natural extension of the general dynamic graph streams. The model essentially assigns two types of edge weights (`$+$' and `$-$'), and ensures that the change of types must be followed by the removal of the other type. To see intuitively why there exists an algorithm to handle dynamic streams as such, note that all the operations on the graph in \Cref{alg:cc-sub-space} (Lines~\ref{line:stream-degree} to \ref{line:stream-sample}) are based on \emph{sampling} from the graph. Therefore, it is possible to use the celebrated $\ell_{0}$ sampler to sample edges in a dynamic stream in the same manner. 

We now formally give the statement of the algorithm.

\begin{proposition}
\label{prop:ub-dynamic-streaming}
Suppose a labeled graph $G=(V,\, E)$ is specified via a \emph{dynamic stream} of edges of $G$. Then, there exists a randomized algorithm that with high probability finds an $O(1)$-approximation to the correlation clustering problem on $G$ in $\tilde{O}(n)$ space. 
\end{proposition}

Before diving into the formal proof, we first introduce the standard tools that will be used. In a dynamic stream, we use the powerful $\ell_{0}$ sampler to sample edges, specified as below:
\begin{proposition}[$\ell_{0}$ sampler]
\label{fact:l0-sampler}
For any graph $G=(V,E)$, suppose the value $\paren{\sum_{i:(u_{i},v_{i})=e}\Delta_{i}}$ for every edge $e$ is $\poly(n)$ bounded. 
Then, there exists a randomized algorithm such that given access to a dynamic stream, returns an edge $e$ along with value $\paren{\sum_{i:(u_{i},v_{i})=e}\Delta_{i}}$ in $O\paren{\log^{2}(n)\log(\frac{1}{\delta})}$ space with probability at least $(1-\delta)$, where $e$ is uniformly distributed among the edges with non-zero values by the end of the stream. 
\end{proposition}
In the following, we refer the algorithm in \Cref{fact:l0-sampler} as a $\ell_{0}$ sampler. Note that the $\ell_{0}$ sampler works for \emph{any} graph, and it is possible to sample from a subgraph $G'\subseteq G$ (e.g. the neighboring edges of a certain vertex). With this standard tool, we show in the following lemma that it is possible to simulate the sampling of vertices and all the adjacent edges in the same manner of \Cref{lem:vertex-sample}. 
\begin{lemma}
\label{lem:dynamic-vertex-sample}
There exists a streaming algorithm that given any arbitrary graph $G=(V,E)$ (not necessarily a labeled graph) specified via any dynamic stream of its edges, by the end of the stream, with high probability outputs a collection $S$ of vertices together with $N(v)$ (the final edges indent to $v$) for all $v \in S$, such that each vertex is sampled independently and with probability $\min\set{(\beta \cdot \log{n})/\deg{v},1}$ in $V$ for some constant $\beta$. Furthermore, the space complexity of the algorithm is $\tilde{O}\paren{n}$. 
\end{lemma}
\begin{proof}
The algorithm is as simple as follows.
\begin{tbox}
Sampling algorithm of~\Cref{lem:dynamic-vertex-sample}. 

\begin{enumerate}[label=$(\roman*)$]
\item Maintain $\log(n)$ buckets $\{B_{i}\}_{i=1}^{\log(n)}$, each bucket is with size $|B_{i}|:=\min\{n, 2\beta\cdot \frac{n\log(n)}{2^{i-1}}\}$. 
\item For each $B_{i}$, sample vertex uniformly at random with replacement, and put the sampled vertices to the bucket until it is full (also with replacement across buckets). 
\item For each vertex $v$ in bucket $i$, maintain $ \paren{100\cdot 2^{i}\cdot \log(n)}$ $\ell_{0}$ samplers to sample its neighbors.
\end{enumerate}
\end{tbox}
We now analyze the correctness and the space complexity.

\paragraph{Correctness.} For any vertex $v\in V$, there exists an index $i'$ such that $\deg{v}\in [2^{i'-1}, 2^{i'}]$. Fix this interval, we show that the probability for $v$ to be sampled in $B_{i'}$ is at least $\frac{\beta\cdot\log(n)}{\deg{v}}$. To see this, note that if $|B_{i'}|=n$, $v$ will be sampled surely; on the other hand, if $|B_{i'}|<n$, the probability for $v$ not to be sampled in \emph{one} sampling of $B_{i'}$ is $(1-\frac{1}{n})$. Therefore, the probability for $v$ to be sampled at least once can be calculated as
\begin{align*}
\Pr\paren{\text{$v$ is sampled at least once}} &= 1-\paren{1-\frac{1}{n}}^{|B_{i'}|}\\
&\geq 1-\exp\paren{-\frac{2\beta\log(n)}{2^{i'-1}}} \tag{$1-x\leq \exp(-x)$}\\
&\geq 1-\paren{1-\frac{\beta\log(n)}{2^{i'-1}}} \tag{$\exp(-x)\leq 1-\frac{x}{2}$ for $x\in [0,1]$}\\
&= \frac{\beta\log(n)}{2^{i'-1}} \geq \frac{\beta\log(n)}{\deg{v}}. \tag{$\deg{v}\geq 2^{i'-1}$}
\end{align*}
Conditioning on vertex $v$ is sampled in bucket $i'$, we show that all the edges indent to $v$ can be sampled with high probability. To see this, note that for one edge $e$ that is indent to $v$, the probability for it not to be sampled by one $\ell_{0}$ sampler is $(1-\frac{1}{\deg{v}})$. Therefore, the probability for it not to be sampled by any $\ell_{0}$ sampler is at most 
\[(1-\frac{1}{\deg{v}})^{100\cdot 2^{i'}\cdot \log(n)}\leq \frac{1}{n^{10}},\] 
where the inequality is obtained by using $(1-x)\leq \exp(-x)$ and $\deg{v}\leq 2^{i'}$. Hence, we can apply a union bound and conclude the correctness for \emph{all} vertices and \emph{all} the adjacent edges.
 
\paragraph{Space complexity.} 
There are two sources of space complexity: the vertices we store and the $\ell_{0}$ samplers. For the vertex we store, the total number is less than $2\beta\cdot \sum_{i=1}^{\infty} \frac{n\log(n)}{2^{i-1}} = O(n\log(n))$. The number of $\ell_{0}$ samplers can be divided into two parts. For the buckets with $i$ such that $|B_{i}|\geq n$ (`low-index buckets'), we have $2^{i-1}\leq 2\beta\log(n)$. Therefore, the number of $\ell_{0}$ samplers one can maintain for each low-index bucket is at most $O(n\cdot 2^{i}\cdot \log(n)) = O(n\cdot \log^2(n))$. Furthermore, for the buckets with size such that $|B_{i}|< n$ (`high-index buckets'), the number of vertices to be stored is at most $2\beta\cdot \frac{n\log(n)}{2^{i-1}}$. Therefore, the number of $\ell_{0}$ samplers one can maintain for each high-index bucket is at most $O(\frac{n\log(n)}{2^{i-1}} \cdot 2^{i}\cdot \log(n)) = O(n\cdot \log^2(n))$. Each vertex takes only $O(1)$ words to store, and each $\ell_{0}$ sampler can be implemented with $O(\log^3(n))$ space to ensure high probability. Therefore, the total space cost is $\tilde{O}(n)$. \Qed{\Cref{lem:dynamic-vertex-sample}} 

\end{proof}

We are now ready to prove \Cref{prop:ub-dynamic-streaming} with the machinery we developed above.

\begin{proof}[Proof of \Cref{prop:ub-dynamic-streaming}]
By some modification of \Cref{alg:cc-sub-space}, the algorithm is as follows.
\begin{Algorithm}\label{alg:cc-dym-stream}
{A single-pass semi-streaming algorithm for correlation clustering in dynamic streams.} 

\begin{itemize}
\item \textbf{Input:} A labeled graph $G=(V,E)$ specified via a dynamic stream of the edges. 
\end{itemize}

\begin{enumerate}[label=$(\roman*)$]

\item Let $\eps > 0$ be a sufficiently small \underline{constant} as prescribed by~\Cref{thm:decomposition}. 

\item \underline{Pre-processing}: For each every tuple $\sigma_{i} = (u_{i}, v_{i}, \Delta_{i})$, \emph{ignore} $\sigma_{i}$ if $\Delta_{i}$ is an update for $(-)$ edges.

\item \label{line:dynamic-stream-degree} For each vertex $v\in V$, use a counter over edges of $E^+(v)$ to maintain $\degp{v}$. 

\item \label{line:dynamic-stream-neighbor} For each vertex $v$, maintain $\paren{t = \frac{c \cdot \log{n}}{\eps^2}}$ $\ell_{0}$ samplers to sample neighbors of $v$ from $N^+(v)$ (with repetition) to get $\NS{v}$ . 

\item \label{line:dynamic-stream-sample} Run the algorithm of~\Cref{lem:dynamic-vertex-sample} on the dynamic stream with parameter $\beta = c$. Let $\sample$ be the final set of vertices maintained by the algorithm. 

\item \label{line:dynamic-decompose-query} Run the algorithm of~\Cref{thm:decomposition} for sparse-dense decomposition with parameter $\eps$ and the inputs $\{\NS{v}\}_{v \in V}$ and $\set{N^+(v)}_{v \in \sample}$ to its recovery algorithm. 

\item \label{line:dynamic-clustering-query} Output clustering $\ALG$ based on the resulting $\Vsparse \sqcup K_1 \sqcup \ldots \sqcup K_k$ as prescribed in \Cref{thm:cc-alg}. 

\end{enumerate}
\end{Algorithm}
We now analyze the correctness and the space complexity of the algorithm.

\paragraph{Correctness.} Note that by the \underline{pre-processing} step, the sampling only happens for the $G^+$ subgraph. 
Therefore, by setting $\delta=1/\poly(n)$, with high probability, Line~\ref{line:dynamic-stream-neighbor} returns the exact random edge samples as prescribed by \Cref{thm:decomposition}. Furthermore, by the guarantee of \Cref{lem:dynamic-vertex-sample}, Lines~\ref{line:dynamic-stream-degree} to \ref{line:dynamic-stream-sample} give the exact information for the algorithm of \Cref{thm:decomposition}. Therefore, we can get a valid sparse-dense decomposition of the positive subgraph $G^+$ with high probability, and returns an $O(1)$-approximation of correlation clustering as shown in \Cref{thm:cc-alg}.

\paragraph{Space complexity.} Line~\ref{line:dynamic-stream-neighbor} requires storing $O(n)$ numbers. Line~\ref{line:dynamic-stream-sample} requires using $\paren{n\cdot t}$ $\ell_{0}$ samplers, and each of them takes $O\paren{\polylog(n)}$ space (setting $\delta = 1/\poly(n)$). This leads to a total $\tilde{O}(n)$ space for Line~\ref{line:dynamic-stream-sample}. Finally, by~\Cref{lem:dynamic-vertex-sample}, the space complexity of Line~\ref{line:dynamic-clustering-query} is $\tilde{O}(n)$ with high probability. Therefore, by the same argument used in \Cref{alg:cc-sub-space}, the total space complexity is $\tilde{O}(n)$. \Qed{\Cref{prop:ub-dynamic-streaming}} 

\end{proof}

There are two main implications of \Cref{prop:ub-dynamic-streaming}. Firstly, since the algorithm can deal with dynamic streams, it is possible to obtain $O(1)$-approximation correlation clustering in $\tilde{O}(n)$ space when the labels of the edges might \emph{change}. This extends the power of our algorithm beyond the memory efficiency, and paves the way for applications with \emph{temporal} relationships between data points. The second implication is on streaming correlation clustering with access of graph $G^-$. By the algorithm for dynamic stream, we can show that there exists a semi-streaming algorithm even the stream only contains the edges of $G^-$. 
\begin{proposition}
\label{prop:ub-streaming-G-minus}
Suppose a labeled graph $G=(V,\, E)$ is specified via a stream of edges of the positive subgraph $G^-$ of $G$. Then, there exists an algorithm that with high probability finds an $O(1)$-approximation to the correlation clustering problem on $G$ in $\tilde{O}\paren{n}$ space. 
\end{proposition}
\begin{proof}
The proof follows from the result of \Cref{prop:ub-dynamic-streaming}. Note that if we are given the stream of $G^-$, we can add a stream with ${n \choose 2}$ tuples that cover every edge slot and every $\Delta_{i}$ is an insertion of $(+)$. Then, we pad the stream with the edges in $G^-$ with two $\Delta$ updates for each edge: the first for removal of $(+)$, the second for insertion of $(-)$. One can then solve this dynamic stream with \Cref{alg:cc-dym-stream} in $\tilde{O}(n)$ space, which the answer is exactly the $O(1)$-correlation clustering of $G$ with high probability.  \Qed{\Cref{prop:ub-streaming-G-minus}}

\end{proof}

\end{document}